\documentclass{fundam}

\usepackage{bm}
\usepackage{enumerate}
\usepackage{enumitem}
\usepackage{bbm}
 \usepackage{amsmath,amsfonts}

\usepackage{url} 
\usepackage[ruled,lined]{algorithm2e}
\usepackage{graphicx}
\usepackage{tikz}
\usetikzlibrary{automata, positioning, arrows}
\usepackage{subcaption}

\usepackage{todonotes}

\newcommand{\Z}{\mathbb{Z}}

\newcommand{\F}{\mathbb{F}}
\newcommand{\bu}{\mathbf{u}}

\newcommand{\bt}{\mathbf{t}}
\newcommand{\bz}{\mathbf{z}}
\newcommand{\bx}{\mathbf{x}}
\newcommand{\by}{\mathbf{y}}
\newcommand{\bv}{\mathbf{v}}
\newcommand{\bw}{\mathbf{w}}
\newcommand{\bd}{\mathbf{d}}

\newcommand{\bc}{\mathbf{c}}
\newcommand{\be}{\mathbf{e}}
\newcommand{\ba}{\mathbf{a}}
\newcommand{\bb}{\mathbf{b}}
\newcommand{\bbd}{\mathbf{d}}

\newcommand{\LD}{\gamma^{LD}}
\newcommand{\ID}{\gamma^{ID}}
\newcommand{\LLD}{\gamma^{L-LD}}
\newcommand{\LID}{\gamma^{L-ID}}

\usepackage{hyperref}

\begin{document}

\setcounter{page}{351}
\publyear{24}
\papernumber{2187}
\volume{191}
\issue{3-4}

\finalVersionForARXIV


\title{Optimal Local Identifying and Local Locating-dominating Codes}

\author{Pyry Herva\thanks{Address for correspondence: Department of Mathematics and Statistics,
                                University of Turku,  Turku, FI-20014, Finland.}\thanks{Research supported
                                 by the Emil Aaltonen Foundation.}
               Tero Laihonen\thanks{Research supported by the Academy of Finland grant 338797.}\, Tuomo Lehtil\"{a}$^\ddag$
\\
Department of Mathematics and Statistics \\
University of Turku \\
Turku FI-20014, Finland \\
\{pysahe, terolai, tualeh\}@utu.fi
}

\maketitle

\runninghead{P. Herva et al.}{Optimal Local Identifying and Local Locating-dominating Codes}

\begin{flushright}
  \textit{In honour of the 60th birthday of Iiro Honkala\;}
\end{flushright}

\begin{abstract}
We introduce two new classes of covering codes in graphs for every positive integer $r$. These new codes are called local $r$-identifying and local $r$-locating-dominating codes and they are derived from $r$-identifying and $r$-locating-dominating codes, respectively.
We study the sizes of optimal local 1-identifying codes in binary hypercubes.
We obtain lower and upper bounds that are asymptotically tight.
Together the bounds show that the cost of changing covering codes into local 1-identifying codes is negligible.
For some small $n$ optimal constructions are obtained.
Moreover, the
upper bound is obtained by a linear code construction.
Also, we study the densities of optimal local 1-identifying codes and local 1-locating-dominating codes in the infinite square grid, the hexagonal grid, the triangular grid and the king grid.
We prove that seven out of eight of our constructions have optimal densities.

\end{abstract}

\section{Introduction and preliminaries}

There are three widely studied ways to locate vertices in a graph using subsets of vertices; namely, \emph{resolving sets} \cite{Harary76,Slater75}
which separate vertices
using the distances to the elements in the subset, \emph{identifying codes} \cite{Karpovsky} and
\emph{locating-dominating codes (or sets)} \cite{Slater2,Slater}  both of which separate using different neighbourhoods of the vertices in the subset.
In the case of resolving sets, the question of
separating only the adjacent vertices \cite{localmd1,localmd2}  has been extensively studied, see, for example, \cite{Klavzar} and the references therein. Such subsets are called \emph{local resolving sets}.
Inspired by this,
we study in this paper the analogous question with respect to identifying codes and locating-dominating codes.

Consequently, we introduce two new classes of codes derived from identifying and locating-dominating codes and study them in some graphs.
We concentrate on the \emph{optimal} codes.
In finite graphs by optimal we refer to the smallest possible size of the code and in infinite graphs to the smallest possible density of the code.
Since the new code classes are closely related to identifying and locating-dominating codes, some comparison is made.
Some of the results in this paper have been published in \cite{rufidim,Herva}.

\subsection*{Graphs and codes}

In this paper we consider simple, connected and undirected graphs $G=(V,E)$ with
vertex set $V$ and edge set $E \subseteq \{ \{u,v\} \mid u,v \in V, u \neq v \}$.
The graph $G$ is finite if its vertex set $V$ is a finite set and infinite if $V$ is an infinite set.
The (graphic) \emph{distance} $d(u,v)$ of two vertices $u,v \in V$ of $G$ is the number of edges in a shortest path between $u$ and $v$.
Let $r$ be a non-negative integer.
A vertex $u$ is said to \emph{$r$-cover}
a vertex $v$ (and vice versa) if $d(u,v) \leq r$.
When $r=1$, we may say just that $u$ \emph{covers} $v$.
More generally, we say that a subset of the vertex set of the graph
$r$-covers a vertex $u$ if the subset has an element which $r$-covers vertex $u$.

Any non-empty subset $C \subseteq V$ of vertices of a graph $G=(V,E)$ is called a \emph{code} (in the graph $G$).
The elements of $C$ are called \emph{codewords} and the elements of $V \setminus C$ are called \emph{non-codewords}.
A code $C \subseteq V$ is an \emph{$r$-covering code} if it $r$-covers every vertex.
If $r=1$, we may say that $C$ is simply a \emph{covering code}.
In other words, if there is a codeword of $C$ in distance at most $r$ from any vertex of $V$, then $C$ is an $r$-covering code in $G$.
Covering codes are also called \emph{dominating sets}.

\medskip
The \emph{(closed) $r$-neighbourhood} of a vertex $u \in V$ in a graph $G=(V,E)$ is the set
$N_r[u] = \{ v \in V \mid d(v,u) \leq r \}$.
The \emph{open $r$-neighbourhood} of $u$ is the set
$N_r(u) = N_r[u] \setminus \{u\}$.
The \emph{$r$-identifying set} $I_{C,r}(u)$ or the \emph{$I$-set} of a vertex $u \in V$ with respect to a code $C$ is the set
$$
I_{C,r}(u) = N_r[u] \cap C.
$$
For $r=1$ we denote $N[u] = N_1[u]$ and $I(u)=I_C(u) = I_{C,1}(u)$.
A code  \emph{$r$-separates} (or just \emph{separates} if $r=1$) two vertices $u$ and $v$ if their $I$-sets are different.
If $C$ separates $u$ and $v$, we may also say that $C$ separates $u$ from $v$ (or vice versa).
More generally, we say that $C$ separates $u$ from a set of vertices $S$ if $C$ separates $u$ from every vertex of $S$.
Note that $C$ is an $r$-covering code if and only if $I_{C,r}(u) \neq \emptyset$ for every $u \in V$.

\medskip
A code in a certain class of covering codes in a finite graph is \emph{optimal} if its size is the smallest among every code in the class.
We also consider optimality with respect to densities in certain infinite graphs (for the definitions of the densities see Section 3).

\medskip
Let us define two widely studied classes of covering codes for every positive integer that are useful in locating vertices in a graph:

\begin{definition}
	A code $C \subseteq V$ in a graph $G = (V,E)$ is an \emph{$r$-identifying code} if it is an $r$-covering code and
	$I_{C,r}(u) \neq I_{C,r}(v)$
	for every distinct $u,v \in V$.
\end{definition}	

\begin{definition}
	A code $C \subseteq V$ in a graph $G = (V,E)$ is an \emph{$r$-locating-dominating code} if it is an $r$-covering code and
	$I_{C,r}(u) \neq I_{C,r}(v)$
	for every distinct $u,v \in V \setminus C$.
\end{definition}	

\noindent	
In other words, a code $C$ is an $r$-identifying code if it $r$-covers every vertex and $r$-separates any two vertices, and it is an $r$-locating-dominating code if it $r$-covers every vertex and $r$-separates any two non-codewords.
By identifying and locating-dominating codes we mean $1$-identifying and $1$-locating-dominating codes, respectively.

The concept of identifying codes was introduced by Karpovsky, Chakrabarty and Levitin in \cite{Karpovsky} 
and the concept of locating-dominating codes was introduced by Slater and Rall in \cite{Slater2,Slater}. 
Since their discovery, these (and many related) classes of codes have been extensively studied in different graphs over the years.
See also the website \cite{Lobstein} for a comprehensive list of references around the topic.

\subsection*{Our new codes: local identifying and local locating-dominating codes}

Two distinct vertices of $G=(V,E)$ are called \emph{neighbours} or \emph{adjacent}
if there is an edge between them, that is, if their distance is $1$.

\begin{definition}
	A code $C \subseteq V$ in a graph $G = (V,E)$ is a \emph{local $r$-identifying code} if it is an $r$-covering code and $I_{C,r}(u) \neq I_{C,r}(v)$ for any two neighbours $u,v \in V$.
\end{definition}

\begin{definition}
	A code $C \subseteq V$ in a graph $G = (V,E)$ is a \emph{local $r$-locating-dominating code} if it is an $r$-covering code and $I_{C,r}(u) \neq I_{C,r}(v)$ for any two non-codeword neighbours $u,v \in V \setminus C$.
\end{definition}

\noindent
Again, by local identifying and local locating-dominating codes we mean local $1$-identifying and local $1$-locating-dominating codes, respectively.

 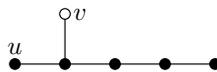
\begin{figure}[!b]
    \centering
    \scalebox{1.1}{
    \begin{tikzpicture}[scale=0.6]
        \draw[fill=black] (1,0) circle(3pt);
        \node[scale=0.8] at (1,0.3) {$u$};
        \draw[fill=black] (2,0) circle(3pt);
        \draw[fill=black] (3,0) circle(3pt);
        \draw[fill=black] (4,0) circle(3pt);
        \draw[fill=black] (5,0) circle(3pt);
        \draw[] (2,1) circle(3pt);
        \node[scale=0.8] at (2.3,1) {$v$};
        \draw (1.1,0) -- (1.9,0);
        \draw (2.1,0) -- (2.9,0);
        \draw (3.1,0) -- (3.9,0);
        \draw (4.1,0) -- (4.9,0);
        \draw (2,0.1) -- (2,0.9);
    \end{tikzpicture} }
    \caption{A graph that admits a local 2-identifying code but does not admit any 2-identifying codes. The  darkened vertices form a local 2-identifying code.}
    \label{Graafi, joka on lokaalisti 2-identifioituva muttei 2-identifioituva.}
\end{figure}

\medskip
Since any graph admits an $r$-locating-dominating code for all $r$ (it is possible to take the whole vertex set as the code), any graph admits also a local $r$-locating-dominating code for all $r$.
However, this is not the case for $r$-identifying and local $r$-identifying codes.
Indeed, any graph that has two distinct vertices with equal $r$-neighbourhoods admits no $r$-identifying codes, and any graph containing two neighbours with equal $r$-neighbourhoods admits no local $r$-identifying codes.
In fact, it is easily seen that a graph $G=(V,E)$ admits an $r$-identifying code if and only if $N_r[u] \neq N_r[v]$ for all $u,v \in V , u \neq v$, and that $G$ admits a local $r$-identifying code if and only if $N_r[u] \neq N_r[v]$ for all $u,v \in V$ such that $u$ and $v$ are neighbours.
For $r=1$ these conditions are the same, that is, a graph admits an identifying code if and only if it admits a local identifying code.
For $r>1$ this is not the case.
See Figure \ref{Graafi, joka on lokaalisti 2-identifioituva muttei 2-identifioituva.} for a graph that admits a local 2-identifying code but does not admit any 2-identifying codes.
Indeed, the vertices $u$ and $v$ in the graph have identical 2-neighborhoods. Hence, no code can 2-separate them.
However, it is easily verified that any two neighbors in the graph have different 2-neighborhoods.
Darkened vertices denote a local $2$-identifying code in the graph. 

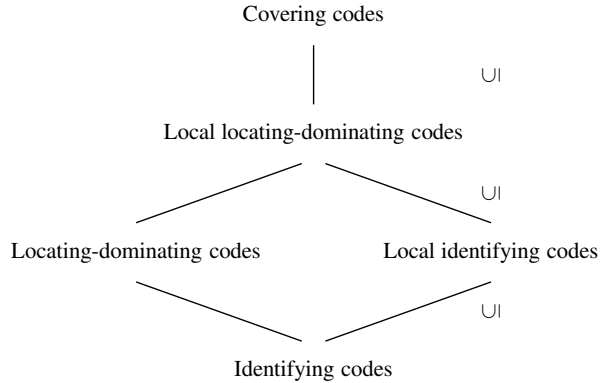
\begin{figure}[!h]
	\centering
\scalebox{1.3}{    \begin{tikzpicture}[scale=0.6]
        \node[scale=0.6] at (0,0) {Covering codes};
        \node[scale=0.6] at (0,-2) {Local locating-dominating codes};
        \node[scale=0.6] at (-3,-4) {Locating-dominating codes};
        \node[scale=0.6] at (3,-4) {Local identifying codes};
        \node[scale=0.6] at (0,-6) {Identifying codes};

        \draw (0,-0.5) -- (0, -1.5);

        \draw (-0.2,-2.5) -- (-3,-3.5);

        \draw (0.2,-2.5) -- (3,-3.5);

        \draw (-3,-4.5) -- (-0.2,-5.5);

        \draw(3,-4.5) -- (0.2,-5.5);


        \node[scale=0.6][rotate=90] at (3,-1) {$\subseteq$};
        \node[scale=0.6][rotate=90] at (3,-3) {$\subseteq$};
        \node[scale=0.6][rotate=90] at (3,-5) {$\subseteq$};
    \end{tikzpicture} }
    \caption{Illustration of the hierarchy between different classes of covering codes.}
    \label{hierarchy}
\end{figure}

\noindent
There is an obvious hierarchy between introduced classes of codes: they are all covering codes, locating-dominating and local identifying codes are both local locating-dominating codes, and identifying codes are locating-dominating codes and also local identifying codes for any fixed covering radius and any fixed graph.
See Figure \ref{hierarchy} for a pictorial illustration.
Depending on the graph these inclusions may or may not be strict. For example, it is quite easy to see that in paths (finite and infinite) and in sufficiently large cycles the classes of identifying and local identifying codes are the same \cite{Herva}.

\medskip
From now on we concentrate on the case $r=1$.
We denote by $\gamma^{ID}(G)$, $\gamma^{LD}(G)$, $\gamma^{L-ID}(G)$ and $\gamma^{L-LD}(G)$ the cardinalities of optimal identifying, locating-dominating, local identifying and local locating-dominating codes, respectively, in a graph $G$. We call these values \textit{identification}, \textit{location-domination}, \textit{local identification} and \textit{local location-domination numbers}, respectively. In particular, we have $\ID(G)\geq\LID(G)\geq \LLD(G)$ and $\ID(G)\geq\LD(G)\geq\LLD(G)$ for any graph $G$ admitting an identifying code.

In \cite{Herva}, local $r$-identifying codes were studied in paths and in cycles.
It was proved that in both finite and infinite paths and in sufficiently large cycles the classes of local $r$-identifying and $r$-identifying codes are equal for all $r$.

The following lemma is useful in our forthcoming considerations.
A graph $G$ is \emph{triangle-free} if it does not contain any triangles by which we mean that the graph does not have any 3-cycles as induced subgraphs.

\begin{lemma}\label{triangle-free lemma}
	A code in a triangle-free graph is a local locating-dominating code if and only if it is a covering code.
\end{lemma}

\begin{proof}
	First, local locating-dominating codes are covering codes by definition.
	
	For the converse claim let $C$ be any covering code in a triangle-free graph $G$
	and
	assume the contrary that $C$ is not local locating-dominating.
	Thus, there exist non-codeword neighbours $u$ and $v$ with equal $I$-sets.
	Since $C$ is a covering code, we have $I(u) = I(v) \neq \emptyset$.
	Hence, there is a triangle in $G$.
	A contradiction.
\end{proof}

In \cite{muller1987np}, it has been shown that finding an optimal covering code -- that is, an optimal dominating set -- in a triangle-free graph (or more specifically in chordal bipartite graphs, a subclass of triangle-free graphs) is an NP-hard problem. Thus, the previous lemma implies the following corollary.

\begin{corollary}
    Finding the cardinality of an optimal local locating-dominating code in a graph $G$ is an NP-hard problem. Even when restricted to chordal bipartite graphs.
\end{corollary}

In \cite{cohen1999identifying}, the authors have shown that finding an optimal identifying code in a graph $G$ is an NP-hard problem. Their proof is based on a reduction from the well-known $3$-SAT problem. Furthermore, the exactly same reduction works also for local identifying codes. Hence, we obtain the following corollary.
\begin{corollary}
    Finding the cardinality of an optimal local identifying code in a graph $G$ is an NP-hard problem.
\end{corollary}

In the following theorem, we show that it is not possible to give any useful lower bound for local identifying (locating-dominating) codes with identification (location-domination) number.

\begin{theorem}
    Let $K_{2,n}$ be a complete bipartite graph on $n+2\geq5$ vertices. We have $\ID(K_{2,n})=\LD(K_{2,n})=n$ and $\LID(K_{2,n})=\LLD(K_{2,n})=2$.
\end{theorem}
\begin{proof}
    Let $K_{2,n}$ have bipartitioning to sets $A$ and $B$ and let $|A|=2$. Observe that $A$ is a local identifying code and thus, also a local locating-dominating code. Moreover, we need at least two vertices to dominate $K_{2,n}$. Hence, $\LID(K_{2,n})=\LLD(K_{2,n})=2$.

\medskip
    Let us then consider usual identification and location-domination. A set containing one vertex from $A$ and all but one from $B$ is an identifying code (and thus, also a locating-dominating code). Moreover, to separate vertices in $A$, we require at least one vertex from $A$ to locating-dominating code. Furthemore, only vertices in $B$ can separate them from other vertices in $B$. Thus, we require all but one vertex from $B$ to any locating-dominating code. The claim follows from the fact $\LD(G)\leq \ID(G)$.
\end{proof}\vspace*{-3mm}

An important concept in finding lower bounds for sizes of optimal codes is the concept of \emph{share} introduced by Slater in \cite{Slater3}.
In the following we define this concept only for $r=1$ although it can be defined for general $r$.

\begin{definition}
Let $C$ be a covering code in graph $G$.
The share $s(c)$ of a codeword $c \in C$ is defined as
$$
s(c) = \sum_{u\in N[c]} \frac{1}{|I_{C}(u)|}.
$$
\end{definition}

For example, the codeword $c$ in Figure~\ref{Graafi, joka on lokaalisti 2-identifioituva muttei 2-identifioituva.2} has share $s(c)=1/2+1+1/3+1/3=13/6$. \\ 

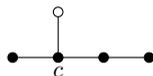
\begin{figure}[!h]
\vspace*{-2mm}
    \centering
    \begin{tikzpicture}[scale=0.6]
        \draw[fill=black] (1,0) circle(3pt);
        \draw[fill=black] (2,0) circle(3pt);
        \node[scale=0.8] at (2,-0.3) {$c$};
        \draw[fill=black] (3,0) circle(3pt);
        \draw[fill=black] (4,0) circle(3pt);
        \draw[] (2,1) circle(3pt);
        \draw (1.1,0) -- (1.9,0);
        \draw (2.1,0) -- (2.9,0);
        \draw (3.1,0) -- (3.9,0);
        \draw (2,0.1) -- (2,0.9);

    \end{tikzpicture}
    \caption{The darkened vertices form  a (non-optimal) covering code.}
    \label{Graafi, joka on lokaalisti 2-identifioituva muttei 2-identifioituva.2}
\end{figure}

The following lemma is well-known and easy to prove. It provides a useful way to find lower bounds for sizes of optimal codes.

\begin{lemma}\label{share finite}
    Let $C$ be a covering code in a finite graph $G=(V,E)$. If $s(c) \leq \alpha$ for every $c\in C$, then
    $$
    |C| \geq \frac{|V|}{\alpha}.
    $$
\end{lemma}

\noindent
Thus, an upper bound for the share of an arbitrary codeword provides a lower bound for the size of the code.
We have a similar lemma for shares and densities in some infinite graphs which we will discuss in Section \ref{section: infinite grids}.

\subsection*{Related concepts}\label{section: Related concepts}
Besides identifying and locating-dominating codes, local identifying and local locating-dominating codes resemble also local resolving sets as we have mentioned. Recently, a new variant of local resolving sets, \textit{nonlocal resolving sets}, was introduced in \cite{klavvzar2023nonlocal}. While local resolving sets distinguish adjacent vertices, nonlocal resolving sets are their dual concept and can distinguish any non-adjacent pair of vertices. 

\medskip
It is possible to define local identifying codes using list-colouring. Let colours be some positive integers. Let $c$ be a function giving a list of colours for each vertex in $V(G)$ for some graph $G$. We do not restrict the maximum length of a list connected to a vertex. However, we give the following two restrictions for the list-colouring. If $d(u,v)=2$ for $u,v\in V(G)$, then $c(u)\cap c(v)\neq \emptyset$, that is, vertices at distance exactly two must share a colour in their list of colours. Secondly, if $d(u,v)=1$ for $u,v\in V(G)$, then $c(u)\cap c(v)= \emptyset$. We call this \textit{locality colouring}. Since there are $(n^2-n)/2$ distinct pairs of vertices in an $n$-vertex graph, there always exists a locality colouring with $(n^2-n)/2$ colours.

Let us consider a dominating set $S$ in $G$ together with the following property: For any pair of vertices $u,v\in V(G)$ with $c(u)\cap  c(v)=\emptyset$ we have $I(u)\neq I(v)$. Observe that set $S$ is a local identifying code. Indeed, it is dominating and it separates any adjacent vertices. Let us then consider a local identifying code $C$ in $G$. Observe that if $c(u)\cap c(v)=\emptyset$, then $u$ and $v$ are either adjacent and $I(u)\neq I(v)$ or $d(u,v)\geq3$ and again $I(u)\neq I(v)$ since $C$ is dominating. Thus, we could have defined local identifying codes also using list-colourings. Moreover, colouring related separation/location problems have been considered in the literature, for example, in \cite{chakraborty2023new, chartrand2002locating, esperet2012locally}.

In \cite{dev2022red, ranjan2022red}, a concept called  \textit{red-blue  separation} was introduced. We next show a connection between this concept and local identification. In red-blue separation, each vertex of graph $G$ is assigned either red or blue colour. After that a set of vertices $S$ is a red-blue separating set if for any two vertices $u,v\in V(G)$ we have $I(v)=I(u)$ only when $v$ and $u$ have been assigned the same colour. Notice that domination was not required here.

Consider a bipartite graph $G$ with bipartition of vertices to sets $A$ and $B$. Let colouring $c$ be such that we assign colour $1$ (red) to each vertex in $A$ and colour $2$ (blue) to each vertex in $B$. Notice that any adjacent vertices share no colours while any vertices at distance two share a colour. Thus, this colouring is a locality colouring. Therefore, a dominating set of graph $G$ is also a red-blue separating set together with colouring $c$ if and only if it is a local identifying code. Hence, these concepts are closely related.
Moreover, perhaps it would be interesting to consider local separating sets in the future, that is, local identifying codes without the domination property.

\subsection*{Structure of the paper}\label{section: Structure}
First in Section \ref{section: binary hypercubes}, we study local identifying codes in binary hypercubes. In Subsection \ref{subsec:small n}, we give some exact solutions for local identifying codes in small hypercubes. Then, in Subsection \ref{subsec:hammingLowBound}, we give a general and asymptotically tight lower bound for local identifying codes in hypercubes. After that, in Subsection \ref{subsec:upper bounds}, we give general methods for constructing local identifying codes in hypercubes. In particular, these constructions show that our lower bound is essentially tight and the size of optimal local identifying codes is significantly smaller than that of usual identifying codes in hypercubes.

In Section \ref{section: infinite grids}, we consider local locating-dominating and local identifying codes in infinite square, hexagonal (Subsection \ref{subsec:square and hex}), triangular (Subsection \ref{subsec:triangle}) and king grids (Subsection \ref{subsec:king}). In particular, we give an optimal construction for seven out of eight of these cases. Finally, we conclude with Section \ref{sec:conclusions}.

\section{Local identifying codes in binary hypercubes} \label{section: binary hypercubes}

Let us denote by $\F = \{0,1\}$ the binary field and let $n \geq 1$ be an integer.
The set of length $n$ binary words is denoted by $\F^n$ as usual.
The \emph{Hamming distance} $d_H(\bx,\by)$ of two binary words $\bx,\by \in \F^n$ is the number of coordinates in which they differ.
The \emph{binary $n$-dimensional hypercube} 
is the graph $G = (V,E)$ where $V = \F^n$ and $E = \{ \{\bx,\by \} \mid \bx,\by \in \F^n, d_H(\bx,\by) = 1 \}$, {\it i.e.}, two binary words are neighbours in the binary hypercube if and only if their Hamming distance is 1.
In fact, it is easy to see that the Hamming distance between two binary words is the same as their graphic distance in the binary hypercube.
So, from now on by $\F^n$ we mean the above graph.

\medskip
We study local identifying codes in binary hypercubes.
Let us denote by $M^L(n)$ the size of an optimal local identifying code and by $M(n)$ the size of an optimal identifying code in the binary $n$-dimensional hypercube.
Moreover, we denote by $M^{LD}(n)$ and $K(n)$ the sizes of optimal locating-dominating codes and optimal covering codes, respectively, in the binary $n$-dimensional hypercube.
Even though there has been much research concerning identifying codes in  binary hypercubes, the exact value of $M(n)$ is known only for $2 \leq n \leq 7$.
In Table \ref{optimal1} we have listed the known values of $M(n)$, $M^{LD}(n)$ and $K(n)$ and our contributions to the values $M^L(n)$ for $n \in [2,10]$.
Note that the numbers $M(1)$ and $M^L(1)$ are not defined since there are no identifying or local identifying codes in the binary 1-dimensional hypercube $\F$.
Note also that since binary hypercubes are triangle-free, the local locating-dominating codes are exactly the covering codes in binary hypercubes by Lemma \ref{triangle-free lemma}.

We start by determining the exact values of $M^L(n)$ for small $n$.
Then we prove a general lower bound for $M^L(n)$ and an upper bound by a linear code construction.
It turns out that this construction shows that our lower bound cannot be significantly improved since for infinitely many $n$, it yields a code whose size is very close to the lower bound (and actually to the lower bound of covering codes).
Consequently, this implies that for infinitely many $n$ the size of an optimal local identifying code is significantly smaller than the size of an optimal identifying code in the binary $n$-dimensional hypercube.
However, this is not the case in every graph as we will see in Table \ref{known bounds} for triangular grid.

\begin{table}[ht]
    \caption{Known values of $M(n), M^{LD}(n)$ and $K(n)$  and our contributions concerning the values of $M^L(n)$ for $n \in [2,10]$. Keys to the table:  (A) \cite{Charon2}, (B) \cite{honkala2004locating}, (C) \cite{Karpovsky}, (D) \cite[Appendix]{ranto2007identifying}, (E) \cite{junnila2022improved}, (F) \cite{Blass}, (G) \cite{exoo2008new}, (H) \cite{coveringcodes}, (I) \cite{stanton1968covering}, (J) \cite{wille1996new}, (K) \cite{ostergard2005unidirectional}. Left key is for the lower bound and right key for the upper bound. When lower and upper bounds are from the same source, the key is placed only on the right side.}
    \label{optimal1}
    \small
    \centering
    \begin{tabular}{|c||r|r|r|r|r|}
        \hline
        $n$ & $M(n)$ & $M^L(n)$ & $M^{LD}(n)$ & $K(n)$ \\ \hline
         2  & 3 (C)  & 2 & 2 (B) & 2 (H)
         \\ \hline
         3  &  4 (C) &  4 & 4 (B) & 2 (H) \\ \hline
         4  &  7 (F) &  6 & 6 (B) & 4 (H) \\ \hline
         5  &  10 (C) &  8 & 10 (B) & 7 (H)  \\ \hline
         6  & (G) 19 (F)  & 12 -- 16 & 16 -- 18 (B) & (I) 12 (H) \\ \hline
         7  &  32 (F)   & 21 -- 28 & 28 -- 32 (B) & 16 (H) \\ \hline
         8  & (C) 56 -- 61 (A) & 35 -- 48 & (B) 50 -- 61 (D) & 32 (H) \\ \hline
         9  & (C) 101 -- 112 (A) & 62 -- 64 & (B) 91 -- 112 (A) & (K) 62 (J)  \\ \hline
         10  & (C) 183 -- 208 (A)  & 110 -- 128 & (E) 171 -- 208 (A) & 107 -- 120 (K)\\ \hline
    \end{tabular}
\end{table}

\subsection{Small $n$}\label{subsec:small n}

The following example shows that $M^L(2) = 2$ which is strictly smaller than $M(2) = 3$.

\begin{example}
    Let us show that $M^L(2) = 2$.
    First, $M^L(2) \geq 2$ since any local identifying code is a covering code and one cannot cover all the vertices of $\F^2$ with only one codeword.
    However, the code $C = \{00,11\}$ is a local identifying code and thus $M^L(2) \leq 2$.
\end{example}

\noindent
In $\F^3$ the classes of local identifying and identifying codes are the same:

\begin{theorem} \label{optimal3D}
    A code $C \subseteq \F^3$ is an identifying code if and only if it is a local identifying code.
    Thus, $M^L(3) = M(3) = 4$.
\end{theorem}

\begin{proof}
    By definition, any identifying code is also a local identifying code.

    For the converse direction assume on the contrary that there exists a local identifying code $C \subseteq \F^3$ which is not an identifying code.
    Notice that $C$ is a covering code.
    There exist $\bx,\by \in \F^3$ such that $I_{C}(\bx) = I_{C}(\by)$.
    Because $C$ is a local identifying code, $\bx$ and $\by$ cannot be neighbours and hence $d(\bx,\by) \geq 2$.
    However, we cannot have $d(\bx,\by) = 3$ since then $\bx$ and $\by$ could not cover a common codeword and hence they could not have equal non-empty $I$-sets.
    Thus, $d(\bx,\by)=2$.
    Without loss of generality we may assume that $\bx = 000$ and $\by = 110$.
    By the assumption that the $I$-sets of $\bx$ and $\by$ are equal, we conclude that
    the symmetric difference $N[\bx] \Delta N[\by] = \{000,001,110,111\}$ of their neighbourhoods is a subset of $\F^3 \setminus C$
    and hence
    $C \subseteq \F^3 \setminus N[\bx] \Delta N[\by] = \{100,010,011,101\} = C'$.
    As a superset of a local identifying code $C$, also the code $C'$ is a local identifying code.
    However, $I_{C'}(010) = \{010,011 \} = I_{C'}(011)$.
    Since $010$ and $011$ are neighbours, this means that $C'$ is not a local identifying code which is a contradiction.
    Thus, $C$ has to be also an identifying code.
\end{proof}

\subsection{A lower bound}\label{subsec:hammingLowBound}

By proving an upper bound for the share of an arbitrary codeword of an arbitrary local identifying code, we prove the following theorem which provides a lower bound for $M^L(n)$ for $n\geq3$.

\begin{theorem} \label{alaraja}
    For every $n\geq3$
    $$
    M^L(n) \geq \frac{3 \cdot 2^n}{3n-2}.
    $$
\end{theorem}

\begin{proof}
For $n=3$ the claim follows from Theorem \ref{optimal3D}.
So, let us assume then that $n \geq 4$.
Let $C \subseteq \F^n$ be a local identifying code and let $\bc \in C$ be an arbitrary codeword. We show that
$$
s(\bc) \leq \frac{3n-2}{3}
$$
which together with Lemma \ref{share finite}
gives the claim.

\medskip
    Since $C$ is a local identifying code, it separates $\bc$ from its neighbours $\bc + \be_1, \ldots , \bc + \be_n$ where $\be_j$ is the binary word of weight one whose $j$th component is 1.
    If none of the points $\bc + \be_j$ is a codeword, {\it i.e.}, if $I(\bc) = \{ \bc \}$, then each of them is covered by at least two codewords and hence
    \begin{equation}\label{EqShareIndependentIDcode}
    s(\bc) = \frac{1}{|I(\bc)|} + \sum_{j=1}^n \frac{1}{|I(\bc + \be_j)|} \leq 1 + \frac{n}{2}.
    \end{equation}

    Let us then assume that $\bc + \be_k \in C$ for some $k \in \{ 1, \ldots , n \}$. In order to separate the neighbours $\bc$ and $\bc + \be_k$, the code $C$ has to cover at least one of them by an extra codeword, {\it i.e.}, $|I(\bc)| \geq 3$ or $|I(\bc + \be_k)| \geq 3$.
    In both cases $|I(\bc + \be_l)| \geq 2$ for some $l \in \{ 1, \ldots , n \} \setminus \{ k \}$. Thus,
    $$
    s(\bc) \leq \frac{1}{3} + 2 \cdot \frac{1}{2} + (n-2) \cdot 1 = \frac{3n-2}{3}.
    $$

    For $n \geq 4$ we have $1 + \frac{n}{2} \leq \frac{3n-2}{3}$.
    The claim follows.
\end{proof}

We can utilize the concept of share also in proofs when the value of $n$ is specified as we will see in the following proof.
In the proof of the following theorem we mean by the weight of $\bx$ the number of its coordinates that have symbol 1.

\begin{theorem}
    $$
    M^L(4) = 6.
    $$
\end{theorem}

\begin{proof}
    It is straight-forward to verify that the code
    $$
    C = \{ 0000,0100,0010,0111,1111,1101\}
    $$
    is a local identifying code in the binary hypercube $\F^4$.
    Thus, $M^L(4) \leq 6$.

\medskip
    Let us show that $M^L(4) \geq 6$.
    Assume the contrary that $M^L(4) < 6$.
    Then there exists a local identifying code $C \subseteq \F^4$ with five codewords.
    We split the proof into two cases based on whether there exists a codeword $\bc$ with $|I(\bc)|\geq3$.

\medskip
    \textbf{Case 1.} Assume first that there exists a codeword $\bc\in C$ with $|I(\bc)|\geq3$ and without loss of generality that $\bc=\mathbf{0}$ and $\{1000,0100\}\subseteq I(\bc)$. Consider next the subcase with $\mathbf{1}\in C$. Let us denote the remaining codeword with $\bc_5$.  The only vertex which is not covered by $C$ at this point is $0011$. If weight of $\bc_5$ is three, then $I(\bc_5)=I(\mathbf{1})$. If $\bc_5$ covers $0011$ and its weight is at most two, then $I(\mathbf{1})=I(1110)$, a contradiction. Thus, we may assume that $\mathbf{1}\not\in C$.

    Since $\mathbf{1}\not\in C$, we require at least one weight three codeword $\bc_4\in C$ to cover $\mathbf{1}$. However, two codewords (other than $\mathbf{1}$) can cover all four weight three codewords only if their weight is two, a contradiction.

\medskip
    \textbf{Case 2.} Assume next that there does not exist a codeword $\bc\in C$ with $|I(\bc)|\geq3$.
    If $I(\bc) = \{\bc,\bc'\}$ for some $\bc$, then $|I(\bc')| \geq 3$.
    Hence, we have $|I(\bc)|=1$ for each codeword $\bc\in C$. In this case, we consider share. As we have counted in Equation (\ref{EqShareIndependentIDcode}), we have $s(\bc)\leq n/2+1=3$ for each $\bc\in C$. Thus, $|C|\geq 2^4/3=5\frac{1}{3}>5$. Therefore, we again have a contradiction and $M^L(4)=6$.\end{proof}

\noindent
In the following subsection we will see that there exist infinitely many $n$ for which there exists a local identifying code in $\F^n$ such that its size is arbitrarily close to the obtained lower bound.
%
This means that for infinitely many $n$ the size of an optimal identifying code in the binary $n$-dimensional hypercube is approximately at least two times larger than the size of an optimal local identifying code.

\subsection{Upper bounds}\label{subsec:upper bounds}

The \emph{direct sum} of two codes $C_1 \subseteq \F^n$ and $C_2 \subseteq \F^m$ is the code
$$
C_1 \oplus C_2 = \{ (\bc_1,\bc_2)  \mid \bc_1 \in C_1, \bc_2 \in C_2 \} \subseteq \F^{n+m}.
$$

\begin{lemma}\label{Lem_dimensionKasvatus}
    Let $C \subseteq \F^n$ be a local identifying code. Then the code $C' = \F \oplus C \subseteq \F^{n+1}$ is a local identifying code if and only if $|I(\bc)| \geq 2$ for every $\bc \in C$.
\end{lemma}

\begin{proof}
    Let us first assume that $C' = \F \oplus C$ is a local identifying code. Assume on the contrary that there exists a codeword $\bc \in C$ such that $I_C(\bc) =  \{ \bc \}$. Then $C'$ does not separate the neighbours $(0,\bc)$ and $(1,\bc)$. A contradiction.

\medskip
    For the converse direction, assume then that $|I_C(\bc)| \geq 2$ for every $\bc \in C$.
    Let $\bx' = (a,\bx) \in \F^{n+1}$ where $a \in \F$ and $\bx \in \F^n$.
    If $\bx \in C$, then $I_{C'}(\bx') = \{ \bx', (a+1,\bx) \} \cup \{ (a,\bc) \mid \bc \in I_C(\bx) \}$ and if $\bx \not \in C$, then $I_{C'}(\bx') = \{ (a,\bc) \mid \bc \in I_C(\bx) \}$.
    Thus, $I_{C'}(\bx') \neq \emptyset$ since $I_C(\bx) \neq \emptyset$ and hence $C'$ is a covering code.
    Let us then show that $C'$ separates any two neighbours.
    So, let $\bx' = (a,\bx)$ and $\by' = (b,\by)$ be neighbours. Assume first that $a=b$. Then $\bx$ and $\by$ have to be neighbours.
    Since $C$ is a local identifying code, we have $I_C(\bx) \neq I_C(\by)$ and hence also $I_{C'}(\bx') \neq I_{C'}(\by')$.
    Assume then that $a \neq b$.
    We have $\bx = \by$ since $\bx'$ and $\by'$ are neighbours.
    By the assumption  $|I_C(\bc)| \geq 2$, for every $\bc \in C$, there exists a codeword $\bc \in I_C(\bx) = I_C(\by)$ such that $\bc \neq \bx = \by$.
    Thus, $(a,\bc), (b,\bc) \in I_{C'}(\bx') \triangle I_{C'}(\by')$ and hence $I_{C'}(\bx') \neq I_{C'}(\by')$.
    So, we conclude that $C'$ separates any two neighbours and hence $C'$ is a local identifying code.
\end{proof}
In the following lemma, we use the fact that if the intersection $N[\bx] \cap N[\by] \cap N[\bz]$ of the neighbourhoods of three distinct binary words $\bx,\by,\bz$ is non-empty, then it in fact contains a unique point.
This implies, in particular, that if $|I(\bc)| \geq 3$ for some codeword $\bc$ of a code $C \subseteq \F^n$, then $I(\bc)$ is unique meaning that $I(\bc') \neq I(\bc)$ for all $\bc' \in C \setminus \{\bc\}$.

\begin{lemma} \label{lemma1}
    Let $C \subseteq \F^n$ be a code such that $|I(\bc)| \geq 3$ for every $\bc \in C$ and $|I(\bx)| \geq 1$ for every $\bx \in \F^n \setminus C$. Then $C$ is a local identifying code.
\end{lemma}

\begin{proof}
    By the assumptions $C$ is a covering code and thus also a local locating-dominating code by Lemma \ref{triangle-free lemma}. It remains to show that $C$ separates any two neighbours. So, let $\bx \in \F^n$ and $\by \in \F^n$ be neighbours.
    If $\bc \in C$, then $I(\bc)$ is unique since $|I(\bc)| \geq 3$ by the assumption. Thus, if one of $\bx$ or $\by$ is a codeword of $C$, then $C$ separates them.
    Hence, let us assume that both $\bx$ and $\by$ are non-codewords. The code $C$ separates them since it is a local locating-dominating code.
\end{proof}

\noindent
The above lemma yields the following corollary.

\begin{corollary}\label{CorCoverBound}
    $$
    M^L(n+2)\leq 4K(n).
    $$
\end{corollary}

\begin{proof}
    Let $C'\subseteq \F^n$ be a covering code in $\F^n$ and $C=\F^2\oplus C'$. Observe that for each $\bc\in C$ we have $|I(\bc)|\geq3$. Moreover, since $C'$ is a covering  code in $\F^n$, $C$ is a covering code in $\F^{n+2}$. Thus, by Lemma \ref{lemma1}, $C$ is a local identifying code.
\end{proof}

\vspace*{-2mm}
\begin{remark}
 An observant reader may notice that Lemma \ref{lemma1} gives also another upper bound for $M^L(n+1)$. Consider a covering code $C$ such that $N(\bw)\cap C\neq \emptyset$ for each $\bw\in \F^n$. In this case, $C$ is said to be a \textit{total dominating set} and the minimum cardinality of a total dominating set in $\F^n$ is
 denoted by $\gamma^{TD}(n)$. By Lemma \ref{lemma1}, we have $M^L(n+1) \leq 2 \gamma^{TD}(n).$ However, by \cite{Verstraten, azarija2017total, goddard2018note}, we have $2K(n)=\gamma^{TD}(n+1)$. Hence, this approach yields the same upper bound as Corollary \ref{CorCoverBound}.
\end{remark}

Similar ideas as in the proof of Lemma \ref{lemma1} give also the code in the following proposition.
\begin{proposition}
We have $M^L(6)\leq 15$.
\end{proposition}
\begin{proof}
    Proof follows from the code $C=\{100000, 010000, 110000\}\cup\{001100, 001110, 001101\}\cup\{000011,$ $ 100011, 010011\}\cup\{111110, 111010, 110110\}\cup\{111101, 011101, 101101\}$. First, code $C$ is covering and hence local locating-dominating. Thus, $C$ separates adjacent non-codewords. Secondly, the subgraph induced by codewords consists of five separate paths of length three. Hence, $C$ separates all codewords from their neighbours.
\end{proof}

A code $C$ is a \textit{linear code} if for any codewords $\bc_1,\bc_2\in C$, we have $\bc_1+\bc_2\in C$.
Next, we present a general construction of a linear local identifying code which then gives an upper bound for $M^L(n)$. Notice that if $C$ is a linear code in $\F^n$, then $C\oplus\F$ is a linear code in $\F^{n+1}$. For a more thorough survey on the topic see the book \cite{coveringcodes}.

\begin{definition}
    Let $s \geq 2$.
    A \emph{binary Hamming code} of length $n=2^s-1$ is a linear covering code $\mathcal{H}_s \subseteq \F^{2^s-1}$ which contains exactly $2^{n-s}$ codewords.
\end{definition}

\begin{theorem}\label{the_hammmingCode}
    For any binary Hamming code $\mathcal{H}_s$ the closed neighbourhoods of its codewords partition the whole space $\F^{2^s-1}$.
\end{theorem}

Notice that by Corollary \ref{CorCoverBound}, code $\mathcal{H}_s\oplus\F^2$ is a linear local identifying code in $\F^{n+2}$ for $n=2^s-1\geq3$. Hence, by Lemma \ref{Lem_dimensionKasvatus}, $\mathcal{H}_s\oplus\F^k$ is a linear local identifying code in $\F^{n+k}$ for $k\geq2$.

\begin{corollary}\label{yläraja}
    Let $s,k\geq2$ and $n = 2^s+k-1$. Then
    $$
    M^L(n) \leq 2^{2^s+k-s-1}.
    $$
\end{corollary}

In particular when $k=2$, we have $n=2^s+1$, for $s\geq2$, and $$\frac{2^n}{n-2/3}\leq M^L(n)\leq 2^{n-\log_2(n-1)}=\frac{2^n}{n-1}.$$
The lower bound is the one from Theorem \ref{alaraja}.
Next, we compare the lower and upper bounds of local identifying codes to covering codes and identifying codes. We will see that both our lower and upper bounds are quite good and essentially tight for infinite number of values of $n$. We see that for arbitrarily large $n = 2^s+k-1$ where $k$ is ``small'' the upper bound of Corollary \ref{yläraja} is close to the lower bound of Theorem \ref{alaraja}. Actually, it is close to the following lower bound for covering codes: $K(n)\geq \frac{2^n}{n+1}$.
This means that the lower bound is very close to optimal for infinitely many $n$ and the cost (the number of extra codewords) of turning a covering code into a local identifying code is small.

\medskip
Furthermore, in \cite{Karpovsky}, the authors have given the following lower bound for identifying codes $M(n)\geq 2 \cdot \frac{2^{n}}{n+1+2/n}.$ When we compare this to the upper bound in Corollary \ref{yläraja}, we notice that the cardinality of local identifying codes is roughly half of the lower bound for identifying codes (actually this holds even for locating-dominating codes). %
For example, by Theorem \ref{alaraja} and Corollary \ref{yläraja} we are able to conclude that $M^L(9) \in \{62,63,64\}$ while $M(9) \in \{ 101, \ldots, 112 \}$
and $M^{LD}(9) \in \{ 91, \ldots , 112 \}$ as we can see from Table \ref{optimal1}.

In particular, combining the above upper bound and our lower bound of Theorem \ref{alaraja}, we have the following result which gives the precise value of $M^L(5)$.

\begin{theorem}
    $$
    M^L(5) = 8.
    $$
\end{theorem}

\begin{proof}
    By Theorem \ref{alaraja} we have
    $
    M^L(5) \geq \frac{3 \cdot 2^5}{3 \cdot 5 -2} = \frac{3 \cdot 32}{13} \approx 7.38
    $
    and hence $M^L(5) \geq 8$.
    On the other hand, by Corollary \ref{yläraja} we have the upper bound
    $$
    M^L(5) \leq 2^{2^2 +2-2-1}= 8.
    $$

    \vspace*{-9mm}
\end{proof}

\section{Local identifying and local locating-dominating codes in infinite grids}
\label{section: infinite grids}

Let us begin by defining the graphs we consider next.
For a pictorial illustration of these graphs see Figure \ref{gridejä}.\\

\begin{figure}[!h]
    \centering
 \scalebox{1.1}{
    \begin{tikzpicture}[scale=0.6]
           \draw[] (0,0) circle(3pt);
        \draw[] (0,1) circle(3pt);
        \draw[] (0,2) circle(3pt);
        \draw[] (0,3) circle(3pt);
        \draw[] (0,4) circle(3pt);

        \draw[] (1,0) circle(3pt);
        \draw[] (1,1) circle(3pt);
        \draw[] (1,2) circle(3pt);
        \draw[] (1,3) circle(3pt);
        \draw[] (1,4) circle(3pt);

        \draw[] (2,0) circle(3pt);
        \draw[] (2,1) circle(3pt);
        \draw[] (2,2) circle(3pt);
        \draw[] (2,3) circle(3pt);
        \draw[] (2,4) circle(3pt);

        \draw[] (3,0) circle(3pt);
        \draw[] (3,1) circle(3pt);
        \draw[] (3,2) circle(3pt);
        \draw[] (3,3) circle(3pt);
        \draw[] (3,4) circle(3pt);

        \draw[] (4,0) circle(3pt);
        \draw[] (4,1) circle(3pt);
        \draw[] (4,2) circle(3pt);
        \draw[] (4,3) circle(3pt);
        \draw[] (4,4) circle(3pt);

        \draw (0.1,0) -- (0.9,0);
        \draw (1.1,0) -- (1.9,0);
        \draw (2.1,0) -- (2.9,0);
        \draw (3.1,0) -- (3.9,0);

        \draw (0.1,1) -- (0.9,1);
        \draw (1.1,1) -- (1.9,1);
        \draw (2.1,1) -- (2.9,1);
        \draw (3.1,1) -- (3.9,1);

        \draw (0.1,2) -- (0.9,2);
        \draw (1.1,2) -- (1.9,2);
        \draw (2.1,2) -- (2.9,2);
        \draw (3.1,2) -- (3.9,2);

        \draw (0.1,3) -- (0.9,3);
        \draw (1.1,3) -- (1.9,3);
        \draw (2.1,3) -- (2.9,3);
        \draw (3.1,3) -- (3.9,3);

        \draw (0.1,4) -- (0.9,4);
        \draw (1.1,4) -- (1.9,4);
        \draw (2.1,4) -- (2.9,4);
        \draw (3.1,4) -- (3.9,4);

        \draw (0,0.1) -- (0,0.9);
        \draw (0,1.1) -- (0,1.9);
        \draw (0,2.1) -- (0,2.9);
        \draw (0,3.1) -- (0,3.9);

        \draw (1,0.1) -- (1,0.9);
        \draw (1,1.1) -- (1,1.9);
        \draw (1,2.1) -- (1,2.9);
        \draw (1,3.1) -- (1,3.9);

        \draw (2,0.1) -- (2,0.9);
        \draw (2,1.1) -- (2,1.9);
        \draw (2,2.1) -- (2,2.9);
        \draw (2,3.1) -- (2,3.9);

        \draw (3,0.1) -- (3,0.9);
        \draw (3,1.1) -- (3,1.9);
        \draw (3,2.1) -- (3,2.9);
        \draw (3,3.1) -- (3,3.9);

        \draw (4,0.1) -- (4,0.9);
        \draw (4,1.1) -- (4,1.9);
        \draw (4,2.1) -- (4,2.9);
        \draw (4,3.1) -- (4,3.9);

        \node[scale=0.6] at (2.3, 1.75) {\tiny $(0,0)$};

        \node[scale=0.6] at (2,-1) {$(a)$ The square grid $\mathcal{S}$};

        \draw[] (6,0) circle(3pt);
        \draw[] (6,1) circle(3pt);
        \draw[] (6,2) circle(3pt);
        \draw[] (6,3) circle(3pt);
        \draw[] (6,4) circle(3pt);

        \draw[] (7,0) circle(3pt);
        \draw[] (7,1) circle(3pt);
        \draw[] (7,2) circle(3pt);
        \draw[] (7,3) circle(3pt);
        \draw[] (7,4) circle(3pt);

        \draw[] (8,0) circle(3pt);
        \draw[] (8,1) circle(3pt);
        \draw[] (8,2) circle(3pt);
        \draw[] (8,3) circle(3pt);
        \draw[] (8,4) circle(3pt);

        \draw[] (9,0) circle(3pt);
        \draw[] (9,1) circle(3pt);
        \draw[] (9,2) circle(3pt);
        \draw[] (9,3) circle(3pt);
        \draw[] (9,4) circle(3pt);

        \draw[] (10,0) circle(3pt);
        \draw[] (10,1) circle(3pt);
        \draw[] (10,2) circle(3pt);
        \draw[] (10,3) circle(3pt);
        \draw[] (10,4) circle(3pt);

        \draw (6.1,0) -- (6.9,0);
        \draw (7.1,0) -- (7.9,0);
        \draw (8.1,0) -- (8.9,0);
        \draw (9.1,0) -- (9.9,0);

        \draw (6.1,1) -- (6.9,1);
        \draw (7.1,1) -- (7.9,1);
        \draw (8.1,1) -- (8.9,1);
        \draw (9.1,1) -- (9.9,1);

        \draw (6.1,2) -- (6.9,2);
        \draw (7.1,2) -- (7.9,2);
        \draw (8.1,2) -- (8.9,2);
        \draw (9.1,2) -- (9.9,2);

        \draw (6.1,3) -- (6.9,3);
        \draw (7.1,3) -- (7.9,3);
        \draw (8.1,3) -- (8.9,3);
        \draw (9.1,3) -- (9.9,3);

        \draw (6.1,4) -- (6.9,4);
        \draw (7.1,4) -- (7.9,4);
        \draw (8.1,4) -- (8.9,4);
        \draw (9.1,4) -- (9.9,4);

        \draw (6,0.1) -- (6,0.9);
        \draw (8,0.1) -- (8,0.9);
        \draw (10,0.1) -- (10,0.9);

        \draw (7,1.1) -- (7,1.9);
        \draw (9,1.1) -- (9,1.9);

        \draw (6,2.1) -- (6,2.9);
        \draw (8,2.1) -- (8,2.9);
        \draw (10,2.1) -- (10,2.9);

        \draw (7,3.1) -- (7,3.9);
        \draw (9,3.1) -- (9,3.9);

        \node[scale=0.6] at (8.3,1.75) {\tiny $(0,0)$};

        \node[scale=0.6] at (8,-1) {$(b)$ The hexagonal grid $\mathcal{H}$};

                    \begin{scope}[xshift=12cm,yshift=7cm]
        \draw[] (0,-3) circle(3pt);
        \draw[] (0,-4) circle(3pt);
        \draw[] (0,-5) circle(3pt);
        \draw[] (0,-6) circle(3pt);
        \draw[] (0,-7) circle(3pt);

        \draw[] (1,-3) circle(3pt);
        \draw[] (1,-4) circle(3pt);
        \draw[] (1,-5) circle(3pt);
        \draw[] (1,-6) circle(3pt);
        \draw[] (1,-7) circle(3pt);

        \draw[] (2,-3) circle(3pt);
        \draw[] (2,-4) circle(3pt);
        \draw[] (2,-5) circle(3pt);
        \draw[] (2,-6) circle(3pt);
        \draw[] (2,-7) circle(3pt);

        \draw[] (3,-3) circle(3pt);
        \draw[] (3,-4) circle(3pt);
        \draw[] (3,-5) circle(3pt);
        \draw[] (3,-6) circle(3pt);
        \draw[] (3,-7) circle(3pt);

        \draw[] (4,-3) circle(3pt);
        \draw[] (4,-4) circle(3pt);
        \draw[] (4,-5) circle(3pt);
        \draw[] (4,-6) circle(3pt);
        \draw[] (4,-7) circle(3pt);

        \draw (0.1,-3) -- (0.9,-3);
        \draw (1.1,-3) -- (1.9,-3);
        \draw (2.1,-3) -- (2.9,-3);
        \draw (3.1,-3) -- (3.9,-3);

        \draw (0.1,-4) -- (0.9,-4);
        \draw (1.1,-4) -- (1.9,-4);
        \draw (2.1,-4) -- (2.9,-4);
        \draw (3.1,-4) -- (3.9,-4);

        \draw (0.1,-5) -- (0.9,-5);
        \draw (1.1,-5) -- (1.9,-5);
        \draw (2.1,-5) -- (2.9,-5);
        \draw (3.1,-5) -- (3.9,-5);

        \draw (0.1,-6) -- (0.9,-6);
        \draw (1.1,-6) -- (1.9,-6);
        \draw (2.1,-6) -- (2.9,-6);
        \draw (3.1,-6) -- (3.9,-6);

        \draw (0.1,-7) -- (0.9,-7);
        \draw (1.1,-7) -- (1.9,-7);
        \draw (2.1,-7) -- (2.9,-7);
        \draw (3.1,-7) -- (3.9,-7);

        \draw (0,-3.1) -- (0,-3.9);
        \draw (0,-4.1) -- (0,-4.9);
        \draw (0,-5.1) -- (0,-5.9);
        \draw (0,-6.1) -- (0,-6.9);

        \draw (1,-3.1) -- (1,-3.9);
        \draw (1,-4.1) -- (1,-4.9);
        \draw (1,-5.1) -- (1,-5.9);
        \draw (1,-6.1) -- (1,-6.9);

        \draw (2,-3.1) -- (2,-3.9);
        \draw (2,-4.1) -- (2,-4.9);
        \draw (2,-5.1) -- (2,-5.9);
        \draw (2,-6.1) -- (2,-6.9);

        \draw (3,-3.1) -- (3,-3.9);
        \draw (3,-4.1) -- (3,-4.9);
        \draw (3,-5.1) -- (3,-5.9);
        \draw (3,-6.1) -- (3,-6.9);

        \draw (4,-3.1) -- (4,-3.9);
        \draw (4,-4.1) -- (4,-4.9);
        \draw (4,-5.1) -- (4,-5.9);
        \draw (4,-6.1) -- (4,-6.9);

        \draw (0.07,-3.93) -- (0.93,-3.07);
        \draw (0.07,-4.93) -- (0.93,-4.07);
        \draw (0.07,-5.93) -- (0.93,-5.07);
        \draw (0.07,-6.93) -- (0.93,-6.07);

        \draw (1.07,-3.93) -- (1.93,-3.07);
        \draw (1.07,-4.93) -- (1.93,-4.07);
        \draw (1.07,-5.93) -- (1.93,-5.07);
        \draw (1.07,-6.93) -- (1.93,-6.07);

        \draw (2.07,-3.93) -- (2.93,-3.07);
        \draw (2.07,-4.93) -- (2.93,-4.07);
        \draw (2.07,-5.93) -- (2.93,-5.07);
        \draw (2.07,-6.93) -- (2.93,-6.07);

        \draw (3.07,-3.93) -- (3.93,-3.07);
        \draw (3.07,-4.93) -- (3.93,-4.07);
        \draw (3.07,-5.93) -- (3.93,-5.07);
        \draw (3.07,-6.93) -- (3.93,-6.07);

        \node[scale=0.6] at (2.3, -5.15) {\tiny $(0,0)$};

        \node[scale=0.6] at (2,-8) {$(c)$ The triangular grid $\mathcal{T}$};

        \draw[] (6,-3) circle(3pt);
        \draw[] (6,-4) circle(3pt);
        \draw[] (6,-5) circle(3pt);
        \draw[] (6,-6) circle(3pt);
        \draw[] (6,-7) circle(3pt);

        \draw[] (7,-3) circle(3pt);
        \draw[] (7,-4) circle(3pt);
        \draw[] (7,-5) circle(3pt);
        \draw[] (7,-6) circle(3pt);
        \draw[] (7,-7) circle(3pt);

        \draw[] (8,-3) circle(3pt);
        \draw[] (8,-4) circle(3pt);
        \draw[] (8,-5) circle(3pt);
        \draw[] (8,-6) circle(3pt);
        \draw[] (8,-7) circle(3pt);

        \draw[] (9,-3) circle(3pt);
        \draw[] (9,-4) circle(3pt);
        \draw[] (9,-5) circle(3pt);
        \draw[] (9,-6) circle(3pt);
        \draw[] (9,-7) circle(3pt);

        \draw[] (10,-3) circle(3pt);
        \draw[] (10,-4) circle(3pt);
        \draw[] (10,-5) circle(3pt);
        \draw[] (10,-6) circle(3pt);
        \draw[] (10,-7) circle(3pt);

        \draw (6.1,-3) -- (6.9,-3);
        \draw (7.1,-3) -- (7.9,-3);
        \draw (8.1,-3) -- (8.9,-3);
        \draw (9.1,-3) -- (9.9,-3);

        \draw (6.1,-4) -- (6.9,-4);
        \draw (7.1,-4) -- (7.9,-4);
        \draw (8.1,-4) -- (8.9,-4);
        \draw (9.1,-4) -- (9.9,-4);

        \draw (6.1,-5) -- (6.9,-5);
        \draw (7.1,-5) -- (7.9,-5);
        \draw (8.1,-5) -- (8.9,-5);
        \draw (9.1,-5) -- (9.9,-5);

        \draw (6.1,-6) -- (6.9,-6);
        \draw (7.1,-6) -- (7.9,-6);
        \draw (8.1,-6) -- (8.9,-6);
        \draw (9.1,-6) -- (9.9,-6);

        \draw (6.1,-7) -- (6.9,-7);
        \draw (7.1,-7) -- (7.9,-7);
        \draw (8.1,-7) -- (8.9,-7);
        \draw (9.1,-7) -- (9.9,-7);

        \draw (6,-3.1) -- (6,-3.9);
        \draw (6,-4.1) -- (6,-4.9);
        \draw (6,-5.1) -- (6,-5.9);
        \draw (6,-6.1) -- (6,-6.9);

        \draw (7,-3.1) -- (7,-3.9);
        \draw (7,-4.1) -- (7,-4.9);
        \draw (7,-5.1) -- (7,-5.9);
        \draw (7,-6.1) -- (7,-6.9);

        \draw (8,-3.1) -- (8,-3.9);
        \draw (8,-4.1) -- (8,-4.9);
        \draw (8,-5.1) -- (8,-5.9);
        \draw (8,-6.1) -- (8,-6.9);

        \draw (9,-3.1) -- (9,-3.9);
        \draw (9,-4.1) -- (9,-4.9);
        \draw (9,-5.1) -- (9,-5.9);
        \draw (9,-6.1) -- (9,-6.9);

        \draw (10,-3.1) -- (10,-3.9);
        \draw (10,-4.1) -- (10,-4.9);
        \draw (10,-5.1) -- (10,-5.9);
        \draw (10,-6.1) -- (10,-6.9);

        \draw (6.07,-3.93) -- (6.93,-3.07);
        \draw (6.07,-4.93) -- (6.93,-4.07);
        \draw (6.07,-5.93) -- (6.93,-5.07);
        \draw (6.07,-6.93) -- (6.93,-6.07);

        \draw (7.07,-3.93) -- (7.93,-3.07);
        \draw (7.07,-4.93) -- (7.93,-4.07);
        \draw (7.07,-5.93) -- (7.93,-5.07);
        \draw (7.07,-6.93) -- (7.93,-6.07);

        \draw (8.07,-3.93) -- (8.93,-3.07);
        \draw (8.07,-4.93) -- (8.93,-4.07);
        \draw (8.07,-5.93) -- (8.93,-5.07);
        \draw (8.07,-6.93) -- (8.93,-6.07);

        \draw (9.07,-3.93) -- (9.93,-3.07);
        \draw (9.07,-4.93) -- (9.93,-4.07);
        \draw (9.07,-5.93) -- (9.93,-5.07);
        \draw (9.07,-6.93) -- (9.93,-6.07);

        \draw (6.07,-3.07) -- (6.93,-3.93);
        \draw (6.07,-4.07) -- (6.93,-4.93);
        \draw (6.07,-5.07) -- (6.93,-5.93);
        \draw (6.07,-6.07) -- (6.93,-6.93);

        \draw (7.07,-3.07) -- (7.93,-3.93);
        \draw (7.07,-4.07) -- (7.93,-4.93);
        \draw (7.07,-5.07) -- (7.93,-5.93);
        \draw (7.07,-6.07) -- (7.93,-6.93);

        \draw (8.07,-3.07) -- (8.93,-3.93);
        \draw (8.07,-4.07) -- (8.93,-4.93);
        \draw (8.07,-5.07) -- (8.93,-5.93);
        \draw (8.07,-6.07) -- (8.93,-6.93);

        \draw (9.07,-3.07) -- (9.93,-3.93);
        \draw (9.07,-4.07) -- (9.93,-4.93);
        \draw (9.07,-5.07) -- (9.93,-5.93);
        \draw (9.07,-6.07) -- (9.93,-6.93);

        \node[scale=0.6] at (8.3, -5.15) {\tiny $(0,0)$};

        \node[scale=0.6] at (8,-8) {$(d)$ The king grid $\mathcal{K}$};
        \end{scope}
    \end{tikzpicture} }
    \caption{Infinite grids.} \label{gridejä}\vspace*{-5mm}
\end{figure}
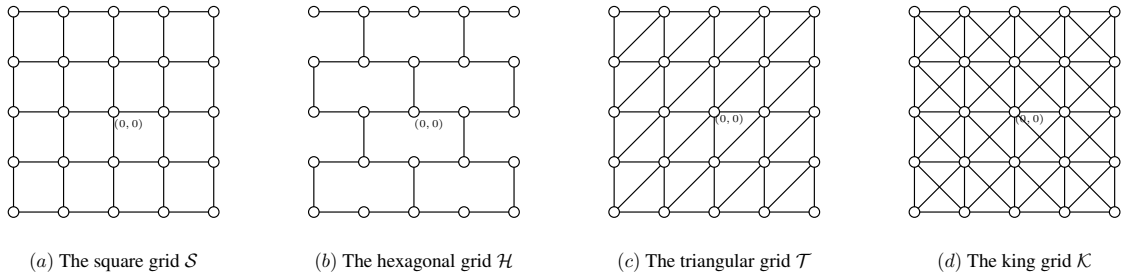

\begin{definition}
An \emph{infinite grid} is one of the following four graphs.
\begin{itemize}
\itemsep=0.9pt
    \item The \emph{square grid} is the graph $\mathcal{S} = (\Z^2, E_{\mathcal{S}})$ where

    \centerline{$ E_{\mathcal{S}} = \{ \{ \bu , \bv \} \mid \bu - \bv \in \{  (\pm 1,0), (0,\pm 1) \} \}.$}

    \item The \emph{hexagonal grid} is the graph  $\mathcal{H} = (\Z^2, E_{\mathcal{H}})$ where

     \centerline{$    E_{\mathcal{H}} = \{ \{ \bu = (i,j), \bv \} \mid \bu - \bv \in  \{(\pm 1,0), (0, (-1)^{i+j+1}) \} \}.$}

    \item The \emph{triangular grid} is the graph $\mathcal{T} =(\Z^2, E_{\mathcal{T}})$ where

     \centerline{$ E_{\mathcal{T}} = \{ \{ \bu, \bv \} \mid \bu - \bv \in \{ (\pm 1,0),(0,\pm 1),(1,1),(-1,-1) \}  \}.$}

    \item The \emph{king grid} is the graph $\mathcal{K} = (\Z^2, E_{\mathcal{K}})$ where

     \centerline{$E_{\mathcal{K}} = \{ \{ \bu, \bv \} \mid \bu - \bv \in \{ (\pm 1,0),(0,\pm 1),(\pm 1,\pm 1) \}  \}.$}
\end{itemize}
\end{definition}

\noindent
Next, we define the concept of \emph{density of a code} in infinite grids.

\begin{definition}
Let $G$ be an infinite grid and let $C \subseteq \Z^2$ be a code in $G$.
The \emph{density} $D(C)$ of $C$ is defined as
$$
D(C) = \limsup_{n \to \infty} \frac{|C \cap Q_n|}{|Q_n|}
$$
where $Q_n = \{ (i,j) \in \Z^2 \mid |i| \leq n, \ |j| \leq n  \}$.
\end{definition}

\noindent
We say that a code in some class of codes is \emph{optimal} if it has the smallest density among the codes in the same class.
In this section, we denote by $\gamma^{ID}(G)$, $\gamma^{LD}(G)$, $\gamma^{L-ID}(G)$ and $\gamma^{L-LD}(G)$ the \textit{densities} of optimal identifying, locating-dominating, local identifying and local locating-dominating codes, respectively, in an infinite grid $G$.
The densities $\gamma^{ID}(G)$ and $\gamma^{LD}(G)$ are all known when $G$ is the square, the triangular or the king grid.
The number $\gamma^{LD}(\mathcal{H})$ is also known while the number $\gamma^{ID}(\mathcal{H})$ is currently still unknown.
However, note that interestingly  the exact value of the density of optimal 2-identifying codes in the hexagonal grid is known~\cite{Junnila}.
In Table \ref{known bounds} we have listed the known values for the densities of optimal identifying and locating-dominating codes in each infinite grid (and an interval for the density $\gamma^{ID}(\mathcal{H})$ where we know it belongs to) and our contributions to the optimal densities of local identifying and local locating-dominating codes.

\begin{table}[ht]
    \small
    \centering
    \caption{Known values for the densities of optimal identifying and locating-dominating codes and contributions of this paper for the densities of optimal local identifying and local locating-dominating codes in infinite grids.
    For the densities of optimal identifying codes in the hexagonal grid and local locating-dominating codes in the triangular grid we give lower and upper bounds.}
    \label{known bounds}
    \begin{tabular}{|r||c|c|c|c|}
        \hline
        $G$ & $\mathcal{S}$ & $\mathcal{H}$ & $\mathcal{T}$ & $\mathcal{K}$   \\\hline 
        $\gamma^{ID}(G)$ & $\frac{7}{20}$ (\cite{Ben-Haim}) & $\frac{5}{12}$ -- $\frac{3}{7}$ (\cite{Cukierman, Cohen1}) & $\frac{1}{4}$ (\cite{Karpovsky}) & $\frac{2}{9}$ (\cite{Charon, Cohen2}) \\
        \hline
        $\gamma^{LD}(G)$ & $\frac{3}{10}$ (\cite{Slater3}) & $\frac{1}{3}$ (\cite{Honkala2}) & $\frac{13}{57}$ (\cite{Honkala1}) & $\frac{1}{5}$ (\cite{Honkala2}) \\  \hline
        $\gamma^{L-ID}(G)$ & $\frac{3}{11}$ & $\frac{3}{8}$ & $\frac{1}{4}$ & $\frac{2}{9}$  \\ \hline
        $\gamma^{L-LD}(G)$ & $\frac{1}{5}$ & $\frac{1}{4}$ & $\frac{2}{11}$ -- $\frac{2}{9}$ & $\frac{3}{16}$ \\
        \hline
    \end{tabular}
\end{table}

We study the densities of optimal local identifying and local locating-dominating codes in these four grids.
Again, we use shares and, in particular, the following well-known lemma analogous to Lemma \ref{share finite}.
For completeness, we provide a proof for this result.

\begin{lemma} \label{share infinite}
    Let $G$ be an infinite grid and let $C \subseteq \Z^2$ be a covering code in $G$.
    If for some real $\alpha>0$ we have $s(\bc) \leq \alpha$ for every $\bc \in C$, then $D(C) \geq \frac{1}{\alpha}$.
\end{lemma}

\begin{proof}
    By the assumption that $s(\bc) \leq \alpha$ for every $\bc \in C$ and by the fact that $C$ is a covering code we have
    $$
    |Q_{n-1}| \leq \sum_{\bc \in C \cap Q_n} s(\bc) \leq |C \cap Q_n| \cdot \alpha\vspace*{-1mm}
    $$
    for any $n \geq 1$.
\eject
     Thus,
    $$
    |C \cap Q_n| \geq \frac{|Q_{n-1}|}{\alpha}
    $$
    and hence
    $$
    D(C) = \limsup_{n \to \infty} \frac{|C \cap Q_n|}{|Q_n|}
    \geq \limsup_{n \to \infty} \frac{|Q_{n-1}|}{\alpha \cdot |Q_n|}
    = \frac{1}{\alpha} \cdot \limsup_{n \to \infty} \frac{|Q_{n-1}|}{|Q_n|} = \frac{1}{\alpha}
    $$
    since $|Q_n| = (2n+1)^2$ which implies that $\limsup_{n \to \infty} \frac{|Q_{n-1}|}{|Q_n|} = 1$.
\end{proof}

So, by finding an upper bound for the share of an arbitrary codeword of a code, we obtain a lower bound for the density of the code.
By analyzing the possible shares of codewords of local identifying and local locating-dominating codes in $G$ we get lower bounds for the numbers $\gamma^{L-ID}(G)$ and $\gamma^{L-LD}(G)$ for different grids $G$.

\medskip
To improve the lower bounds obtained by analyzing the maximal shares of the codewords of a code, we sometimes use a \textit{share shifting scheme} where we modify the share function by shifting shares among codewords according to some local rules such that the total share remains the same.
We denote by $s'(\bc)$ the modified share of $\bc$ after applying a share shifting scheme.
For a code $C$ in a finite graph this means that $\sum_{\bc \in C} s(\bc) = \sum_{\bc \in C} s'(\bc)$, and for a code $C$ in an infinite grid this means that
$\sum_{\bc \in C \cap Q_n} s(\bc) \leq \sum_{\bc \in C \cap Q_{n+r}} s'(\bc)$
where $r$ is the maximum distance from a codeword to another codeword it shifts share to.
The following lemma states that an upper bound for the modified share function yields a lower bound for the density of a code.
Our share shifting scheme can be seen as a discharging method, see \cite{cranston2017introduction} for more discussion on this topic.

\begin{lemma}\label{Lemma shareshift}
    Let $G$ be an infinite grid and let $C \subseteq \Z^2$ be a covering code in $G$.
    Let $s'$ be a modified share function of $C$ obtained by a share shifting scheme.
    If $s'(\bc)\! \leq \alpha$ for every $\bc\! \in C$, then $D(C)\! \geq \frac{1}{\alpha}$.
\end{lemma}

\begin{proof}
    Assume that in the share shifting scheme that defines $s'$ codewords obtain shifted share from codewords within distance $r$ from them.
    Since $C$ is a covering code, we have
    $$
    |Q_{n-1}| \leq \sum_{\bc \in C \cap Q_n} s(\bc)
    $$
    and since the total share in $Q_n$ stays in $Q_{n+r}$, we have
    $$
    \sum_{\bc \in C \cap Q_n} s(\bc) \leq \sum_{\bc \in C \cap Q_n} s'(\bc) + \alpha \cdot |Q_{n+r} \setminus Q_n| \leq \alpha \cdot |C \cap Q_n| + \alpha \cdot |Q_{n+r} \setminus Q_n|.
    $$
    By combining these, we get
    $$
    |C \cap Q_n| \geq \frac{1}{\alpha} |Q_{n-1}| - |Q_{n+r} \setminus Q_n|.
    $$
    Thus,
    $$
    D(C) = \limsup_{n \to \infty} \frac{|C \cap Q_n|}{|Q_n|} \geq \frac{1}{\alpha} \limsup_{n \to \infty} \frac{|Q_{n-1}|}{|Q_n|} - \liminf_{n \to \infty} \frac{|Q_{n+r} \setminus Q_n|}{|Q_n|} = \frac{1}{\alpha} - 0 = \frac{1}{\alpha}.
    $$

   \vspace*{-6mm}
\end{proof}

\subsection{The square and the hexagonal grids}\label{subsec:square and hex}

Since the square and the hexagonal grids are triangle-free, Lemma \ref{triangle-free lemma} gives the following theorem.
See Figures \ref{Optimal square grid} and \ref{Optimal hexagonal grid} for constructions.\\

\begin{figure}[!ht]
    \centering
    \begin{subfigure}[]{0.4\textwidth}
      \begin{tikzpicture}[scale=0.6]
        \draw[fill=black] (0,0) circle(3pt);
        \draw[] (0,1) circle(3pt);
        \draw[] (0,2) circle(3pt);
        \draw[] (0,3) circle(3pt);
        \draw[] (0,4) circle(3pt);
        \draw[fill=black] (0,5) circle(3pt);
        \draw[] (0,6) circle(3pt);
        \draw[] (0,7) circle(3pt);
        \draw[] (0,8) circle(3pt);
        \draw[] (0,9) circle(3pt);

        \draw[] (1,0) circle(3pt);
        \draw[] (1,1) circle(3pt);
        \draw[fill=black] (1,2) circle(3pt);
        \draw[] (1,3) circle(3pt);
        \draw[] (1,4) circle(3pt);
        \draw[] (1,5) circle(3pt);
        \draw[] (1,6) circle(3pt);
        \draw[fill=black] (1,7) circle(3pt);
        \draw[] (1,8) circle(3pt);
        \draw[] (1,9) circle(3pt);

        \draw[] (2,0) circle(3pt);
        \draw[] (2,1) circle(3pt);
        \draw[] (2,2) circle(3pt);
        \draw[] (2,3) circle(3pt);
        \draw[fill=black] (2,4) circle(3pt);
        \draw[] (2,5) circle(3pt);
        \draw[] (2,6) circle(3pt);
        \draw[] (2,7) circle(3pt);
        \draw[] (2,8) circle(3pt);
        \draw[fill=black] (2,9) circle(3pt);

        \draw[] (3,0) circle(3pt);
        \draw[fill=black] (3,1) circle(3pt);
        \draw[] (3,2) circle(3pt);
        \draw[] (3,3) circle(3pt);
        \draw[] (3,4) circle(3pt);
        \draw[] (3,5) circle(3pt);
        \draw[fill=black] (3,6) circle(3pt);
        \draw[] (3,7) circle(3pt);
        \draw[] (3,8) circle(3pt);
        \draw[] (3,9) circle(3pt);

        \draw[] (4,0) circle(3pt);
        \draw[] (4,1) circle(3pt);
        \draw[] (4,2) circle(3pt);
        \draw[fill=black] (4,3) circle(3pt);
        \draw[] (4,4) circle(3pt);
        \draw[] (4,5) circle(3pt);
        \draw[] (4,6) circle(3pt);
        \draw[] (4,7) circle(3pt);
        \draw[fill=black] (4,8) circle(3pt);
        \draw[] (4,9) circle(3pt);

        \draw[fill=black] (5,0) circle(3pt);
        \draw[] (5,1) circle(3pt);
        \draw[] (5,2) circle(3pt);
        \draw[] (5,3) circle(3pt);
        \draw[] (5,4) circle(3pt);
        \draw[fill=black] (5,5) circle(3pt);
        \draw[] (5,6) circle(3pt);
        \draw[] (5,7) circle(3pt);
        \draw[] (5,8) circle(3pt);
        \draw[] (5,9) circle(3pt);

        \draw[] (6,0) circle(3pt);
        \draw[] (6,1) circle(3pt);
        \draw[fill=black] (6,2) circle(3pt);
        \draw[] (6,3) circle(3pt);
        \draw[] (6,4) circle(3pt);
        \draw[] (6,5) circle(3pt);
        \draw[] (6,6) circle(3pt);
        \draw[fill=black] (6,7) circle(3pt);
        \draw[] (6,8) circle(3pt);
        \draw[] (6,9) circle(3pt);

        \draw[] (7,0) circle(3pt);
        \draw[] (7,1) circle(3pt);
        \draw[] (7,2) circle(3pt);
        \draw[] (7,3) circle(3pt);
        \draw[fill=black] (7,4) circle(3pt);
        \draw[] (7,5) circle(3pt);
        \draw[] (7,6) circle(3pt);
        \draw[] (7,7) circle(3pt);
        \draw[] (7,8) circle(3pt);
        \draw[fill=black] (7,9) circle(3pt);

        \draw[] (8,0) circle(3pt);
        \draw[fill=black] (8,1) circle(3pt);
        \draw[] (8,2) circle(3pt);
        \draw[] (8,3) circle(3pt);
        \draw[] (8,4) circle(3pt);
        \draw[] (8,5) circle(3pt);
        \draw[fill=black] (8,6) circle(3pt);
        \draw[] (8,7) circle(3pt);
        \draw[] (8,8) circle(3pt);
        \draw[] (8,9) circle(3pt);

        \draw[] (9,0) circle(3pt);
        \draw[] (9,1) circle(3pt);
        \draw[] (9,2) circle(3pt);
        \draw[fill=black] (9,3) circle(3pt);
        \draw[] (9,4) circle(3pt);
        \draw[] (9,5) circle(3pt);
        \draw[] (9,6) circle(3pt);
        \draw[] (9,7) circle(3pt);
        \draw[fill=black] (9,8) circle(3pt);
        \draw[] (9,9) circle(3pt);

        \draw (0.1,0) -- (0.9,0);
        \draw (1.1,0) -- (1.9,0);
        \draw (2.1,0) -- (2.9,0);
        \draw (3.1,0) -- (3.9,0);
        \draw (4.1,0) -- (4.9,0);
        \draw (5.1,0) -- (5.9,0);
        \draw (6.1,0) -- (6.9,0);
        \draw (7.1,0) -- (7.9,0);
        \draw (8.1,0) -- (8.9,0);

        \draw (0.1,1) -- (0.9,1);
        \draw (1.1,1) -- (1.9,1);
        \draw (2.1,1) -- (2.9,1);
        \draw (3.1,1) -- (3.9,1);
        \draw (4.1,1) -- (4.9,1);
        \draw (5.1,1) -- (5.9,1);
        \draw (6.1,1) -- (6.9,1);
        \draw (7.1,1) -- (7.9,1);
        \draw (8.1,1) -- (8.9,1);

        \draw (0.1,2) -- (0.9,2);
        \draw (1.1,2) -- (1.9,2);
        \draw (2.1,2) -- (2.9,2);
        \draw (3.1,2) -- (3.9,2);
        \draw (4.1,2) -- (4.9,2);
        \draw (5.1,2) -- (5.9,2);
        \draw (6.1,2) -- (6.9,2);
        \draw (7.1,2) -- (7.9,2);
        \draw (8.1,2) -- (8.9,2);

        \draw (0.1,3) -- (0.9,3);
        \draw (1.1,3) -- (1.9,3);
        \draw (2.1,3) -- (2.9,3);
        \draw (3.1,3) -- (3.9,3);
        \draw (4.1,3) -- (4.9,3);
        \draw (5.1,3) -- (5.9,3);
        \draw (6.1,3) -- (6.9,3);
        \draw (7.1,3) -- (7.9,3);
        \draw (8.1,3) -- (8.9,3);

        \draw (0.1,4) -- (0.9,4);
        \draw (1.1,4) -- (1.9,4);
        \draw (2.1,4) -- (2.9,4);
        \draw (3.1,4) -- (3.9,4);
        \draw (4.1,4) -- (4.9,4);
        \draw (5.1,4) -- (5.9,4);
        \draw (6.1,4) -- (6.9,4);
        \draw (7.1,4) -- (7.9,4);
        \draw (8.1,4) -- (8.9,4);

        \draw (0.1,5) -- (0.9,5);
        \draw (1.1,5) -- (1.9,5);
        \draw (2.1,5) -- (2.9,5);
        \draw (3.1,5) -- (3.9,5);
        \draw (4.1,5) -- (4.9,5);
        \draw (5.1,5) -- (5.9,5);
        \draw (6.1,5) -- (6.9,5);
        \draw (7.1,5) -- (7.9,5);
        \draw (8.1,5) -- (8.9,5);

        \draw (0.1,6) -- (0.9,6);
        \draw (1.1,6) -- (1.9,6);
        \draw (2.1,6) -- (2.9,6);
        \draw (3.1,6) -- (3.9,6);
        \draw (4.1,6) -- (4.9,6);
        \draw (5.1,6) -- (5.9,6);
        \draw (6.1,6) -- (6.9,6);
        \draw (7.1,6) -- (7.9,6);
        \draw (8.1,6) -- (8.9,6);

        \draw (0.1,7) -- (0.9,7);
        \draw (1.1,7) -- (1.9,7);
        \draw (2.1,7) -- (2.9,7);
        \draw (3.1,7) -- (3.9,7);
        \draw (4.1,7) -- (4.9,7);
        \draw (5.1,7) -- (5.9,7);
        \draw (6.1,7) -- (6.9,7);
        \draw (7.1,7) -- (7.9,7);
        \draw (8.1,7) -- (8.9,7);

        \draw (0.1,8) -- (0.9,8);
        \draw (1.1,8) -- (1.9,8);
        \draw (2.1,8) -- (2.9,8);
        \draw (3.1,8) -- (3.9,8);
        \draw (4.1,8) -- (4.9,8);
        \draw (5.1,8) -- (5.9,8);
        \draw (6.1,8) -- (6.9,8);
        \draw (7.1,8) -- (7.9,8);
        \draw (8.1,8) -- (8.9,8);

        \draw (0.1,9) -- (0.9,9);
        \draw (1.1,9) -- (1.9,9);
        \draw (2.1,9) -- (2.9,9);
        \draw (3.1,9) -- (3.9,9);
        \draw (4.1,9) -- (4.9,9);
        \draw (5.1,9) -- (5.9,9);
        \draw (6.1,9) -- (6.9,9);
        \draw (7.1,9) -- (7.9,9);
        \draw (8.1,9) -- (8.9,9);

        \draw (0,0.1) -- (0,0.9);
        \draw (0,1.1) -- (0,1.9);
        \draw (0,2.1) -- (0,2.9);
        \draw (0,3.1) -- (0,3.9);
        \draw (0,4.1) -- (0,4.9);
        \draw (0,5.1) -- (0,5.9);
        \draw (0,6.1) -- (0,6.9);
        \draw (0,7.1) -- (0,7.9);
        \draw (0,8.1) -- (0,8.9);

        \draw (1,0.1) -- (1,0.9);
        \draw (1,1.1) -- (1,1.9);
        \draw (1,2.1) -- (1,2.9);
        \draw (1,3.1) -- (1,3.9);
        \draw (1,4.1) -- (1,4.9);
        \draw (1,5.1) -- (1,5.9);
        \draw (1,6.1) -- (1,6.9);
        \draw (1,7.1) -- (1,7.9);
        \draw (1,8.1) -- (1,8.9);

        \draw (2,0.1) -- (2,0.9);
        \draw (2,1.1) -- (2,1.9);
        \draw (2,2.1) -- (2,2.9);
        \draw (2,3.1) -- (2,3.9);
        \draw (2,4.1) -- (2,4.9);
        \draw (2,5.1) -- (2,5.9);
        \draw (2,6.1) -- (2,6.9);
        \draw (2,7.1) -- (2,7.9);
        \draw (2,8.1) -- (2,8.9);

        \draw (3,0.1) -- (3,0.9);
        \draw (3,1.1) -- (3,1.9);
        \draw (3,2.1) -- (3,2.9);
        \draw (3,3.1) -- (3,3.9);
        \draw (3,4.1) -- (3,4.9);
        \draw (3,5.1) -- (3,5.9);
        \draw (3,6.1) -- (3,6.9);
        \draw (3,7.1) -- (3,7.9);
        \draw (3,8.1) -- (3,8.9);

        \draw (4,0.1) -- (4,0.9);
        \draw (4,1.1) -- (4,1.9);
        \draw (4,2.1) -- (4,2.9);
        \draw (4,3.1) -- (4,3.9);
        \draw (4,4.1) -- (4,4.9);
        \draw (4,5.1) -- (4,5.9);
        \draw (4,6.1) -- (4,6.9);
        \draw (4,7.1) -- (4,7.9);
        \draw (4,8.1) -- (4,8.9);

        \draw (5,0.1) -- (5,0.9);
        \draw (5,1.1) -- (5,1.9);
        \draw (5,2.1) -- (5,2.9);
        \draw (5,3.1) -- (5,3.9);
        \draw (5,4.1) -- (5,4.9);
        \draw (5,5.1) -- (5,5.9);
        \draw (5,6.1) -- (5,6.9);
        \draw (5,7.1) -- (5,7.9);
        \draw (5,8.1) -- (5,8.9);

        \draw (6,0.1) -- (6,0.9);
        \draw (6,1.1) -- (6,1.9);
        \draw (6,2.1) -- (6,2.9);
        \draw (6,3.1) -- (6,3.9);
        \draw (6,4.1) -- (6,4.9);
        \draw (6,5.1) -- (6,5.9);
        \draw (6,6.1) -- (6,6.9);
        \draw (6,7.1) -- (6,7.9);
        \draw (6,8.1) -- (6,8.9);

        \draw (7,0.1) -- (7,0.9);
        \draw (7,1.1) -- (7,1.9);
        \draw (7,2.1) -- (7,2.9);
        \draw (7,3.1) -- (7,3.9);
        \draw (7,4.1) -- (7,4.9);
        \draw (7,5.1) -- (7,5.9);
        \draw (7,6.1) -- (7,6.9);
        \draw (7,7.1) -- (7,7.9);
        \draw (7,8.1) -- (7,8.9);

        \draw (8,0.1) -- (8,0.9);
        \draw (8,1.1) -- (8,1.9);
        \draw (8,2.1) -- (8,2.9);
        \draw (8,3.1) -- (8,3.9);
        \draw (8,4.1) -- (8,4.9);
        \draw (8,5.1) -- (8,5.9);
        \draw (8,6.1) -- (8,6.9);
        \draw (8,7.1) -- (8,7.9);
        \draw (8,8.1) -- (8,8.9);

        \draw (9,0.1) -- (9,0.9);
        \draw (9,1.1) -- (9,1.9);
        \draw (9,2.1) -- (9,2.9);
        \draw (9,3.1) -- (9,3.9);
        \draw (9,4.1) -- (9,4.9);
        \draw (9,5.1) -- (9,5.9);
        \draw (9,6.1) -- (9,6.9);
        \draw (9,7.1) -- (9,7.9);
        \draw (9,8.1) -- (9,8.9);
    \end{tikzpicture}
       \qquad\caption{An optimal local locating-dominating code which is also an optimal covering code in the square grid.}
       \label{Optimal square grid}
    \end{subfigure}\qquad\qquad
    \begin{subfigure}[]{0.4\textwidth}
\begin{tikzpicture}[scale=0.6]
        \draw[fill=black] (0,0) circle(3pt);
        \draw[] (0,1) circle(3pt);
        \draw[fill=black] (0,2) circle(3pt);
        \draw[] (0,3) circle(3pt);
        \draw[fill=black] (0,4) circle(3pt);
        \draw[] (0,5) circle(3pt);
        \draw[fill=black] (0,6) circle(3pt);
        \draw[] (0,7) circle(3pt);
        \draw[fill=black] (0,8) circle(3pt);
        \draw[] (0,9) circle(3pt);

        \draw[] (1,0) circle(3pt);
        \draw[] (1,1) circle(3pt);
        \draw[] (1,2) circle(3pt);
        \draw[] (1,3) circle(3pt);
        \draw[] (1,4) circle(3pt);
        \draw[] (1,5) circle(3pt);
        \draw[] (1,6) circle(3pt);
        \draw[] (1,7) circle(3pt);
        \draw[] (1,8) circle(3pt);
        \draw[] (1,9) circle(3pt);

        \draw[] (2,0) circle(3pt);
        \draw[fill=black] (2,1) circle(3pt);
        \draw[] (2,2) circle(3pt);
        \draw[fill=black] (2,3) circle(3pt);
        \draw[] (2,4) circle(3pt);
        \draw[fill=black] (2,5) circle(3pt);
        \draw[] (2,6) circle(3pt);
        \draw[fill=black] (2,7) circle(3pt);
        \draw[] (2,8) circle(3pt);
        \draw[fill=black] (2,9) circle(3pt);

        \draw[] (3,0) circle(3pt);
        \draw[] (3,1) circle(3pt);
        \draw[] (3,2) circle(3pt);
        \draw[] (3,3) circle(3pt);
        \draw[] (3,4) circle(3pt);
        \draw[] (3,5) circle(3pt);
        \draw[] (3,6) circle(3pt);
        \draw[] (3,7) circle(3pt);
        \draw[] (3,8) circle(3pt);
        \draw[] (3,9) circle(3pt);

        \draw[fill=black] (4,0) circle(3pt);
        \draw[] (4,1) circle(3pt);
        \draw[fill=black] (4,2) circle(3pt);
        \draw[] (4,3) circle(3pt);
        \draw[fill=black] (4,4) circle(3pt);
        \draw[] (4,5) circle(3pt);
        \draw[fill=black] (4,6) circle(3pt);
        \draw[] (4,7) circle(3pt);
        \draw[fill=black] (4,8) circle(3pt);
        \draw[] (4,9) circle(3pt);

        \draw[] (5,0) circle(3pt);
        \draw[] (5,1) circle(3pt);
        \draw[] (5,2) circle(3pt);
        \draw[] (5,3) circle(3pt);
        \draw[] (5,4) circle(3pt);
        \draw[] (5,5) circle(3pt);
        \draw[] (5,6) circle(3pt);
        \draw[] (5,7) circle(3pt);
        \draw[] (5,8) circle(3pt);
        \draw[] (5,9) circle(3pt);

        \draw[] (6,0) circle(3pt);
        \draw[fill=black] (6,1) circle(3pt);
        \draw[] (6,2) circle(3pt);
        \draw[fill=black] (6,3) circle(3pt);
        \draw[] (6,4) circle(3pt);
        \draw[fill=black] (6,5) circle(3pt);
        \draw[] (6,6) circle(3pt);
        \draw[fill=black] (6,7) circle(3pt);
        \draw[] (6,8) circle(3pt);
        \draw[fill=black] (6,9) circle(3pt);

        \draw[] (7,0) circle(3pt);
        \draw[] (7,1) circle(3pt);
        \draw[] (7,2) circle(3pt);
        \draw[] (7,3) circle(3pt);
        \draw[] (7,4) circle(3pt);
        \draw[] (7,5) circle(3pt);
        \draw[] (7,6) circle(3pt);
        \draw[] (7,7) circle(3pt);
        \draw[] (7,8) circle(3pt);
        \draw[] (7,9) circle(3pt);

        \draw[fill=black] (8,0) circle(3pt);
        \draw[] (8,1) circle(3pt);
        \draw[fill=black] (8,2) circle(3pt);
        \draw[] (8,3) circle(3pt);
        \draw[fill=black] (8,4) circle(3pt);
        \draw[] (8,5) circle(3pt);
        \draw[fill=black] (8,6) circle(3pt);
        \draw[] (8,7) circle(3pt);
        \draw[fill=black] (8,8) circle(3pt);
        \draw[] (8,9) circle(3pt);

        \draw[] (9,0) circle(3pt);
        \draw[] (9,1) circle(3pt);
        \draw[] (9,2) circle(3pt);
        \draw[] (9,3) circle(3pt);
        \draw[] (9,4) circle(3pt);
        \draw[] (9,5) circle(3pt);
        \draw[] (9,6) circle(3pt);
        \draw[] (9,7) circle(3pt);
        \draw[] (9,8) circle(3pt);
        \draw[] (9,9) circle(3pt);

        \draw (0.1,0) -- (0.9,0);
        \draw (1.1,0) -- (1.9,0);
        \draw (2.1,0) -- (2.9,0);
        \draw (3.1,0) -- (3.9,0);
        \draw (4.1,0) -- (4.9,0);
        \draw (5.1,0) -- (5.9,0);
        \draw (6.1,0) -- (6.9,0);
        \draw (7.1,0) -- (7.9,0);
        \draw (8.1,0) -- (8.9,0);

        \draw (0.1,1) -- (0.9,1);
        \draw (1.1,1) -- (1.9,1);
        \draw (2.1,1) -- (2.9,1);
        \draw (3.1,1) -- (3.9,1);
        \draw (4.1,1) -- (4.9,1);
        \draw (5.1,1) -- (5.9,1);
        \draw (6.1,1) -- (6.9,1);
        \draw (7.1,1) -- (7.9,1);
        \draw (8.1,1) -- (8.9,1);

        \draw (0.1,2) -- (0.9,2);
        \draw (1.1,2) -- (1.9,2);
        \draw (2.1,2) -- (2.9,2);
        \draw (3.1,2) -- (3.9,2);
        \draw (4.1,2) -- (4.9,2);
        \draw (5.1,2) -- (5.9,2);
        \draw (6.1,2) -- (6.9,2);
        \draw (7.1,2) -- (7.9,2);
        \draw (8.1,2) -- (8.9,2);

        \draw (0.1,3) -- (0.9,3);
        \draw (1.1,3) -- (1.9,3);
        \draw (2.1,3) -- (2.9,3);
        \draw (3.1,3) -- (3.9,3);
        \draw (4.1,3) -- (4.9,3);
        \draw (5.1,3) -- (5.9,3);
        \draw (6.1,3) -- (6.9,3);
        \draw (7.1,3) -- (7.9,3);
        \draw (8.1,3) -- (8.9,3);

        \draw (0.1,4) -- (0.9,4);
        \draw (1.1,4) -- (1.9,4);
        \draw (2.1,4) -- (2.9,4);
        \draw (3.1,4) -- (3.9,4);
        \draw (4.1,4) -- (4.9,4);
        \draw (5.1,4) -- (5.9,4);
        \draw (6.1,4) -- (6.9,4);
        \draw (7.1,4) -- (7.9,4);
        \draw (8.1,4) -- (8.9,4);

        \draw (0.1,5) -- (0.9,5);
        \draw (1.1,5) -- (1.9,5);
        \draw (2.1,5) -- (2.9,5);
        \draw (3.1,5) -- (3.9,5);
        \draw (4.1,5) -- (4.9,5);
        \draw (5.1,5) -- (5.9,5);
        \draw (6.1,5) -- (6.9,5);
        \draw (7.1,5) -- (7.9,5);
        \draw (8.1,5) -- (8.9,5);

        \draw (0.1,6) -- (0.9,6);
        \draw (1.1,6) -- (1.9,6);
        \draw (2.1,6) -- (2.9,6);
        \draw (3.1,6) -- (3.9,6);
        \draw (4.1,6) -- (4.9,6);
        \draw (5.1,6) -- (5.9,6);
        \draw (6.1,6) -- (6.9,6);
        \draw (7.1,6) -- (7.9,6);
        \draw (8.1,6) -- (8.9,6);

        \draw (0.1,7) -- (0.9,7);
        \draw (1.1,7) -- (1.9,7);
        \draw (2.1,7) -- (2.9,7);
        \draw (3.1,7) -- (3.9,7);
        \draw (4.1,7) -- (4.9,7);
        \draw (5.1,7) -- (5.9,7);
        \draw (6.1,7) -- (6.9,7);
        \draw (7.1,7) -- (7.9,7);
        \draw (8.1,7) -- (8.9,7);

        \draw (0.1,8) -- (0.9,8);
        \draw (1.1,8) -- (1.9,8);
        \draw (2.1,8) -- (2.9,8);
        \draw (3.1,8) -- (3.9,8);
        \draw (4.1,8) -- (4.9,8);
        \draw (5.1,8) -- (5.9,8);
        \draw (6.1,8) -- (6.9,8);
        \draw (7.1,8) -- (7.9,8);
        \draw (8.1,8) -- (8.9,8);

        \draw (0.1,9) -- (0.9,9);
        \draw (1.1,9) -- (1.9,9);
        \draw (2.1,9) -- (2.9,9);
        \draw (3.1,9) -- (3.9,9);
        \draw (4.1,9) -- (4.9,9);
        \draw (5.1,9) -- (5.9,9);
        \draw (6.1,9) -- (6.9,9);
        \draw (7.1,9) -- (7.9,9);
        \draw (8.1,9) -- (8.9,9);

        \draw (0,0.1) -- (0,0.9);
        \draw (0,2.1) -- (0,2.9);
        \draw (0,4.1) -- (0,4.9);
        \draw (0,6.1) -- (0,6.9);
        \draw (0,8.1) -- (0,8.9);

        \draw (1,1.1) -- (1,1.9);
        \draw (1,3.1) -- (1,3.9);
        \draw (1,5.1) -- (1,5.9);
        \draw (1,7.1) -- (1,7.9);

        \draw (2,0.1) -- (2,0.9);
        \draw (2,2.1) -- (2,2.9);
        \draw (2,4.1) -- (2,4.9);
        \draw (2,6.1) -- (2,6.9);
        \draw (2,8.1) -- (2,8.9);

        \draw (3,1.1) -- (3,1.9);
        \draw (3,3.1) -- (3,3.9);
        \draw (3,5.1) -- (3,5.9);
        \draw (3,7.1) -- (3,7.9);

        \draw (4,0.1) -- (4,0.9);
        \draw (4,2.1) -- (4,2.9);
        \draw (4,4.1) -- (4,4.9);
        \draw (4,6.1) -- (4,6.9);
        \draw (4,8.1) -- (4,8.9);

        \draw (5,1.1) -- (5,1.9);
        \draw (5,3.1) -- (5,3.9);
        \draw (5,5.1) -- (5,5.9);
        \draw (5,7.1) -- (5,7.9);

        \draw (6,0.1) -- (6,0.9);
        \draw (6,2.1) -- (6,2.9);
        \draw (6,4.1) -- (6,4.9);
        \draw (6,6.1) -- (6,6.9);
        \draw (6,8.1) -- (6,8.9);

        \draw (7,1.1) -- (7,1.9);
        \draw (7,3.1) -- (7,3.9);
        \draw (7,5.1) -- (7,5.9);
        \draw (7,7.1) -- (7,7.9);

        \draw (8,0.1) -- (8,0.9);
        \draw (8,2.1) -- (8,2.9);
        \draw (8,4.1) -- (8,4.9);
        \draw (8,6.1) -- (8,6.9);
        \draw (8,8.1) -- (8,8.9);

        \draw (9,1.1) -- (9,1.9);
        \draw (9,3.1) -- (9,3.9);
        \draw (9,5.1) -- (9,5.9);
        \draw (9,7.1) -- (9,7.9);
    \end{tikzpicture}
    \caption{An optimal local locating-dominating code which is also an optimal covering code in the hexagonal grid.}
    \label{Optimal hexagonal grid}
    \end{subfigure}\vspace*{-2mm}
    \caption{Local location-domination in the square and hexagonal grids.}\vspace*{-5mm}
\end{figure}
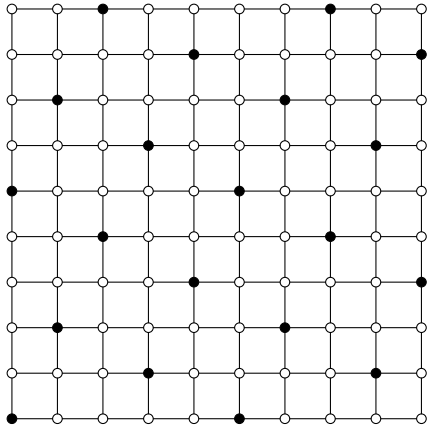
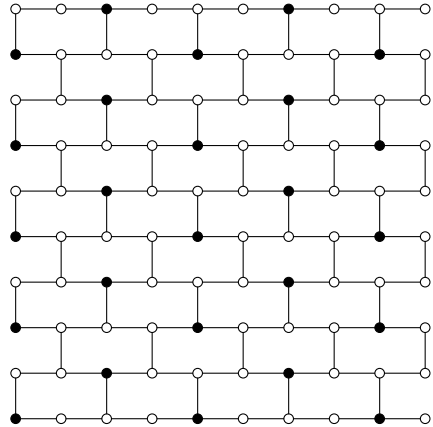

\begin{theorem}
$$
\gamma^{L-LD}(\mathcal{S}) = \frac{1}{5}
$$

and
$$
\gamma^{L-LD}(\mathcal{H}) = \frac{1}{4}.
$$
\end{theorem}
The following two theorems give the exact values for the densities of optimal local identifying codes in the square and the hexagonal grids. Observe that those constructions are clearly optimal, since each vertex has exactly one codeword in its closed neighbourhood.

\begin{theorem}
$$
\gamma^{L-ID}(\mathcal{S}) = \frac{3}{11}.
$$
\end{theorem}

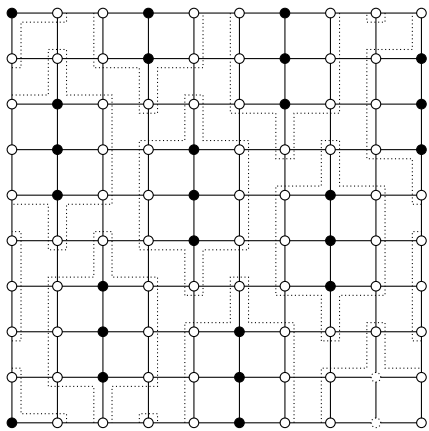
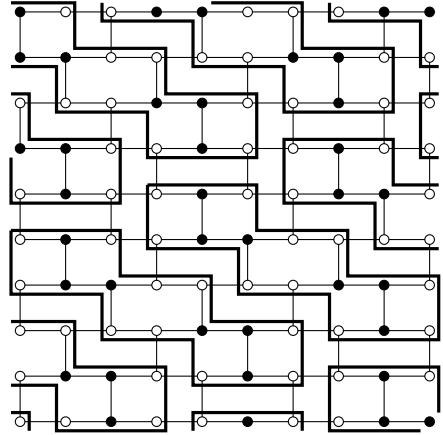
\begin{figure}[ht]
    \centering
    \begin{subfigure}[b]{0.4\textwidth}
      \begin{tikzpicture}[scale=0.6]
        \draw[fill=black] (0,0) circle(3pt);
        \draw[] (0,1) circle(3pt);
        \draw[] (0,2) circle(3pt);
        \draw[] (0,3) circle(3pt);
        \draw[] (0,4) circle(3pt);
        \draw[] (0,5) circle(3pt);
        \draw[] (0,6) circle(3pt);
        \draw[] (0,7) circle(3pt);
        \draw[] (0,8) circle(3pt);
        \draw[fill=black] (0,9) circle(3pt);

        \draw[densely dotted] (1.2,0) -- (1.2,0.2) -- (0.2,0.2) -- (0.2,1.2) -- (0,1.2);

        \draw[densely dotted] (0,1.8) -- (0.2,1.8) -- (0.2,4.2) -- (0,4.2);

        \draw[densely dotted] (0,7.8) -- (0.2,7.8) -- (0.2,8.8) -- (1.2,8.8) -- (1.2,9);

        \draw[] (1,0) circle(3pt);
        \draw[] (1,1) circle(3pt);
        \draw[] (1,2) circle(3pt);
        \draw[] (1,3) circle(3pt);
        \draw[] (1,4) circle(3pt);
        \draw[fill=black] (1,5) circle(3pt);
        \draw[fill=black] (1,6) circle(3pt);
        \draw[fill=black] (1,7) circle(3pt);
        \draw[] (1,8) circle(3pt);
        \draw[] (1,9) circle(3pt);

        \draw[densely dotted] (0,7.2) -- (0.8,7.2) -- (0.8,8.2) -- (1.2,8.2) -- (1.2,7.2) -- (2.2,7.2) -- (2.2,4.8) -- (1.2,4.8) -- (1.2,3.8) -- (0.8,3.8) -- (0.8,4.8) -- (0,4.8);

        \draw[] (2,0) circle(3pt);
        \draw[fill=black] (2,1) circle(3pt);
        \draw[fill=black] (2,2) circle(3pt);
        \draw[fill=black] (2,3) circle(3pt);
        \draw[] (2,4) circle(3pt);
        \draw[] (2,5) circle(3pt);
        \draw[] (2,6) circle(3pt);
        \draw[] (2,7) circle(3pt);
        \draw[] (2,8) circle(3pt);
        \draw[] (2,9) circle(3pt);

        \draw[densely dotted] (1.8,0) -- (1.8,0.8) -- (0.8,0.8) -- (0.8,3.2) -- (1.8,3.2) -- (1.8,4.2) -- (2.2,4.2) -- (2.2,3.2) -- (3.2,3.2) -- (3.2,0.8) -- (2.2,0.8) -- (2.2,0);

        \draw[] (3,0) circle(3pt);
        \draw[] (3,1) circle(3pt);
        \draw[] (3,2) circle(3pt);
        \draw[] (3,3) circle(3pt);
        \draw[] (3,4) circle(3pt);
        \draw[] (3,5) circle(3pt);
        \draw[] (3,6) circle(3pt);
        \draw[] (3,7) circle(3pt);
        \draw[fill=black] (3,8) circle(3pt);
        \draw[fill=black] (3,9) circle(3pt);

        \draw[densely dotted] (1.8,9) -- (1.8,7.8) -- (2.8,7.8) -- (2.8,6.8) -- (3.2,6.8) -- (3.2,7.8) -- (4.2,7.8) -- (4.2,9);

        \draw[densely dotted] (2.8,0) -- (2.8,0.2) -- (3.2,0.2) -- (3.2,0);

        \draw[] (4,0) circle(3pt);
        \draw[] (4,1) circle(3pt);
        \draw[] (4,2) circle(3pt);
        \draw[] (4,3) circle(3pt);
        \draw[fill=black] (4,4) circle(3pt);
        \draw[fill=black] (4,5) circle(3pt);
        \draw[fill=black] (4,6) circle(3pt);
        \draw[] (4,7) circle(3pt);
        \draw[] (4,8) circle(3pt);
        \draw[] (4,9) circle(3pt);

        \draw[densely dotted] (3.8,2.8) -- (3.8,3.8) -- (2.8,3.8) -- (2.8,6.2) -- (3.8,6.2) -- (3.8,7.2) -- (4.2,7.2) -- (4.2,6.2) -- (5.2,6.2) -- (5.2,3.8) -- (4.2,3.8) -- (4.2, 2.8) -- (3.8,2.8);

        \draw[fill=black] (5,0) circle(3pt);
        \draw[fill=black] (5,1) circle(3pt);
        \draw[fill=black] (5,2) circle(3pt);
        \draw[] (5,3) circle(3pt);
        \draw[] (5,4) circle(3pt);
        \draw[] (5,5) circle(3pt);
        \draw[] (5,6) circle(3pt);
        \draw[] (5,7) circle(3pt);
        \draw[] (5,8) circle(3pt);
        \draw[] (5,9) circle(3pt);

        \draw[densely dotted] (3.8,0) -- (3.8,2.2) -- (4.8,2.2) -- (4.8,3.2) -- (5.2,3.2) -- (5.2, 2.2) -- (6.2,2.2) -- (6.2,0);

        \draw[] (6,0) circle(3pt);
        \draw[] (6,1) circle(3pt);
        \draw[] (6,2) circle(3pt);
        \draw[] (6,3) circle(3pt);
        \draw[] (6,4) circle(3pt);
        \draw[] (6,5) circle(3pt);
        \draw[] (6,6) circle(3pt);
        \draw[fill=black] (6,7) circle(3pt);
        \draw[fill=black] (6,8) circle(3pt);
        \draw[fill=black] (6,9) circle(3pt);

        \draw[densely dotted] (7.2,9) -- (7.2,6.8) -- (6.2,6.8) -- (6.2,5.8) -- (5.8,5.8) -- (5.8,6.8) -- (4.8,6.8) -- (4.8,9);

        \draw[] (7,0) circle(3pt);
        \draw[] (7,1) circle(3pt);
        \draw[] (7,2) circle(3pt);
        \draw[fill=black] (7,3) circle(3pt);
        \draw[fill=black] (7,4) circle(3pt);
        \draw[fill=black] (7,5) circle(3pt);
        \draw[] (7,6) circle(3pt);
        \draw[] (7,7) circle(3pt);
        \draw[] (7,8) circle(3pt);
        \draw[] (7,9) circle(3pt);

        \draw[densely dotted] (6.8,1.8) -- (6.8,2.8) -- (5.8,2.8) -- (5.8,5
        .2) -- (6.8,5.2) -- (6.8,6.2) -- (7.2,6.2) -- (7.2,5.2) -- (8.2,5.2) -- (8.2,2.8) -- (7.2,2.8) -- (7.2, 1.8) -- (6.8,1.8);

        \draw[densely dotted] (8,0) circle(3pt);
        \draw[densely dotted] (8,1) circle(3pt);
        \draw[] (8,2) circle(3pt);
        \draw[] (8,3) circle(3pt);
        \draw[] (8,4) circle(3pt);
        \draw[] (8,5) circle(3pt);
        \draw[] (8,6) circle(3pt);
        \draw[] (8,7) circle(3pt);
        \draw[] (8,8) circle(3pt);
        \draw[] (8,9) circle(3pt);

        \draw[densely dotted] (6.8,0) -- (6.8,1.2) -- (7.8,1.2) -- (7.8,2.2) -- (8.2,2.2) -- (8.2,1.2) -- (9,1.2);

        \draw[densely dotted] (8.2,9) -- (8.2,8.8) -- (7.8,8.8) -- (7.8,9);

        \draw[] (9,0) circle(3pt);
        \draw[] (9,1) circle(3pt);
        \draw[] (9,2) circle(3pt);
        \draw[] (9,3) circle(3pt);
        \draw[] (9,4) circle(3pt);
        \draw[] (9,5) circle(3pt);
        \draw[fill=black] (9,6) circle(3pt);
        \draw[fill=black] (9,7) circle(3pt);
        \draw[fill=black] (9,8) circle(3pt);
        \draw[] (9,9) circle(3pt);

        \draw[densely dotted] (9,1.8) -- (8.8,1.8) -- (8.8,4.2) -- (9,4.2);

        \draw[densely dotted] (9,4.8) -- (8.8,4.8) -- (8.8,5.8) -- (7.8,5.8) --(7.8,8.2) -- (8.8,8.2) -- (8.8,9);

        \draw (0.1,0) -- (0.9,0);
        \draw (1.1,0) -- (1.9,0);
        \draw (2.1,0) -- (2.9,0);
        \draw (3.1,0) -- (3.9,0);
        \draw (4.1,0) -- (4.9,0);
        \draw (5.1,0) -- (5.9,0);
        \draw (6.1,0) -- (6.9,0);
        \draw (7.1,0) -- (7.9,0);
        \draw (8.1,0) -- (8.9,0);

        \draw (0.1,1) -- (0.9,1);
        \draw (1.1,1) -- (1.9,1);
        \draw (2.1,1) -- (2.9,1);
        \draw (3.1,1) -- (3.9,1);
        \draw (4.1,1) -- (4.9,1);
        \draw (5.1,1) -- (5.9,1);
        \draw (6.1,1) -- (6.9,1);
        \draw (7.1,1) -- (7.9,1);
        \draw (8.1,1) -- (8.9,1);

        \draw (0.1,2) -- (0.9,2);
        \draw (1.1,2) -- (1.9,2);
        \draw (2.1,2) -- (2.9,2);
        \draw (3.1,2) -- (3.9,2);
        \draw (4.1,2) -- (4.9,2);
        \draw (5.1,2) -- (5.9,2);
        \draw (6.1,2) -- (6.9,2);
        \draw (7.1,2) -- (7.9,2);
        \draw (8.1,2) -- (8.9,2);

        \draw (0.1,3) -- (0.9,3);
        \draw (1.1,3) -- (1.9,3);
        \draw (2.1,3) -- (2.9,3);
        \draw (3.1,3) -- (3.9,3);
        \draw (4.1,3) -- (4.9,3);
        \draw (5.1,3) -- (5.9,3);
        \draw (6.1,3) -- (6.9,3);
        \draw (7.1,3) -- (7.9,3);
        \draw (8.1,3) -- (8.9,3);

        \draw (0.1,4) -- (0.9,4);
        \draw (1.1,4) -- (1.9,4);
        \draw (2.1,4) -- (2.9,4);
        \draw (3.1,4) -- (3.9,4);
        \draw (4.1,4) -- (4.9,4);
        \draw (5.1,4) -- (5.9,4);
        \draw (6.1,4) -- (6.9,4);
        \draw (7.1,4) -- (7.9,4);
        \draw (8.1,4) -- (8.9,4);

        \draw (0.1,5) -- (0.9,5);
        \draw (1.1,5) -- (1.9,5);
        \draw (2.1,5) -- (2.9,5);
        \draw (3.1,5) -- (3.9,5);
        \draw (4.1,5) -- (4.9,5);
        \draw (5.1,5) -- (5.9,5);
        \draw (6.1,5) -- (6.9,5);
        \draw (7.1,5) -- (7.9,5);
        \draw (8.1,5) -- (8.9,5);

        \draw (0.1,6) -- (0.9,6);
        \draw (1.1,6) -- (1.9,6);
        \draw (2.1,6) -- (2.9,6);
        \draw (3.1,6) -- (3.9,6);
        \draw (4.1,6) -- (4.9,6);
        \draw (5.1,6) -- (5.9,6);
        \draw (6.1,6) -- (6.9,6);
        \draw (7.1,6) -- (7.9,6);
        \draw (8.1,6) -- (8.9,6);

        \draw (0.1,7) -- (0.9,7);
        \draw (1.1,7) -- (1.9,7);
        \draw (2.1,7) -- (2.9,7);
        \draw (3.1,7) -- (3.9,7);
        \draw (4.1,7) -- (4.9,7);
        \draw (5.1,7) -- (5.9,7);
        \draw (6.1,7) -- (6.9,7);
        \draw (7.1,7) -- (7.9,7);
        \draw (8.1,7) -- (8.9,7);

        \draw (0.1,8) -- (0.9,8);
        \draw (1.1,8) -- (1.9,8);
        \draw (2.1,8) -- (2.9,8);
        \draw (3.1,8) -- (3.9,8);
        \draw (4.1,8) -- (4.9,8);
        \draw (5.1,8) -- (5.9,8);
        \draw (6.1,8) -- (6.9,8);
        \draw (7.1,8) -- (7.9,8);
        \draw (8.1,8) -- (8.9,8);

        \draw (0.1,9) -- (0.9,9);
        \draw (1.1,9) -- (1.9,9);
        \draw (2.1,9) -- (2.9,9);
        \draw (3.1,9) -- (3.9,9);
        \draw (4.1,9) -- (4.9,9);
        \draw (5.1,9) -- (5.9,9);
        \draw (6.1,9) -- (6.9,9);
        \draw (7.1,9) -- (7.9,9);
        \draw (8.1,9) -- (8.9,9);

        \draw (0,0.1) -- (0,0.9);
        \draw (0,1.1) -- (0,1.9);
        \draw (0,2.1) -- (0,2.9);
        \draw (0,3.1) -- (0,3.9);
        \draw (0,4.1) -- (0,4.9);
        \draw (0,5.1) -- (0,5.9);
        \draw (0,6.1) -- (0,6.9);
        \draw (0,7.1) -- (0,7.9);
        \draw (0,8.1) -- (0,8.9);

        \draw (1,0.1) -- (1,0.9);
        \draw (1,1.1) -- (1,1.9);
        \draw (1,2.1) -- (1,2.9);
        \draw (1,3.1) -- (1,3.9);
        \draw (1,4.1) -- (1,4.9);
        \draw (1,5.1) -- (1,5.9);
        \draw (1,6.1) -- (1,6.9);
        \draw (1,7.1) -- (1,7.9);
        \draw (1,8.1) -- (1,8.9);

        \draw (2,0.1) -- (2,0.9);
        \draw (2,1.1) -- (2,1.9);
        \draw (2,2.1) -- (2,2.9);
        \draw (2,3.1) -- (2,3.9);
        \draw (2,4.1) -- (2,4.9);
        \draw (2,5.1) -- (2,5.9);
        \draw (2,6.1) -- (2,6.9);
        \draw (2,7.1) -- (2,7.9);
        \draw (2,8.1) -- (2,8.9);

        \draw (3,0.1) -- (3,0.9);
        \draw (3,1.1) -- (3,1.9);
        \draw (3,2.1) -- (3,2.9);
        \draw (3,3.1) -- (3,3.9);
        \draw (3,4.1) -- (3,4.9);
        \draw (3,5.1) -- (3,5.9);
        \draw (3,6.1) -- (3,6.9);
        \draw (3,7.1) -- (3,7.9);
        \draw (3,8.1) -- (3,8.9);

        \draw (4,0.1) -- (4,0.9);
        \draw (4,1.1) -- (4,1.9);
        \draw (4,2.1) -- (4,2.9);
        \draw (4,3.1) -- (4,3.9);
        \draw (4,4.1) -- (4,4.9);
        \draw (4,5.1) -- (4,5.9);
        \draw (4,6.1) -- (4,6.9);
        \draw (4,7.1) -- (4,7.9);
        \draw (4,8.1) -- (4,8.9);

        \draw (5,0.1) -- (5,0.9);
        \draw (5,1.1) -- (5,1.9);
        \draw (5,2.1) -- (5,2.9);
        \draw (5,3.1) -- (5,3.9);
        \draw (5,4.1) -- (5,4.9);
        \draw (5,5.1) -- (5,5.9);
        \draw (5,6.1) -- (5,6.9);
        \draw (5,7.1) -- (5,7.9);
        \draw (5,8.1) -- (5,8.9);

        \draw (6,0.1) -- (6,0.9);
        \draw (6,1.1) -- (6,1.9);
        \draw (6,2.1) -- (6,2.9);
        \draw (6,3.1) -- (6,3.9);
        \draw (6,4.1) -- (6,4.9);
        \draw (6,5.1) -- (6,5.9);
        \draw (6,6.1) -- (6,6.9);
        \draw (6,7.1) -- (6,7.9);
        \draw (6,8.1) -- (6,8.9);

        \draw (7,0.1) -- (7,0.9);
        \draw (7,1.1) -- (7,1.9);
        \draw (7,2.1) -- (7,2.9);
        \draw (7,3.1) -- (7,3.9);
        \draw (7,4.1) -- (7,4.9);
        \draw (7,5.1) -- (7,5.9);
        \draw (7,6.1) -- (7,6.9);
        \draw (7,7.1) -- (7,7.9);
        \draw (7,8.1) -- (7,8.9);

        \draw (8,0.1) -- (8,0.9);
        \draw (8,1.1) -- (8,1.9);
        \draw (8,2.1) -- (8,2.9);
        \draw (8,3.1) -- (8,3.9);
        \draw (8,4.1) -- (8,4.9);
        \draw (8,5.1) -- (8,5.9);
        \draw (8,6.1) -- (8,6.9);
        \draw (8,7.1) -- (8,7.9);
        \draw (8,8.1) -- (8,8.9);

        \draw (9,0.1) -- (9,0.9);
        \draw (9,1.1) -- (9,1.9);
        \draw (9,2.1) -- (9,2.9);
        \draw (9,3.1) -- (9,3.9);
        \draw (9,4.1) -- (9,4.9);
        \draw (9,5.1) -- (9,5.9);
        \draw (9,6.1) -- (9,6.9);
        \draw (9,7.1) -- (9,7.9);
        \draw (9,8.1) -- (9,8.9);
    \end{tikzpicture}
      \caption{A local identifying code of density $\frac{3}{11}$ in the square grid.}
       \label{Construction in square grid}
    \end{subfigure}\qquad\qquad
    \begin{subfigure}[b]{0.4\textwidth}
       \begin{tikzpicture}[scale=0.6]
        \draw[] (0,0) circle(3pt);
        \draw[] (0,1) circle(3pt);
        \draw[] (0,2) circle(3pt);
        \draw[] (0,3) circle(3pt);
        \draw[] (0,4) circle(3pt);
        \draw[] (0,5) circle(3pt);
        \draw[fill=black] (0,6) circle(3pt);
        \draw[] (0,7) circle(3pt);
        \draw[fill=black] (0,8) circle(3pt);
        \draw[fill=black] (0,9) circle(3pt);

        \draw[] (1,0) circle(3pt);
        \draw[fill=black] (1,1) circle(3pt);
        \draw[] (1,2) circle(3pt);
        \draw[fill=black] (1,3) circle(3pt);
        \draw[fill=black] (1,4) circle(3pt);
        \draw[fill=black] (1,5) circle(3pt);
        \draw[fill=black] (1,6) circle(3pt);
        \draw[] (1,7) circle(3pt);
        \draw[fill=black] (1,8) circle(3pt);
        \draw[] (1,9) circle(3pt);

        \draw[fill=black] (2,0) circle(3pt);
        \draw[fill=black] (2,1) circle(3pt);
        \draw[] (2,2) circle(3pt);
        \draw[fill=black] (2,3) circle(3pt);
        \draw[] (2,4) circle(3pt);
        \draw[] (2,5) circle(3pt);
        \draw[] (2,6) circle(3pt);
        \draw[] (2,7) circle(3pt);
        \draw[] (2,8) circle(3pt);
        \draw[] (2,9) circle(3pt);

        \draw[] (3,0) circle(3pt);
        \draw[] (3,1) circle(3pt);
        \draw[] (3,2) circle(3pt);
        \draw[] (3,3) circle(3pt);
        \draw[] (3,4) circle(3pt);
        \draw[] (3,5) circle(3pt);
        \draw[] (3,6) circle(3pt);
        \draw[fill=black] (3,7) circle(3pt);
        \draw[] (3,8) circle(3pt);
        \draw[fill=black] (3,9) circle(3pt);

        \draw[] (4,0) circle(3pt);
        \draw[] (4,1) circle(3pt);
        \draw[fill=black] (4,2) circle(3pt);
        \draw[] (4,3) circle(3pt);
        \draw[fill=black] (4,4) circle(3pt);
        \draw[fill=black] (4,5) circle(3pt);
        \draw[fill=black] (4,6) circle(3pt);
        \draw[fill=black] (4,7) circle(3pt);
        \draw[] (4,8) circle(3pt);
        \draw[fill=black] (4,9) circle(3pt);

        \draw[fill=black] (5,0) circle(3pt);
        \draw[fill=black] (5,1) circle(3pt);
        \draw[fill=black] (5,2) circle(3pt);
        \draw[] (5,3) circle(3pt);
        \draw[fill=black] (5,4) circle(3pt);
        \draw[] (5,5) circle(3pt);
        \draw[] (5,6) circle(3pt);
        \draw[] (5,7) circle(3pt);
        \draw[] (5,8) circle(3pt);
        \draw[] (5,9) circle(3pt);

        \draw[] (6,0) circle(3pt);
        \draw[] (6,1) circle(3pt);
        \draw[] (6,2) circle(3pt);
        \draw[] (6,3) circle(3pt);
        \draw[] (6,4) circle(3pt);
        \draw[] (6,5) circle(3pt);
        \draw[] (6,6) circle(3pt);
        \draw[] (6,7) circle(3pt);
        \draw[fill=black] (6,8) circle(3pt);
        \draw[] (6,9) circle(3pt);

        \draw[] (7,0) circle(3pt);
        \draw[] (7,1) circle(3pt);
        \draw[] (7,2) circle(3pt);
        \draw[fill=black] (7,3) circle(3pt);
        \draw[] (7,4) circle(3pt);
        \draw[fill=black] (7,5) circle(3pt);
        \draw[fill=black] (7,6) circle(3pt);
        \draw[fill=black] (7,7) circle(3pt);
        \draw[fill=black] (7,8) circle(3pt);
        \draw[] (7,9) circle(3pt);

        \draw[fill=black] (8,0) circle(3pt);
        \draw[fill=black] (8,1) circle(3pt);
        \draw[fill=black] (8,2) circle(3pt);
        \draw[fill=black] (8,3) circle(3pt);
        \draw[] (8,4) circle(3pt);
        \draw[fill=black] (8,5) circle(3pt);
        \draw[] (8,6) circle(3pt);
        \draw[] (8,7) circle(3pt);
        \draw[] (8,8) circle(3pt);
        \draw[fill=black] (8,9) circle(3pt);

        \draw[fill=black] (9,0) circle(3pt);
        \draw[] (9,1) circle(3pt);
        \draw[] (9,2) circle(3pt);
        \draw[] (9,3) circle(3pt);
        \draw[] (9,4) circle(3pt);
        \draw[] (9,5) circle(3pt);
        \draw[] (9,6) circle(3pt);
        \draw[] (9,7) circle(3pt);
        \draw[] (9,8) circle(3pt);
        \draw[fill=black] (9,9) circle(3pt);

        \draw (0.1,0) -- (0.9,0);
        \draw (1.1,0) -- (1.9,0);
        \draw (2.1,0) -- (2.9,0);
        \draw (3.1,0) -- (3.9,0);
        \draw (4.1,0) -- (4.9,0);
        \draw (5.1,0) -- (5.9,0);
        \draw (6.1,0) -- (6.9,0);
        \draw (7.1,0) -- (7.9,0);
        \draw (8.1,0) -- (8.9,0);

        \draw (0.1,1) -- (0.9,1);
        \draw (1.1,1) -- (1.9,1);
        \draw (2.1,1) -- (2.9,1);
        \draw (3.1,1) -- (3.9,1);
        \draw (4.1,1) -- (4.9,1);
        \draw (5.1,1) -- (5.9,1);
        \draw (6.1,1) -- (6.9,1);
        \draw (7.1,1) -- (7.9,1);
        \draw (8.1,1) -- (8.9,1);

        \draw (0.1,2) -- (0.9,2);
        \draw (1.1,2) -- (1.9,2);
        \draw (2.1,2) -- (2.9,2);
        \draw (3.1,2) -- (3.9,2);
        \draw (4.1,2) -- (4.9,2);
        \draw (5.1,2) -- (5.9,2);
        \draw (6.1,2) -- (6.9,2);
        \draw (7.1,2) -- (7.9,2);
        \draw (8.1,2) -- (8.9,2);

        \draw (0.1,3) -- (0.9,3);
        \draw (1.1,3) -- (1.9,3);
        \draw (2.1,3) -- (2.9,3);
        \draw (3.1,3) -- (3.9,3);
        \draw (4.1,3) -- (4.9,3);
        \draw (5.1,3) -- (5.9,3);
        \draw (6.1,3) -- (6.9,3);
        \draw (7.1,3) -- (7.9,3);
        \draw (8.1,3) -- (8.9,3);

        \draw (0.1,4) -- (0.9,4);
        \draw (1.1,4) -- (1.9,4);
        \draw (2.1,4) -- (2.9,4);
        \draw (3.1,4) -- (3.9,4);
        \draw (4.1,4) -- (4.9,4);
        \draw (5.1,4) -- (5.9,4);
        \draw (6.1,4) -- (6.9,4);
        \draw (7.1,4) -- (7.9,4);
        \draw (8.1,4) -- (8.9,4);

        \draw (0.1,5) -- (0.9,5);
        \draw (1.1,5) -- (1.9,5);
        \draw (2.1,5) -- (2.9,5);
        \draw (3.1,5) -- (3.9,5);
        \draw (4.1,5) -- (4.9,5);
        \draw (5.1,5) -- (5.9,5);
        \draw (6.1,5) -- (6.9,5);
        \draw (7.1,5) -- (7.9,5);
        \draw (8.1,5) -- (8.9,5);

        \draw (0.1,6) -- (0.9,6);
        \draw (1.1,6) -- (1.9,6);
        \draw (2.1,6) -- (2.9,6);
        \draw (3.1,6) -- (3.9,6);
        \draw (4.1,6) -- (4.9,6);
        \draw (5.1,6) -- (5.9,6);
        \draw (6.1,6) -- (6.9,6);
        \draw (7.1,6) -- (7.9,6);
        \draw (8.1,6) -- (8.9,6);

        \draw (0.1,7) -- (0.9,7);
        \draw (1.1,7) -- (1.9,7);
        \draw (2.1,7) -- (2.9,7);
        \draw (3.1,7) -- (3.9,7);
        \draw (4.1,7) -- (4.9,7);
        \draw (5.1,7) -- (5.9,7);
        \draw (6.1,7) -- (6.9,7);
        \draw (7.1,7) -- (7.9,7);
        \draw (8.1,7) -- (8.9,7);

        \draw (0.1,8) -- (0.9,8);
        \draw (1.1,8) -- (1.9,8);
        \draw (2.1,8) -- (2.9,8);
        \draw (3.1,8) -- (3.9,8);
        \draw (4.1,8) -- (4.9,8);
        \draw (5.1,8) -- (5.9,8);
        \draw (6.1,8) -- (6.9,8);
        \draw (7.1,8) -- (7.9,8);
        \draw (8.1,8) -- (8.9,8);

        \draw (0.1,9) -- (0.9,9);
        \draw (1.1,9) -- (1.9,9);
        \draw (2.1,9) -- (2.9,9);
        \draw (3.1,9) -- (3.9,9);
        \draw (4.1,9) -- (4.9,9);
        \draw (5.1,9) -- (5.9,9);
        \draw (6.1,9) -- (6.9,9);
        \draw (7.1,9) -- (7.9,9);
        \draw (8.1,9) -- (8.9,9);

        \draw (0,0.1) -- (0,0.9);
        \draw (0,2.1) -- (0,2.9);
        \draw (0,4.1) -- (0,4.9);
        \draw (0,6.1) -- (0,6.9);
        \draw (0,8.1) -- (0,8.9);

        \draw (1,1.1) -- (1,1.9);
        \draw (1,3.1) -- (1,3.9);
        \draw (1,5.1) -- (1,5.9);
        \draw (1,7.1) -- (1,7.9);

        \draw (2,0.1) -- (2,0.9);
        \draw (2,2.1) -- (2,2.9);
        \draw (2,4.1) -- (2,4.9);
        \draw (2,6.1) -- (2,6.9);
        \draw (2,8.1) -- (2,8.9);

        \draw (3,1.1) -- (3,1.9);
        \draw (3,3.1) -- (3,3.9);
        \draw (3,5.1) -- (3,5.9);
        \draw (3,7.1) -- (3,7.9);

        \draw (4,0.1) -- (4,0.9);
        \draw (4,2.1) -- (4,2.9);
        \draw (4,4.1) -- (4,4.9);
        \draw (4,6.1) -- (4,6.9);
        \draw (4,8.1) -- (4,8.9);

        \draw (5,1.1) -- (5,1.9);
        \draw (5,3.1) -- (5,3.9);
        \draw (5,5.1) -- (5,5.9);
        \draw (5,7.1) -- (5,7.9);

        \draw (6,0.1) -- (6,0.9);
        \draw (6,2.1) -- (6,2.9);
        \draw (6,4.1) -- (6,4.9);
        \draw (6,6.1) -- (6,6.9);
        \draw (6,8.1) -- (6,8.9);

        \draw (7,1.1) -- (7,1.9);
        \draw (7,3.1) -- (7,3.9);
        \draw (7,5.1) -- (7,5.9);
        \draw (7,7.1) -- (7,7.9);

        \draw (8,0.1) -- (8,0.9);
        \draw (8,2.1) -- (8,2.9);
        \draw (8,4.1) -- (8,4.9);
        \draw (8,6.1) -- (8,6.9);
        \draw (8,8.1) -- (8,8.9);

        \draw (9,1.1) -- (9,1.9);
        \draw (9,3.1) -- (9,3.9);
        \draw (9,5.1) -- (9,5.9);
        \draw (9,7.1) -- (9,7.9);

        \draw[very thick] (-0.2,0.2) -- (0.2,0.2) -- (0.2,-0.2);

        \draw[very thick] (-0.2,0.8) -- (0.8,0.8) -- (0.8,-0.2) -- (3.2,-0.2) -- (3.2,1.2) -- (1.2,1.2) -- (1.2,2.2) -- (-0.2,2.2);

        \draw[very thick] (-0.2,4.2) -- (-0.2,2.8) -- (1.8,2.8) -- (1.8,1.8) -- (3.8,1.8) -- (3.8,0.8) -- (6.2,0.8) -- (6.2,2.2) -- (4.2,2.2) -- (4.2,3.2) -- (2.2,3.2) -- (2.2,4.2) -- (-0.2,4.2);

        \draw[very thick] (-0.2+3,4.2+1) -- (-0.2+3,2.8+1) -- (1.8+3,2.8+1) -- (1.8+3,1.8+1) -- (3.8+3,1.8+1) -- (3.8+3,0.8+1) -- (6.2+3,0.8+1) -- (6.2+3,2.2+1) -- (4.2+3,2.2+1) -- (4.2+3,3.2+1) -- (2.2+3,3.2+1) -- (2.2+3,4.2+1) -- (-0.2+3,4.2+1);

        \draw[very thick]
        (9.2,5.2) -- (8.2,5.2) -- (8.2,6.2) --
        (-0.2+6,4.2+2) -- (-0.2+6,2.8+2) -- (1.8+6,2.8+2) -- (1.8+6,1.8+2)
        -- (9.2,3.8);

        \draw[very thick] (-0.2,2.8+5) -- (1.8-1,2.8+5) -- (1.8-1,1.8+5) -- (3.8-1,1.8+5) -- (3.8-1,0.8+5) -- (6.2-1,0.8+5) -- (6.2-1,2.2+5) -- (4.2-1,2.2+5) -- (4.2-1,3.2+5) -- (2.2-1,3.2+5) -- (2.2-1,4.2+5) -- (-0.2,4.2+5);

        \draw[very thick] (-0.2,5.8) --(-0.2,4.8) -- (2.2,4.8) -- (2.2,6.2) -- (0.2,6.2) -- (0.2,7.2) -- (-0.2,7.2);

        \draw[very thick] (3.8,-0.2) -- (3.8,0.2) -- (6.2,0.2) -- (6.2,-0.2);

        \draw[very thick] (9.2,0.2) -- (9.2,1.2) -- (6.8,1.2) -- (6.8,-0.2) -- (8.8,-0.2);

        \draw[very thick] (1.8,9.2) -- (1.8,8.8) -- (3.8,8.8) -- (3.8,7.8) -- (5.8,7.8) -- (5.8,6.8) -- (8.2,6.8) -- (8.2,8.2) -- (6.2,8.2) -- (6.2,9.2) -- (4.2,9.2);

        \draw[very thick] (6.8,9.2) -- (6.8,8.8) -- (8.8,8.8)-- (8.8,7.8) -- (9.2,7.8);

        \draw[very thick] (9.2,7.2) --(8.8,7.2) -- (8.8,5.8) -- (9.2,5.8);
    \end{tikzpicture}
      \caption{A local identifying code of density $\frac{3}{8}$ in the hexagonal grid.}
      \label{Construction LID hexagonal grid}
    \end{subfigure}\vspace*{-2mm}
    \caption{Local identifying codes in the square and hexagonal grids.}\vspace*{-3mm}
    \end{figure}

\begin{proof}
    By a construction, in Figure \ref{Construction in square grid}, of a local identifying code in the square grid of density $\frac{3}{11}$, we have $\gamma^{L-ID}(\mathcal{S})  \leq \frac{3}{11}$.
    Next, we prove that $\gamma^{L-ID}(\mathcal{S}) \geq \frac{3}{11}$ using a share shifting scheme.

\medskip
    Let $C$ be a local identifying code in the square grid.
    In our share shifting scheme we shift $1/6$ share units from a codeword $\bc \in C$ to its unique codeword neighbour if $|I(\bc)| = 2$.
    In all the other cases no share is shifted.
    Let us denote by $s'$ the modified share function after applying the introduced scheme.
    We claim that $s'(\bc) \leq \frac{11}{3}$ for all $\bc \in C$ which yields by Lemma \ref{Lemma shareshift} that $D(C) \geq \frac{3}{11}$ and hence $\gamma^{L-ID}(\mathcal{S}) \geq \frac{3}{11}$.
    So, let $\bc \in C$ be an arbitrary codeword of $C$.

    \medskip
     Assume first that $|I(\bc)| = 1$, {\it i.e.}, that $I(\bc) = \{ \bc \}$. Every neighbour of $\bc$ is covered by at least 2 codewords since otherwise the code $C$ would not separate $\bc$ from all of its neighbours.
    So, in this case we have
    $
    s(\bc) \leq 1 + 4 \cdot \frac{1}{2} = 3 < \frac{11}{3}.
    $
    Since $\bc$ has no codeword neighbours, no share is shifted to $\bc$ and hence
    $$
    s'(\bc) = s(\bc) = 3 < \frac{11}{3}.
    $$

    Assume then that $|I(\bc)| = 2$ and let $\bc'$ be the unique codeword neighbour of $\bc$.
    Since $C$ separates $\bc$ and $\bc'$, we have $|I(\bc')| \geq 3$ and hence
    $s(\bc) \leq \frac{1}{2} + \frac{1}{3} + 3 \cdot 1 = \frac{23}{6} = 3 \frac{5}{6}$.
    Next, we shift $1/6$ share units from $\bc$ and no share is shifted to $\bc$ because $\bc$ has no codeword neighbours with exactly one codeword neighbour.
    So, we have
    $$
    s'(\bc) \leq \frac{23}{6} - \frac{1}{6} = 3 \frac{4}{6} = \frac{11}{3}.
    $$

    Finally, assume that $|I(\bc)| \geq 3$.
    If $|I(\bc)| = 3$, then $s(\bc) \leq \frac{1}{3} + 2 \cdot \frac{1}{2} + 2 \cdot 1 = \frac{10}{3} = 3 \frac{1}{3}$ and hence
    $s'(\bc) \leq \frac{10}{3} + 2 \cdot \frac{1}{6} = \frac{11}{3}$.
    If $|I(\bc)| \geq 4$, then
    $s(\bc) \leq \frac{1}{4} + 3 \cdot \frac{1}{2} + 1 = \frac{11}{4} = 2 \frac{3}{4}$ and hence
    $s'(\bc) \leq 2 \frac{3}{4} + 4 \cdot \frac{1}{6} < \frac{11}{3}$.

\medskip
    We have shown that $s'(\bc) \leq \frac{11}{3}$ for an arbitrary $\bc \in C$.
    The claim follows.
  \end{proof}

\begin{theorem}
$$
\gamma^{L-ID}(\mathcal{H}) = \frac{3}{8}.
$$
\end{theorem}

\begin{proof}
    By a construction in Figure \ref{Construction LID hexagonal grid} of a local identifying code in the hexagonal grid of density $\frac{3}{8}$, we have $\gamma^{L-ID}(\mathcal{H})  \leq \frac{3}{8}$.
    Next, we prove that $\gamma^{L-ID}(\mathcal{H}) \geq \frac{3}{8}$ using a share shifting scheme.

    Let $C$ be a local identifying code in the hexagonal grid.
    In our share shifting scheme we shift $1/6$ share units from a codeword $\bc \in C$ to its unique codeword neighbour if $|I(\bc)| = 2$.
    In all the other cases no share is shifted.
    Let us denote by $s'$ the modified share function after applying the introduced scheme.
    We claim that $s'(\bc) \leq \frac{8}{3}$ for all $\bc \in C$ which yields by Lemma \ref{Lemma shareshift} that $D(C) \geq \frac{3}{8}$ and hence $\gamma^{L-ID}(\mathcal{H}) \geq \frac{3}{8}$.
    So, let $\bc \in C$ be an arbitrary codeword of $C$.

\medskip
    Assume first that
    $|I(\bc)| = 1$.
    Every neighbour of $\bc$ is covered by at least two codewords since otherwise the code $C$ would not separate $\bc$ from all of its neighbours. In this case we have $s(\bc) \leq 1 + 3 \cdot \frac{1}{2} < \frac{8}{3}$.
    Since $\bc$ has no codeword neighbours, no share is shifted to $\bc$ and hence
    $$
    s'(\bc) = s(\bc) < \frac{8}{3}.
    $$

    Assume then that $|I(\bc)|  = 2$ and let $\bc'$ be the unique codeword neighbour of $\bc$. Since $C$ separates $\bc$ and $\bc'$, we have $|I(\bc')| \geq 3$ and hence
    $s(\bc) \leq \frac{1}{2} + \frac{1}{3} + 2 \cdot 1 = \frac{17}{6} = 2 \frac{5}{6}$.

   \medskip
    Now we shift $1/6$ share units from $\bc$ to $\bc'$ and clearly no share is shifted to $\bc$.
    Thus,
    $$
    s'(\bc) \leq \frac{17}{6} - \frac{1}{6} = \frac{8}{3}.
    $$

    Finally, assume that $|I(\bc)| \geq 3$.
    If $|I(\bc)| = 3$, then $s(\bc) \leq \frac{1}{3} + 2 \cdot \frac{1}{2} + 1 = 2 \frac{1}{3}$ and hence
    $s'(\bc) \leq 2 \frac{1}{3} + 2 \cdot \frac{1}{6} = \frac{8}{3}$.
    If $|I(\bc)| = 4$, then $s(\bc) \leq \frac{1}{4} + 3 \cdot \frac{1}{2} = 1 \frac{3}{4}$ and hence
    $s'(\bc) \leq 1 \frac{3}{4} + 3 \cdot \frac{1}{6} < \frac{8}{3}$. \vspace*{-3mm}
 \end{proof}

\subsection{The triangular grid}\label{subsec:triangle}

\begin{theorem}
$$
\gamma^{L-LD}(\mathcal{T}) \in \left[\frac{2}{11}, \frac{2}{9}\right].
$$
\end{theorem}

\begin{proof}
    In Figure \ref{localLDtriangular} we have constructed a local locating-dominating code of density $\frac{2}{9}$.
    Thus, $\gamma^{L-LD}(\mathcal{T}) \leq \frac{2}{9}$.
    Next, we show that $\gamma^{L-LD}(\mathcal{T}) \geq \frac{2}{11}$.
    So, let $C$ be a local locating-dominating code in the triangular grid and let $\bc \in C$ be an arbitrary codeword.
    We show that $s(\bc) \leq \frac{11}{2}$ which gives the claim together with Lemma \ref{share infinite}.

 \medskip
    Assume first that $\bc$ has a codeword neighbour.
    We have $s(\bc) \leq 4 \cdot \frac{1}{2} + 3 \cdot 1 = \frac{10}{2} < \frac{11}{2}$.
Assume then that $I(\bc) = \{ \bc \}$. Since $C$ is a local locating-dominating code, it has to separate any two non-codeword neighbours, and hence it can cover at most three neighbours of $\bc$ only by $\bc$.

\smallskip
Thus, $s(\bc) \leq 4 \cdot 1 + 3 \cdot \frac{1}{2} = \frac{11}{2}$.
\end{proof}

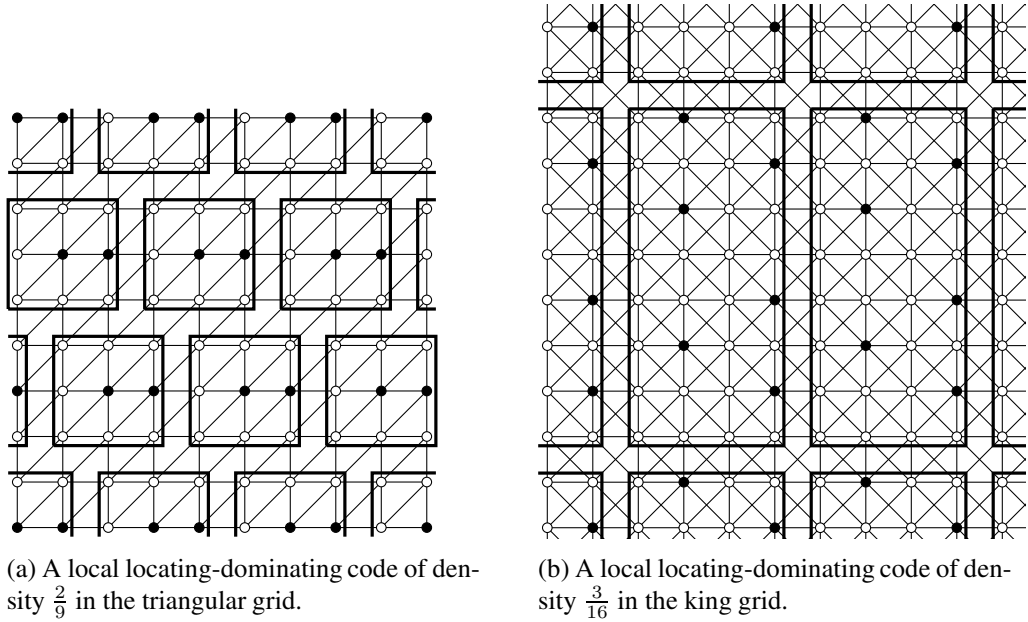
\begin{figure}[!ht]
\vspace*{2mm}
    \centering\begin{subfigure}[b]{0.4\textwidth}
        \begin{tikzpicture}[scale=0.6]
        \draw (0,0) grid (9,9);

        \draw (0,0) -- (9,9);
        \draw (1,0) -- (9,8);
        \draw (2,0) -- (9,7);
        \draw (3,0) -- (9,6);
        \draw (4,0) -- (9,5);
        \draw (5,0) -- (9,4);
        \draw (6,0) -- (9,3);
        \draw (7,0) -- (9,2);
        \draw (8,0) -- (9,1);

        \draw (0,1) -- (8,9);
        \draw (0,2) -- (7,9);
        \draw (0,3) -- (6,9);
        \draw (0,4) -- (5,9);
        \draw (0,5) -- (4,9);
        \draw (0,6) -- (3,9);
        \draw (0,7) -- (2,9);
        \draw (0,8) -- (1,9);


        \draw[fill=black] (0,0) circle(3pt);
        \draw[fill=black] (1,0) circle(3pt);
        \draw[fill=white] (2,0) circle(3pt);
        \draw[fill=black] (3,0) circle(3pt);
        \draw[fill=black] (4,0) circle(3pt);
        \draw[fill=white] (5,0) circle(3pt);
        \draw[fill=black] (6,0) circle(3pt);
        \draw[fill=black] (7,0) circle(3pt);
        \draw[fill=white] (8,0) circle(3pt);
        \draw[fill=black] (9,0) circle(3pt);

        \draw[fill=white] (0,1) circle(3pt);
        \draw[fill=white] (1,1) circle(3pt);
        \draw[fill=white] (2,1) circle(3pt);
        \draw[fill=white] (3,1) circle(3pt);
        \draw[fill=white] (4,1) circle(3pt);
        \draw[fill=white] (5,1) circle(3pt);
        \draw[fill=white] (6,1) circle(3pt);
        \draw[fill=white] (7,1) circle(3pt);
        \draw[fill=white] (8,1) circle(3pt);
        \draw[fill=white] (9,1) circle(3pt);

        \draw[fill=white] (0,2) circle(3pt);
        \draw[fill=white] (1,2) circle(3pt);
        \draw[fill=white] (2,2) circle(3pt);
        \draw[fill=white] (3,2) circle(3pt);
        \draw[fill=white] (4,2) circle(3pt);
        \draw[fill=white] (5,2) circle(3pt);
        \draw[fill=white] (6,2) circle(3pt);
        \draw[fill=white] (7,2) circle(3pt);
        \draw[fill=white] (8,2) circle(3pt);
        \draw[fill=white] (9,2) circle(3pt);

        \draw[fill=black] (0,3) circle(3pt);
        \draw[fill=white] (1,3) circle(3pt);
        \draw[fill=black] (2,3) circle(3pt);
        \draw[fill=black] (3,3) circle(3pt);
        \draw[fill=white] (4,3) circle(3pt);
        \draw[fill=black] (5,3) circle(3pt);
        \draw[fill=black] (6,3) circle(3pt);
        \draw[fill=white] (7,3) circle(3pt);
        \draw[fill=black] (8,3) circle(3pt);
        \draw[fill=black] (9,3) circle(3pt);

        \draw[fill=white] (0,4) circle(3pt);
        \draw[fill=white] (1,4) circle(3pt);
        \draw[fill=white] (2,4) circle(3pt);
        \draw[fill=white] (3,4) circle(3pt);
        \draw[fill=white] (4,4) circle(3pt);
        \draw[fill=white] (5,4) circle(3pt);
        \draw[fill=white] (6,4) circle(3pt);
        \draw[fill=white] (7,4) circle(3pt);
        \draw[fill=white] (8,4) circle(3pt);
        \draw[fill=white] (9,4) circle(3pt);

        \draw[fill=white] (0,5) circle(3pt);
        \draw[fill=white] (1,5) circle(3pt);
        \draw[fill=white] (2,5) circle(3pt);
        \draw[fill=white] (3,5) circle(3pt);
        \draw[fill=white] (4,5) circle(3pt);
        \draw[fill=white] (5,5) circle(3pt);
        \draw[fill=white] (6,5) circle(3pt);
        \draw[fill=white] (7,5) circle(3pt);
        \draw[fill=white] (8,5) circle(3pt);
        \draw[fill=white] (9,5) circle(3pt);

        \draw[fill=white] (0,6) circle(3pt);
        \draw[fill=black] (1,6) circle(3pt);
        \draw[fill=black] (2,6) circle(3pt);
        \draw[fill=white] (3,6) circle(3pt);
        \draw[fill=black] (4,6) circle(3pt);
        \draw[fill=black] (5,6) circle(3pt);
        \draw[fill=white] (6,6) circle(3pt);
        \draw[fill=black] (7,6) circle(3pt);
        \draw[fill=black] (8,6) circle(3pt);
        \draw[fill=white] (9,6) circle(3pt);

        \draw[fill=white] (0,7) circle(3pt);
        \draw[fill=white] (1,7) circle(3pt);
        \draw[fill=white] (2,7) circle(3pt);
        \draw[fill=white] (3,7) circle(3pt);
        \draw[fill=white] (4,7) circle(3pt);
        \draw[fill=white] (5,7) circle(3pt);
        \draw[fill=white] (6,7) circle(3pt);
        \draw[fill=white] (7,7) circle(3pt);
        \draw[fill=white] (8,7) circle(3pt);
        \draw[fill=white] (9,7) circle(3pt);

        \draw[fill=white] (0,8) circle(3pt);
        \draw[fill=white] (1,8) circle(3pt);
        \draw[fill=white] (2,8) circle(3pt);
        \draw[fill=white] (3,8) circle(3pt);
        \draw[fill=white] (4,8) circle(3pt);
        \draw[fill=white] (5,8) circle(3pt);
        \draw[fill=white] (6,8) circle(3pt);
        \draw[fill=white] (7,8) circle(3pt);
        \draw[fill=white] (8,8) circle(3pt);
        \draw[fill=white] (9,8) circle(3pt);

        \draw[fill=black] (0,9) circle(3pt);
        \draw[fill=black] (1,9) circle(3pt);
        \draw[fill=white] (2,9) circle(3pt);
        \draw[fill=black] (3,9) circle(3pt);
        \draw[fill=black] (4,9) circle(3pt);
        \draw[fill=white] (5,9) circle(3pt);
        \draw[fill=black] (6,9) circle(3pt);
        \draw[fill=black] (7,9) circle(3pt);
        \draw[fill=white] (8,9) circle(3pt);
        \draw[fill=black] (9,9) circle(3pt);

        \draw[very thick] (-0.2,1.2) -- (1.2,1.2) -- (1.2,-0.2);

        \draw[very thick] (1.8,-0.2) -- (1.8,1.2) -- (4.2,1.2) -- (4.2,-0.2);

        \draw[very thick] (4.8,-0.2) -- (4.8,1.2) -- (7.2,1.2) -- (7.2,-0.2);

        \draw[very thick] (7.8,-0.2) -- (7.8,1.2) -- (9.2,1.2);

        \draw[very thick] (-0.2,1.8) -- (0.2,1.8) -- (0.2,4.2) -- (-0.2,4.2);

        \draw[very thick] (0.8,1.8) -- (0.8,4.2) -- (3.2,4.2) -- (3.2,1.8) -- (0.8,1.8);

        \draw[very thick] (3.8,1.8) -- (3.8,4.2) -- (6.2,4.2) -- (6.2,1.8) -- (3.8,1.8);

        \draw[very thick] (3.8+3,1.8) -- (3.8+3,4.2) -- (6.2+3,4.2) -- (6.2+3,1.8) -- (3.8+3,1.8);

        \draw[very thick] (-0.2,4.8) -- (-0.2,7.2) -- (2.2,7.2) -- (2.2,4.8) -- (-0.2,4.8);

        \draw[very thick] (-0.2+3,4.8) -- (-0.2+3,7.2) -- (2.2+3,7.2) -- (2.2+3,4.8) -- (-0.2+3,4.8);

        \draw[very thick] (-0.2+6,4.8) -- (-0.2+6,7.2) -- (2.2+6,7.2) -- (2.2+6,4.8) -- (-0.2+6,4.8);

        \draw[very thick] (9.2,4.8) -- (8.8,4.8) -- (8.8,7.2) -- (9.2,7.2);

        \draw[very thick] (-0.2,7.8) -- (1.2,7.8) -- (1.2, 9.2);

        \draw[very thick] (1.8,9.2) -- (1.8,7.8) -- (4.2,7.8) -- (4.2,9.2);

        \draw[very thick] (1.8+3,9.2) -- (1.8+3,7.8) -- (4.2+3,7.8) -- (4.2+3,9.2);

        \draw[very thick] (9.2,7.8) -- (7.8,7.8) -- (7.8,9.2);
    \end{tikzpicture}
       \caption{A local locating-dominating code of density $\frac{2}{9}$ in the triangular grid.}
       \label{localLDtriangular}
    \end{subfigure}\qquad
    \begin{subfigure}[b]{0.4\textwidth}
 \begin{tikzpicture}[scale=0.6]\begin{scope}
    \clip(-0.2,1.77) rectangle (10.5,13.5);

        \draw (0,0) grid (14,14);

        \draw (0,14) -- (14,0);
        \draw (1,14) -- (14,1);
        \draw (2,14) -- (14,2);
        \draw (3,14) -- (14,3);
        \draw (4,14) -- (14,4);
        \draw (5,14) -- (14,5);
        \draw (6,14) -- (14,6);
        \draw (7,14) -- (14,7);
        \draw (8,14) -- (14,8);
        \draw (9,14) -- (14,9);
        \draw (10,14) -- (14,10);
        \draw (11,14) -- (14,11);
        \draw (12,14) -- (14,12);
        \draw (13,14) -- (14,13);

        \draw (0,13) -- (13,0);
        \draw (0,12) -- (12,0);
        \draw (0,11) -- (11,0);
        \draw (0,10) -- (10,0);
        \draw (0,9) -- (9,0);
        \draw (0,8) -- (8,0);
        \draw (0,7) -- (7,0);
        \draw (0,6) -- (6,0);
        \draw (0,5) -- (5,0);
        \draw (0,4) -- (4,0);
        \draw (0,3) -- (3,0);
        \draw (0,2) -- (2,0);
        \draw (0,1) -- (1,0);

        \draw (0,0) -- (14,14);
        \draw (1,0) -- (14,13);
        \draw (2,0) -- (14,12);
        \draw (3,0) -- (14,11);
        \draw (4,0) -- (14,10);
        \draw (5,0) -- (14,9);
        \draw (6,0) -- (14,8);
        \draw (7,0) -- (14,7);
        \draw (8,0) -- (14,6);
        \draw (9,0) -- (14,5);
        \draw (10,0) -- (14,4);
        \draw (11,0) -- (14,3);
        \draw (12,0) -- (14,2);
        \draw (13,0) -- (14,1);

        \draw (0,1) -- (13,14);
        \draw (0,2) -- (12,14);
        \draw (0,3) -- (11,14);
        \draw (0,4) -- (10,14);
        \draw (0,5) -- (9,14);
        \draw (0,6) -- (8,14);
        \draw (0,7) -- (7,14);
        \draw (0,8) -- (6,14);
        \draw (0,9) -- (5,14);
        \draw (0,10) -- (4,14);
        \draw (0,11) -- (3,14);
        \draw (0,12) -- (2,14);
        \draw (0,13) -- (1,14);


        \draw[fill=white] (0,0) circle(3pt);
        \draw[fill=white] (1,0) circle(3pt);
        \draw[fill=white] (2,0) circle(3pt);
        \draw[fill=white] (3,0) circle(3pt);
        \draw[fill=white] (4,0) circle(3pt);
        \draw[fill=white] (5,0) circle(3pt);
        \draw[fill=white] (6,0) circle(3pt);
        \draw[fill=white] (7,0) circle(3pt);
        \draw[fill=white] (8,0) circle(3pt);
        \draw[fill=white] (9,0) circle(3pt);
        \draw[fill=white] (10,0) circle(3pt);
        \draw[fill=white] (11,0) circle(3pt);
        \draw[fill=white] (12,0) circle(3pt);
        \draw[fill=white] (13,0) circle(3pt);
        \draw[fill=white] (14,0) circle(3pt);

        \draw[fill=white] (0,1) circle(3pt);
        \draw[fill=white] (1,1) circle(3pt);
        \draw[fill=white] (2,1) circle(3pt);
        \draw[fill=black] (3,1) circle(3pt);
        \draw[fill=white] (4,1) circle(3pt);
        \draw[fill=white] (5,1) circle(3pt);
        \draw[fill=white] (6,1) circle(3pt);
        \draw[fill=black] (7,1) circle(3pt);
        \draw[fill=white] (8,1) circle(3pt);
        \draw[fill=white] (9,1) circle(3pt);
        \draw[fill=white] (10,1) circle(3pt);
        \draw[fill=black] (11,1) circle(3pt);
        \draw[fill=white] (12,1) circle(3pt);
        \draw[fill=white] (13,1) circle(3pt);
        \draw[fill=white] (14,1) circle(3pt);

        \draw[fill=white] (0,2) circle(3pt);
        \draw[fill=black] (1,2) circle(3pt);
        \draw[fill=white] (2,2) circle(3pt);
        \draw[fill=white] (3,2) circle(3pt);
        \draw[fill=white] (4,2) circle(3pt);
        \draw[fill=black] (5,2) circle(3pt);
        \draw[fill=white] (6,2) circle(3pt);
        \draw[fill=white] (7,2) circle(3pt);
        \draw[fill=white] (8,2) circle(3pt);
        \draw[fill=black] (9,2) circle(3pt);
        \draw[fill=white] (10,2) circle(3pt);
        \draw[fill=white] (11,2) circle(3pt);
        \draw[fill=white] (12,2) circle(3pt);
        \draw[fill=black] (13,2) circle(3pt);
        \draw[fill=white] (14,2) circle(3pt);

        \draw[fill=white] (0,3) circle(3pt);
        \draw[fill=white] (1,3) circle(3pt);
        \draw[fill=white] (2,3) circle(3pt);
        \draw[fill=black] (3,3) circle(3pt);
        \draw[fill=white] (4,3) circle(3pt);
        \draw[fill=white] (5,3) circle(3pt);
        \draw[fill=white] (6,3) circle(3pt);
        \draw[fill=black] (7,3) circle(3pt);
        \draw[fill=white] (8,3) circle(3pt);
        \draw[fill=white] (9,3) circle(3pt);
        \draw[fill=white] (10,3) circle(3pt);
        \draw[fill=black] (11,3) circle(3pt);
        \draw[fill=white] (12,3) circle(3pt);
        \draw[fill=white] (13,3) circle(3pt);
        \draw[fill=white] (14,3) circle(3pt);

        \draw[fill=white] (0,4) circle(3pt);
        \draw[fill=white] (1,4) circle(3pt);
        \draw[fill=white] (2,4) circle(3pt);
        \draw[fill=white] (3,4) circle(3pt);
        \draw[fill=white] (4,4) circle(3pt);
        \draw[fill=white] (5,4) circle(3pt);
        \draw[fill=white] (6,4) circle(3pt);
        \draw[fill=white] (7,4) circle(3pt);
        \draw[fill=white] (8,4) circle(3pt);
        \draw[fill=white] (9,4) circle(3pt);
        \draw[fill=white] (10,4) circle(3pt);
        \draw[fill=white] (11,4) circle(3pt);
        \draw[fill=white] (12,4) circle(3pt);
        \draw[fill=white] (13,4) circle(3pt);
        \draw[fill=white] (14,4) circle(3pt);

        \draw[fill=white] (0,5) circle(3pt);
        \draw[fill=black] (1,5) circle(3pt);
        \draw[fill=white] (2,5) circle(3pt);
        \draw[fill=white] (3,5) circle(3pt);
        \draw[fill=white] (4,5) circle(3pt);
        \draw[fill=black] (5,5) circle(3pt);
        \draw[fill=white] (6,5) circle(3pt);
        \draw[fill=white] (7,5) circle(3pt);
        \draw[fill=white] (8,5) circle(3pt);
        \draw[fill=black] (9,5) circle(3pt);
        \draw[fill=white] (10,5) circle(3pt);
        \draw[fill=white] (11,5) circle(3pt);
        \draw[fill=white] (12,5) circle(3pt);
        \draw[fill=black] (13,5) circle(3pt);
        \draw[fill=white] (14,5) circle(3pt);

        \draw[fill=white] (0,6) circle(3pt);
        \draw[fill=white] (1,6) circle(3pt);
        \draw[fill=white] (2,6) circle(3pt);
        \draw[fill=black] (3,6) circle(3pt);
        \draw[fill=white] (4,6) circle(3pt);
        \draw[fill=white] (5,6) circle(3pt);
        \draw[fill=white] (6,6) circle(3pt);
        \draw[fill=black] (7,6) circle(3pt);
        \draw[fill=white] (8,6) circle(3pt);
        \draw[fill=white] (9,6) circle(3pt);
        \draw[fill=white] (10,6) circle(3pt);
        \draw[fill=black] (11,6) circle(3pt);
        \draw[fill=white] (12,6) circle(3pt);
        \draw[fill=white] (13,6) circle(3pt);
        \draw[fill=white] (14,6) circle(3pt);

        \draw[fill=white] (0,7) circle(3pt);
        \draw[fill=black] (1,7) circle(3pt);
        \draw[fill=white] (2,7) circle(3pt);
        \draw[fill=white] (3,7) circle(3pt);
        \draw[fill=white] (4,7) circle(3pt);
        \draw[fill=black] (5,7) circle(3pt);
        \draw[fill=white] (6,7) circle(3pt);
        \draw[fill=white] (7,7) circle(3pt);
        \draw[fill=white] (8,7) circle(3pt);
        \draw[fill=black] (9,7) circle(3pt);
        \draw[fill=white] (10,7) circle(3pt);
        \draw[fill=white] (11,7) circle(3pt);
        \draw[fill=white] (12,7) circle(3pt);
        \draw[fill=black] (13,7) circle(3pt);
        \draw[fill=white] (14,7) circle(3pt);

        \draw[fill=white] (0,8) circle(3pt);
        \draw[fill=white] (1,8) circle(3pt);
        \draw[fill=white] (2,8) circle(3pt);
        \draw[fill=white] (3,8) circle(3pt);
        \draw[fill=white] (4,8) circle(3pt);
        \draw[fill=white] (5,8) circle(3pt);
        \draw[fill=white] (6,8) circle(3pt);
        \draw[fill=white] (7,8) circle(3pt);
        \draw[fill=white] (8,8) circle(3pt);
        \draw[fill=white] (9,8) circle(3pt);
        \draw[fill=white] (10,8) circle(3pt);
        \draw[fill=white] (11,8) circle(3pt);
        \draw[fill=white] (12,8) circle(3pt);
        \draw[fill=white] (13,8) circle(3pt);
        \draw[fill=white] (14,8) circle(3pt);

        \draw[fill=white] (0,9) circle(3pt);
        \draw[fill=white] (1,9) circle(3pt);
        \draw[fill=white] (2,9) circle(3pt);
        \draw[fill=black] (3,9) circle(3pt);
        \draw[fill=white] (4,9) circle(3pt);
        \draw[fill=white] (5,9) circle(3pt);
        \draw[fill=white] (6,9) circle(3pt);
        \draw[fill=black] (7,9) circle(3pt);
        \draw[fill=white] (8,9) circle(3pt);
        \draw[fill=white] (9,9) circle(3pt);
        \draw[fill=white] (10,9) circle(3pt);
        \draw[fill=black] (11,9) circle(3pt);
        \draw[fill=white] (12,9) circle(3pt);
        \draw[fill=white] (13,9) circle(3pt);
        \draw[fill=white] (14,9) circle(3pt);

        \draw[fill=white] (0,10) circle(3pt);
        \draw[fill=black] (1,10) circle(3pt);
        \draw[fill=white] (2,10) circle(3pt);
        \draw[fill=white] (3,10) circle(3pt);
        \draw[fill=white] (4,10) circle(3pt);
        \draw[fill=black] (5,10) circle(3pt);
        \draw[fill=white] (6,10) circle(3pt);
        \draw[fill=white] (7,10) circle(3pt);
        \draw[fill=white] (8,10) circle(3pt);
        \draw[fill=black] (9,10) circle(3pt);
        \draw[fill=white] (10,10) circle(3pt);
        \draw[fill=white] (11,10) circle(3pt);
        \draw[fill=white] (12,10) circle(3pt);
        \draw[fill=black] (13,10) circle(3pt);
        \draw[fill=white] (14,10) circle(3pt);

        \draw[fill=white] (0,11) circle(3pt);
        \draw[fill=white] (1,11) circle(3pt);
        \draw[fill=white] (2,11) circle(3pt);
        \draw[fill=black] (3,11) circle(3pt);
        \draw[fill=white] (4,11) circle(3pt);
        \draw[fill=white] (5,11) circle(3pt);
        \draw[fill=white] (6,11) circle(3pt);
        \draw[fill=black] (7,11) circle(3pt);
        \draw[fill=white] (8,11) circle(3pt);
        \draw[fill=white] (9,11) circle(3pt);
        \draw[fill=white] (10,11) circle(3pt);
        \draw[fill=black] (11,11) circle(3pt);
        \draw[fill=white] (12,11) circle(3pt);
        \draw[fill=white] (13,11) circle(3pt);
        \draw[fill=white] (14,11) circle(3pt);

        \draw[fill=white] (0,12) circle(3pt);
        \draw[fill=white] (1,12) circle(3pt);
        \draw[fill=white] (2,12) circle(3pt);
        \draw[fill=white] (3,12) circle(3pt);
        \draw[fill=white] (4,12) circle(3pt);
        \draw[fill=white] (5,12) circle(3pt);
        \draw[fill=white] (6,12) circle(3pt);
        \draw[fill=white] (7,12) circle(3pt);
        \draw[fill=white] (8,12) circle(3pt);
        \draw[fill=white] (9,12) circle(3pt);
        \draw[fill=white] (10,12) circle(3pt);
        \draw[fill=white] (11,12) circle(3pt);
        \draw[fill=white] (12,12) circle(3pt);
        \draw[fill=white] (13,12) circle(3pt);
        \draw[fill=white] (14,12) circle(3pt);

        \draw[fill=white] (0,13) circle(3pt);
        \draw[fill=black] (1,13) circle(3pt);
        \draw[fill=white] (2,13) circle(3pt);
        \draw[fill=white] (3,13) circle(3pt);
        \draw[fill=white] (4,13) circle(3pt);
        \draw[fill=black] (5,13) circle(3pt);
        \draw[fill=white] (6,13) circle(3pt);
        \draw[fill=white] (7,13) circle(3pt);
        \draw[fill=white] (8,13) circle(3pt);
        \draw[fill=black] (9,13) circle(3pt);
        \draw[fill=white] (10,13) circle(3pt);
        \draw[fill=white] (11,13) circle(3pt);
        \draw[fill=white] (12,13) circle(3pt);
        \draw[fill=black] (13,13) circle(3pt);
        \draw[fill=white] (14,13) circle(3pt);

        \draw[fill=white] (0,14) circle(3pt);
        \draw[fill=white] (1,14) circle(3pt);
        \draw[fill=white] (2,14) circle(3pt);
        \draw[fill=black] (3,14) circle(3pt);
        \draw[fill=white] (4,14) circle(3pt);
        \draw[fill=white] (5,14) circle(3pt);
        \draw[fill=white] (6,14) circle(3pt);
        \draw[fill=black] (7,14) circle(3pt);
        \draw[fill=white] (8,14) circle(3pt);
        \draw[fill=white] (9,14) circle(3pt);
        \draw[fill=white] (10,14) circle(3pt);
        \draw[fill=black] (11,14) circle(3pt);
        \draw[fill=white] (12,14) circle(3pt);
        \draw[fill=white] (13,14) circle(3pt);
        \draw[fill=white] (14,14) circle(3pt);

        \draw[very thick] (1.8,11.2) -- (5.2,11.2) -- (5.2,3.8) -- (1.8,3.8) -- (1.8,11.2);

        \draw[very thick] (1.8+4,11.2) -- (5.2+4,11.2) -- (5.2+4,3.8) -- (1.8+4,3.8) -- (1.8+4,11.2);

        \draw[very thick] (1.8+8,11.2) -- (5.2+8,11.2) -- (5.2+8,3.8) -- (1.8+8,3.8) -- (1.8+8,11.2);

        \draw[very thick] (-0.2,3.8) -- (1.2,3.8) -- (1.2,11.2) -- (-0.2,11.2);

        \draw[very thick] (14.2,3.8) -- (13.8,3.8) --
        (13.8, 11.2) -- (14.2,11.2);

        \draw[very thick] (14.2,11.8) -- (13.8,11.8) -- (13.8,14.2);

        \draw[very thick] (1.8,14.2) -- (1.8,11.8) -- (5.2,11.8) -- (5.2,14.2);

        \draw[very thick] (1.8+4,14.2) -- (1.8+4,11.8) -- (5.2+4,11.8) -- (5.2+4,14.2);

        \draw[very thick] (1.8+8,14.2) -- (1.8+8,11.8) -- (5.2+8,11.8) -- (5.2+8,14.2);

        \draw[very thick] (-0.2,11.8) -- (1.2,11.8) -- (1.2,14.2);

        \draw[very thick] (-0.2,3.2) -- (1.2,3.2) -- (1.2,-0.2);

        \draw[very thick] (14.2,3.2) -- (13.8,3.2) -- (13.8,-0.2);

        \draw[very thick] (1.8,-0.2) -- (1.8,3.2) -- (5.2,3.2) -- (5.2,-0.2);

        \draw[very thick] (1.8+4,-0.2) -- (1.8+4,3.2) -- (5.2+4,3.2) -- (5.2+4,-0.2);

        \draw[very thick] (1.8+8,-0.2) -- (1.8+8,3.2) -- (5.2+8,3.2) -- (5.2+8,-0.2);
    \end{scope}\end{tikzpicture}
    \caption{A local locating-dominating code of density $\frac{3}{16}$ in the king grid.}
    \label{kinglocalLD}
    \end{subfigure}
        \caption{Local identifying and locating-dominating codes in the triangular and king grids.}
 \end{figure}

\begin{theorem}
    $$
    \gamma^{L-ID}(\mathcal{T}) = \frac{1}{4} = \gamma^{ID}(\mathcal{T}).
    $$
\end{theorem}

\begin{proof}
    Since any identifying code is also a local identifying code, we have the upper bound
    $\gamma^{L-ID}(\mathcal{T}) \leq \gamma^{ID}(\mathcal{T}) = \frac{1}{4}$ (see Table \ref{known bounds}).
    Next, we prove the lower bound $\gamma^{L-ID}(\mathcal{T}) \geq \frac{1}{4}$.
    So, let $C \subseteq \Z^2$ be a local identifying code in the triangular grid.
    We show that $D(C) \geq \frac{1}{4}$ by showing that $s(\bc) \leq 4$ for all $\bc \in C$ which then gives the claim by Lemma \ref{share infinite}.
    Let $\bc \in C$ be an arbitrary codeword.

\medskip
    Assume first that $I(\bc) = \{ \bc \}$.
    Since $C$ separates $\bc$ from its neighbours, every neighbour of $\bc$ is covered by at least two codewords and hence $s(\bc) \leq 1 + 6 \cdot \frac{1}{2} = 4$.

\medskip
    Assume then that $|I(\bc)| \geq 2$ and let $\bc' \in C$ be a codeword neighbour of $\bc$. The codeword $\bc$ has three neighbours, say $\bu_1, \bu_2$ and $\bu_3$, that are not covered by $\bc'$. One of them is a neighbour of the two others.
    Without loss of generality we may assume that $\bu_2$ is a neighbour of both $\bu_1$ and $\bu_3$. It follows that $\bu_1$ and $\bu_3$ are not neighbours.
    If $C$ covers at least two of these three points by at least two codewords, then $s(\bc) \leq 6 \cdot \frac{1}{2} + 1 = 4$.
    So, let us assume that $C$ covers two of these points by only one codeword -- with $\bc$.
    Note that if $C$ covers $\bu_2$ only by $\bc$, then its neighbours $\bu_1$ and $\bu_3$ are covered by at least two codewords because $C$ -- being a local identifying code -- separates $\bu_2$ from its neighbours.
    So, we assume that $I(\bu_1) = \{ \bc \} = I(\bu_3)$,
    and then
    we have $|I(\bu_2)| \geq 2$ and $I(\bc) = \{ \bc, \bc' \}$.
    The codeword neighbour $\bc'$ of $\bc$ covers the two remaining neighbours of $\bc$, say $\bu_4$ and $\bu_5$, and of course $\bc$ and $\bc'$.
    Since $I(\bc) = \{ \bc, \bc' \}$, the code $C$ covers the points $\bu_4, \bu_5$ and $\bc'$ by at least three codewords in order to separate them from their neighbour $\bc$.
    Thus, we have
    $$
    s(\bc) \leq 2 \cdot \frac{1}{2} + 2 \cdot 1 + 3 \cdot \frac{1}{3} = 4.
    $$

\vspace*{-8mm}
\end{proof}

\subsection{The king grid}\label{subsec:king}

Next, we consider consider the optimal densities of local locating-dominating and local identifying codes in the king grid $\mathcal{K}$.

\medskip
Let us begin with some terminology.
We call the points $\bx + (\pm 1,\pm 1)$ the \emph{corner neighbours} of $\bx \in \Z^2$
and the points $\bx + (\pm 1,0), \bx + (0,\pm 1)$ the \emph{non-corner neighbours} of $\bx$.
If $\by$ is a corner neighbour of $\bx$, then $|N[\by] \cap N[\bx]| = 4$ and if $\by$ is a non-corner neighbour of $\bx$, then
$|N[\by] \cap N[\bx]| = 6$.
We say that two corner neighbours of a point $\bx$ are \emph{adjoining} if their Euclidean distance is $2$ and they are \emph{opposite} if they are not adjoining, {\it i.e.}, if their Euclidean distance is $2 \sqrt{2}$.
Two non-corner neighbours are adjoining if their Euclidean distance is $\sqrt{2}$ and opposite if they are not adjoining in which case their Euclidean distance is $2$.
Note that two adjoining non-corner neighbours of a vertex are neighbours, in particular.
A non-corner neighbour of $\bx$ is \emph{between} two adjoining corner neighbours of $\bx$ if it is at Euclidean distance 1 from both of them.

\begin{theorem}
    $$
    \gamma^{L-LD}(\mathcal{K}) = \frac{3}{16}.
    $$
\end{theorem}

\begin{proof}
    By a construction we have $\gamma^{L-LD}(\mathcal{K}) \leq \frac{3}{16}$, see Figure \ref{kinglocalLD}.

\medskip
    Next, we show that $\gamma^{L-LD}(\mathcal{K}) \geq \frac{3}{16}$.
    Let $C \subseteq \Z^2$ be a local locating-dominating code in the king grid.
    We claim that any $4 \times 4$ square $D \subseteq \Z^2$ contains at least three codewords of $C$.
    This implies that $D(C) \geq \frac{3}{16}$.
    Let $\bt \in \Z^2$ be such that $D = \{0,1,2,3\} \times \{0,1,2,3\} + \bt$.
    Let $D' = \{1,2\} \times \{1,2\} + \bt$ be the $2 \times 2$ square inside of $D$ that does not intersect the border of $D$.
    Notice that the neighbourhood of $D'$ is $D$.
    We have four separate cases according to the number of codewords in $D'$.

    \begin{itemize}
        \item If $|C \cap D'| \in \{3,4\}$, then  $|C \cap D| \geq 3$ and hence the claim holds.

        \item Assume that $|C \cap D'| = 2$.
        Let $\ba$ and $\bb$ be the two non-codewords in $D'$.
        They are neighbours and the two codewords in $D'$ are also neighbours of both $\ba$ and $\bb$.
        This means that there must be a third codeword in $D$ which separates $\ba$ and $\bb$ and hence
        $|C \cap D| \geq 3$.

        \item Assume that $|C \cap D'| = 1$.
        We need at least two more codewords in $D$ to separate the three non-codewords in $D'$ from each other.
        Thus, $|C \cap D| \geq 3$.

        \item Finally, assume that $|C \cap D'| = 0$.
        With one codeword in $D \setminus D'$ the code $C$ can cover at most two points of $D'$.
        So, we need at least two codewords in $D$ to cover the points of $D'$.
        To separate them we need at least three codewords in $D$.
        So, also in this case
        $|C \cap D| \geq 3$.
    \end{itemize}

    \vspace*{-8mm}
\end{proof}

Finally, let us settle the question of the optimal density of local identifying codes in the king grid.
It turns out that it is $2/9$, the same as the optimal density of identifying codes.
For the proof we introduce a share shifting scheme with the following two rules:

\begin{itemize}
 \leftskip=4.5mm 
\item[\textbf{Rule 1}:] If $\bc\in C$ has $|I(\bc)|=2$ and $ s(\bc)>\frac{9}{2}$, then we shift $1/4$ share units from $\bc$ to the adjacent corner neighbour codeword $\bc'$ which has a non-corner codeword neighbour.
\item[\textbf{Rule 2}:] If $\bc\in C$ has $|I(\bc)|=1$ and $s(\bc)>\frac{9}{2}$, then we shift $1/12$ share units from $\bc$ to two pairwise non-adjacent codewords  at (graphic) distance $2$ and Euclidean distance $\sqrt{5}$ from $\bc$ which are covered by at least three codewords. If there are adjacent codewords $\bc_1$ and $\bc_2$ which satisfy these conditions, then we choose a  codeword $\bc_i$ which satisfies $I(\bc_i)\not\subseteq I(\bc_j)$ where $\{i,j\}=\{1,2\}$.
\end{itemize}

\noindent
We have illustrated Rule $2$ in Figure \ref{constellations2} Constellation 5.
We denote by $s'(\bc)$ the share of a codeword $\bc\in C$ after applying Rule $1$ and by $s''(\bc)$ the share of $\bc\in C$ after applying both Rules $1$ and $2$ (in that order). In Lemma \ref{LemmaKingRule1} and Theorem \ref{TheLIDKing}, we will notice that Rule $1$ is applied only to codewords with $\frac{9}{2} < s(\bc)\leq \frac{19}{4}$ and Rule $2$ is applied only to codewords with $s(\bc)\in\{4\frac{7}{12} , 4\frac{2}{3}\}$. Moreover, after applying them, we will have $s''(\bc)\leq \frac{9}{2}$ for every $\bc\in C$.

\begin{lemma}\label{LemmaKingRule1}
    Let $C$ be a local identifying code in the king grid and let $\bc,\bc'\in C$ be such that $|I(\bc)|\geq2$ and $\bc'\in I(\bc)$. After applying Rule $1$, we have $s'(\bc)\leq 9/2$ and $s'(\bc')\leq4$ if Rule $1$ shifted share to $\bc'$.
\end{lemma}

\begin{proof}
       We have two claims, that $s'(\bc)\leq 9/2$ for all $\bc\in C$ and $s'(\bc')\leq 4$ if Rule $1$ shifts share to $\bc'\in C$  (notice that Rule $1$ may shift share to $\bc'$ more than once). We will confirm the second claim each time after we have shifted share into a codeword. In the following, we assume that $\bc'\in C$ is a codeword neighbour of $\bc\in C$.

\medskip
    \textbf{Case 1.} Assume first that $\bc'$ is a non-corner neighbour of $\bc$ and that Rule $1$ does not shift share to $\bc$. Without loss of generality, we may assume that $\bc' = \bc + (1,0)$. We have illustrated this case in Figure \ref{constellations1} Constellation 1. The codeword $\bc'$ does not cover the points $\bc + (-1,1)$, $\bc + (-1,0)$ and $\bc + (-1,-1)$. At most one of these points can be covered by exactly one codeword.
    Since $C$ separates $\bc$ and $\bc'$, at least one of them is covered by at least three codewords.
    Also, $C$ separates the neighbours $\bc + (0,1)$ and $\bc + (1,1)$ which means that at least one of them is covered by at least three codewords. Similarly, $C$ separates  the neighbours $\bc + (0,-1)$ and $\bc + (1,-1)$ which means that at least one of them is covered by at least three codewords.
    Thus, in this case in the neighbourhood of $\bc$, at most one point is covered by only one codeword, at most five points are covered by only two codewords and at least three points are covered by at least three codewords and hence
       $$
       s'(\bc)=s(\bc) \leq  1 + 5 \cdot \frac{1}{2} + 3 \cdot \frac{1}{3} = \frac{9}{2}.
       $$

\textbf{Case 2.} Assume then that $\bc'$ is a corner neighbour of $\bc$ and that $\bc$ does not have any non-corner codeword neighbours. Thus, $s'(\bc)\leq s(\bc)$. Without loss of generality, we may assume that $\bc' = \bc + (1,1)$. First, there are at least three vertices in the closed neighbourhood of $\bc$ covered by at least three codewords: Indeed, $\bc'$ covers the points in the set $\{\bc,\bc',\bc + (0,1),\bc + (1,0)\} = N[\bc]\cap N[\bc']$ and since all of these vertices are adjacent, at most one of them can have  $\{\bc,\bc'\}$ as its $I$-set. Hence, $\bc'$ has an adjacent non-corner codeword.

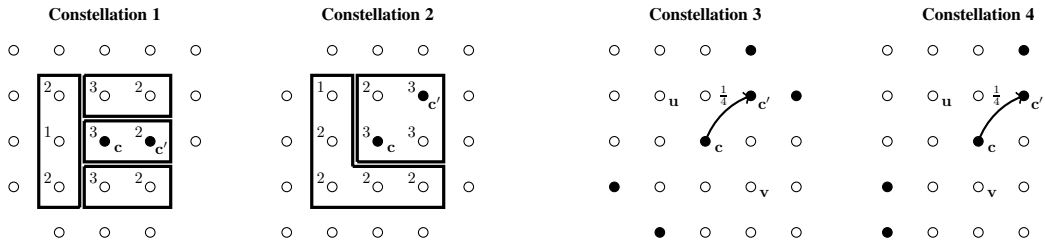
\begin{figure}[!ht]
\vspace*{-2mm}
    \centering
    \begin{tikzpicture}[scale=0.6]
                \begin{scope}[xshift=-9cm]

        \node[scale=0.6] at (1,4.8) {\textbf{Constellation 1}};

        \draw[] (0,0) circle(3pt);
        \draw[] (1,0) circle(3pt);
        \draw[] (2,0) circle(3pt);
        \draw[] (0,1) circle(3pt);
        \draw[] (1,1) circle(3pt); 
        \draw[] (2,1) circle(3pt);

        \draw[] (-1,2) circle(3pt);
        \draw[] (0,2) circle(3pt);
        \draw[fill=black] (1,2) circle(3pt); \node[scale=0.6] at (1.3,1.85) {\small $\bc$};
        \draw[fill=black] (2,2) circle(3pt); \node[scale=0.6] at (2.25,1.85) {\small $\bc'$};
        \draw[] (3,2) circle(3pt);

        \draw[] (-1,3) circle(3pt);
        \draw[] (0,3) circle(3pt);
        \draw[] (1,3) circle(3pt);
        \draw[] (2,3) circle(3pt);
        \draw[] (3,3) circle(3pt); 
        \draw[] (-1,1) circle(3pt); 

        \draw[] (-1,4) circle(3pt);
        \draw[] (0,4) circle(3pt);
        \draw[] (1,4) circle(3pt);
        \draw[] (2,4) circle(3pt);
        \draw[] (3,4) circle(3pt);

        \draw[very thick] (0.45,0.55) -- (-0.45,0.55) -- (-0.45,3.45) -- (0.45,3.45) -- (0.45,0.55);
        \node[scale=0.6] at (-0.25,3.2) {\small $2$};
        \node[scale=0.6] at (-0.25,2.2) {\small $1$};
        \node[scale=0.6] at (-0.25,1.2) {\small $2$};

        \draw[very thick] (0.55,2.45) -- (2.45,2.45) -- (2.45,1.55) -- (0.55,1.55) -- (0.55,2.45);
        \node[scale=0.6] at (1.75,3.2) {\small $2$};
        \node[scale=0.6] at (0.75,3.2) {\small $3$};

        \draw[very thick] (0.55,3.45) -- (2.45,3.45) -- (2.45,2.55) -- (0.55,2.55) -- (0.55,3.45);
        \node[scale=0.6] at (1.75,2.2) {\small $2$};
        \node[scale=0.6] at (0.75,2.2) {\small $3$};

        \draw[very thick] (0.55,1.45) -- (2.45,1.45) -- (2.45,0.55) -- (0.55,0.55) -- (0.55,1.45);
        \node[scale=0.6] at (1.75,1.2) {\small $2$};
        \node[scale=0.6] at (0.75,1.2) {\small $3$};


        \node[scale=0.6] at (7,4.8) {\textbf{Constellation 2}};

        \draw[] (5,0) circle(3pt);
        \draw[] (6,0) circle(3pt);
        \draw[] (7,0) circle(3pt);
        \draw[] (8,0) circle(3pt);

        \draw[] (5,1) circle(3pt);
        \draw[] (6,1) circle(3pt);
        \draw[] (7,1) circle(3pt);
        \draw[] (8,1) circle(3pt);
        \draw[] (9,1) circle(3pt);

        \draw[] (5,2) circle(3pt);
        \draw[] (6,2) circle(3pt);
        \draw[fill=black] (7,2) circle(3pt); 		\node[scale=0.6] at (7.3,1.85) {\small $\bc$};
        \draw[] (8,2) circle(3pt);
        \draw[] (9,2) circle(3pt);

        \draw[] (5,3) circle(3pt);
        \draw[] (6,3) circle(3pt);
        \draw[] (7,3) circle(3pt);
        \draw[fill=black] (8,3) circle(3pt);
        \node[scale=0.6] at (8.25,2.85) {\small $\bc'$};
        \draw[] (9,3) circle(3pt);

        \draw[] (6,4) circle(3pt);
        \draw[] (7,4) circle(3pt);
        \draw[] (8,4) circle(3pt);
        \draw[] (9,4) circle(3pt);

                  \draw[very thick] (6.55,1.55) -- (8.45,1.55) -- (8.45,3.45) -- (6.55,3.45) -- (6.55,1.55);
                \node[scale=0.6] at (6.75,3.2) {\small $2$};
                \node[scale=0.6] at (6.75,2.2) {\small $3$};
                \node[scale=0.6] at (7.75,3.2) {\small $3$};
                \node[scale=0.6] at (7.75,2.2) {\small $3$};

        \draw[very thick] (6.45,1.45) -- (8.45,1.45) -- (8.45,0.55) -- (5.55,0.55) -- (5.55,3.45) -- (6.45,3.45) -- (6.45,1.45);
                \node[scale=0.6] at (5.75,3.2) {\small $1$};
                \node[scale=0.6] at (5.75,2.2) {\small $2$};
                \node[scale=0.6] at (5.75,1.2) {\small $2$};
                \node[scale=0.6] at (6.75,1.2) {\small $2$};
                \node[scale=0.6] at (7.75,1.2) {\small $2$};

        \end{scope}

\begin{scope}[xshift=4.2cm,yshift=12cm]
        \node[scale=0.6] at (1,-7.2) {\textbf{Constellation 3}};

        \draw (-1,-8) circle(3pt);
        \draw (0,-8) circle(3pt);
        \draw (1,-8) circle(3pt);
        \draw[fill=black] (2,-8) circle(3pt);

        \draw (-1,-9) circle(3pt);
        \draw (0,-9) circle(3pt); \node[scale=0.6] at (0.3,-9.15) {\small $\bu$};
        \draw[] (1,-9) circle(3pt);
        \draw[fill=black] (2,-9) circle(3pt); \node[scale=0.6] at (2.3,-9.15) {\small $\bc'$};
        \draw[fill=black] (3,-9) circle(3pt);

        \draw[] (-1,-10) circle(3pt);
        \draw (0,-10) circle(3pt);
        \draw[fill=black] (1,-10) circle(3pt); \node[scale=0.6] at (1.3,-10.15) {\small $\bc$};
        \draw[] (2,-10) circle(3pt);
        \draw[] (3,-10) circle(3pt);

        \draw[fill=black] (-1,-11) circle(3pt);
        \draw (0,-11) circle(3pt);
        \draw[] (1,-11) circle(3pt);
        \draw[] (2,-11) circle(3pt); \node[scale=0.6] at (2.3,-11.15) {\small $\bv$};
        \draw[] (3,-11) circle(3pt);

        \draw[fill=black] (0,-12) circle(3pt);
        \draw[] (1,-12) circle(3pt);
        \draw[] (2,-12) circle(3pt);
        \draw[] (3,-12) circle(3pt);

        \node[scale=0.6] at (1.4,-9) {\small $\frac{1}{4}$};
        \path[draw,thick]
        (1,-10) edge [->,bend left=25, looseness=0.9] node {} (2,-9)
         ;

        \node[scale=0.6] at (7,-7.2) {\textbf{Constellation 4}};

        \draw[] (5,-8) circle(3pt);
        \draw[] (6,-8) circle(3pt);
        \draw[] (7,-8) circle(3pt);
        \draw[fill=black] (8,-8) circle(3pt);

        \draw[] (5,-9) circle(3pt);
        \draw[] (6,-9) circle(3pt); \node[scale=0.6] at (6.3,-9.15) {\small $\bu$};
        \draw[] (7,-9) circle(3pt);
        \draw[fill=black] (8,-9) circle(3pt); \node[scale=0.6] at (8.3,-9.15) {\small $\bc'$};

        \draw[] (5,-10) circle(3pt);
        \draw[] (6,-10) circle(3pt);
        \draw[fill=black] (7,-10) circle(3pt); \node[scale=0.6] at (7.3,-10.15) {\small $\bc$};
        \draw[] (8,-10) circle(3pt);

        \draw[fill=black] (5,-11) circle(3pt);
        \draw[] (6,-11) circle(3pt);
        \draw[] (7,-11) circle(3pt); \node[scale=0.6] at (7.3,-11.15) {\small $\bv$};
        \draw[] (8,-11) circle(3pt);

        \draw[fill=black] (5,-12) circle(3pt);
        \draw[] (6,-12) circle(3pt);
        \draw[] (7,-12) circle(3pt);
        \draw[] (8,-12) circle(3pt);

             \node[scale=0.6] at (7.4,-9) {\small $\frac{1}{4}$};
        \path[draw,thick]
        (7,-10) edge [->,bend left=25, looseness=0.9] node {} (8,-9)
         ;
     \end{scope}    \end{tikzpicture}
    \caption{Constellations in the king grid. The edges have been omitted for simplicity.  Numbers adjacent to vertices denote the minimum number of codewords that might cover them. Numbers within boxes can sometimes be switched with each other within the same box. Constellations 3 and 4 also illustrate shifting using Rule 1.}
    \label{constellations1}\vspace*{-3mm}
\end{figure}

The codeword $\bc'$ does not cover the points $\bc + (-1,1)$, $\bc + (-1,0)$, $\bc + (-1,-1)$, $\bc + (0,-1)$ and $\bc + (1,-1)$ among the points in $N[\bc]$.
Clearly
at most two of these points are covered by only one codeword --- $\bc$. Let us assume that there exist two such points and let us name them $\bu$ and $\bv$.
Otherwise, we have the case as in Figure \ref{constellations1} Constellation 2 giving
$s'(\bc)=s(\bc) \leq 1 + 5 \cdot \frac{1}{2} + 3 \cdot \frac{1}{3} = \frac{9}{2}$. 

\smallskip
Since the two non-corner neighbours of $\bc$ that $\bc'$ does not cover are neighbours, at least one of $\bu$ and $\bv$ is a corner neighbour of $\bc$.

 \medskip
\textbf{Case 2.1} Assume first that $\bu$ and $\bv$ are both corner neighbours of $\bc$.
If they are adjoining, then $C$ separates neither of them from the non-corner neighbour between them.
Thus, $\bu$ and $\bv$ are opposite corner neighbours of $\bc$ and hence, $\bu = \bc + (-1,1)$ and $\bv = \bc + (1,-1)$ (or vice versa).

Note that $\bc'$ is covered by at least four codewords. Indeed, recall that at least one of $\bc + (1,0)$ and $\bc + (0,1)$ is covered by at least three codewords. Moreover, since $I(\bu)=I(\bv)=\{\bc\}$, the only possible locations for the third codeword are $\bv'=\bc + (2,1)$ and $\bu'=\bc + (1,2)$. However, if only one of these two vertices is a codeword, say $\bu'$, then we require a fourth codeword in $I(\bc')$ to separate $\bc'$ and $\bc+(0,1)$. Thus, $|I(\bc')|\geq4$.
If $\bc + (-1,-1) \in C$, then $C$ covers it by four codewords due to the same arguments as above and hence, $s(\bc) \leq 2 \cdot 1 + 2 \cdot \frac{1}{2} + 3 \cdot \frac{1}{3} + 2 \cdot \frac{1}{4} = \frac{9}{2}$.

\medskip
So, let us assume that $\bc + (-1,-1) \not \in C$.
Thus, $I(\bc) = \{ \bc, \bc' \}$ and hence, $\bu',\bv'\in C$.  Now, $\bc + (-2,-1) \in C$ since $C$ separates $\bc + (-1,0)$ and $\bu$.
Similarly, $\bc + (-1,-2) \in C$ since $C$ separates $\bc + (0,-1)$ and $\bv$. See Constellation $3$ in Figure \ref{constellations1}.
Thus,
$$
s(\bc) \leq 2 \cdot 1 + 3 \cdot \frac{1}{2} + 3 \cdot \frac{1}{3} +  \frac{1}{4} = \frac{19}{4}.
$$

 Furthermore, we can give a rough upper bound $$s(\bc')\leq \frac{3}{2}+\frac{5}{3}+\frac{1}{4}=3\frac{5}{12}$$ as we can see from Figure \ref{constellations1} Constellation 3.

\medskip
 Now, we shift $1/4$ share units according to Rule $1$ from $\bc$ to $\bc'$. After this, we have
 $$s'(\bc)\leq\frac{19}{4}-\frac{1}{4}=\frac{9}{2}$$

 and
  \begin{equation}\label{EqCase21s'c'}
     s'(\bc')\leq3\frac{5}{12}+\frac{1}{4}= 3\frac{2}{3}<4.
 \end{equation} Indeed, observe that Rule $1$ can shift share to vertex $\bc'$ only once as $\bc'$ does not have any other suitable corner neighbours.

\medskip
\textbf{Case 2.2} Finally, assume that $\bu$ is a corner neighbour of $\bc$ and $\bv$ is a non-corner neighbour of $\bc$.
Without loss of generality, we may assume that $\bu = \bc + (-1,1)$ and that $\bv = \bc + (0,-1)$. See Constellation $4$ of Figure \ref{constellations1}.
Now, $I(\bc) = \{ \bc , \bc' \}$.
We have $\bc + (-2,-1) \in C$ and $\bc + (-2,-2)\in C$ because $C$ separates $\bc + (-1,0)$ and $\bu$ and because $C$ separates $\bc + (-1,-1)$ and $\bc + (-1,0)$, respectively.
Since $C$ separates $\bc + (0,1)$ and $\bc$, we have $\bc + (1,2) \in C$ and
since $C$ separates $\bc + (0,1)$ and $\bc' = \bc + (1,1)$, $\bc'$ is covered by at least four codewords.
The code $C$ separates $\bc + (1,0)$ and $\bc$ and thus $\bc + (1,0)$ is covered by at least three codewords.
Finally, $C$ separates $\bc + (1,-1)$ and $\bv$ and hence the point $\bc + (1,-1)$ is covered by at least two codewords.
Thus,
$$
s(\bc) \leq 2 \cdot 1 + 3 \cdot \frac{1}{2} + 3 \cdot \frac{1}{3} +  \frac{1}{4} = \frac{19}{4}.
$$

Let us then consider $s(\bc')$. Recall that $\bc'$ is covered by at least four codewords. Hence, at least one of vertices $\bc'+(1,-1),\bc'+(1,0),\bc'+(1,1)$ is a codeword and thus $\bc'+(1,0)$ is covered by at least three codewords. Furthermore, $\bc+(1,0)$ must be covered by at least three codewords to separate it from $\bc$. Moreover, at most one of $\bc'+(-1,1),\bc'+(0,1),\bc'+(1,1)$ can be covered by only two codewords. Consequently, the only other neighbours of $\bc'$ which can be covered by only two codewords are $\bc$ and $\bc'+(1,-1)$. Notice that since $\bc+(1,0)$ is covered by at least three codewords, $\bc'+(1,-1)$ is covered by at least two codewords. Shares of other points have been considered for $\bc$ and can be re-verified with Constellation 4 of Figure \ref{constellations1}. Thus,
$$s(\bc')\leq \frac{3}{2}+\frac{5}{3}+\frac{1}{4}=3\frac{5}{12}.$$

When we apply Rule $1$, we shift $1/4$ share units away from $s(\bc)$ and hence $s'(\bc)\leq \frac{9}{2}$. Moreover, we shift to $\bc'$ at most $\frac{1}{2}$ share (if $\bc+(2,0)\in C$ and it is in somewhat similar position to $\bc$). Hence, we have \begin{equation}\label{EqCase22s'c'}s'(\bc')\leq 3\frac{11}{12}<4.\end{equation}

Together these two cases give the claim, with the observations that we have calculated value of modified share $s'$ for $\bc$ and $\bc'$ (see Equations (\ref{EqCase21s'c'}) and (\ref{EqCase22s'c'})) whenever we have shifted share and all cases in which share can be shifted with Rule $1$ have been considered.
\end{proof}

Now, we are ready to prove the exact density of optimal local identifying codes in the king grid.

\begin{theorem}\label{TheLIDKing}
    $$
    \gamma^{L-ID}(\mathcal{K}) = \frac{2}{9}.
    $$
\end{theorem}

\begin{proof}
    Since any identifying code is, in particular, a local identifying code, we have $\gamma^{L-ID}(\mathcal{K}) \leq \gamma^{ID}(\mathcal{K}) = \frac{2}{9}$.
    We prove the lower bound $\gamma^{L-ID}(\mathcal{K}) \geq \frac{2}{9}$ by showing that $s''(\bc) \leq \frac{9}{2}$ for each $\bc \in C$. After that, the claim follows from  Lemma \ref{Lemma shareshift}. Recall that we apply first Rule $1$ and then Rule $2$ to obtain value for modified share function $s''$.
    By Lemma \ref{LemmaKingRule1}, after applying Rule $1$, each codeword $\bc$ with $|I(\bc)|\geq2$ has share of at most $s'(\bc)\leq \frac{9}{2}$. In the following, we consider a codeword $\bc$ with $I(\bc)=\{\bc\}$. Thus, Rule $1$ does not shift any share into or away from $\bc$ and $s(\bc)=s'(\bc)$. Moreover, also Rule $2$ cannot shift share to $\bc$ and hence, $s''(\bc)\leq s(\bc)$.

Observe that now each neighbour of $\bc$ is covered by at least two codewords. Moreover, if three or more of them are covered by at least three codewords, then $s''(\bc)\leq s(\bc)\leq 1+\frac{5}{2}+\frac{3}{3}=\frac{9}{2}$. Hence, we may assume that at most two of them are covered by three or more codewords. Consider now any non-corner neighbour $\bu$ of $\bc$. We may assume that $I(\bu)=\{\bc,\bc'\}$. However, $\bc'$ is also adjacent to at least one of the corner neighbours of $\bc$ adjacent to $\bu$. Thus, to separate that corner neighbour from $\bu$, it is covered by at least three codewords. Hence, there has to be at least two opposite corner neighbours of $\bc$ which are covered by three codewords. Assume that $\bc+(1,-1)$ and $\bc+(-1,1)$ are covered by three codewords.

\medskip
Assume first that $\{\bc+(2,-1),\bc+(1,-2)\}\not\subseteq C$ and without loss of generality that $\bc+(2,-1)\in C$. 
Observe that since $\bc+(0,-1)$ is covered by exactly two codewords, one of those codewords is either $\bc+(0,-2)$ or $\bc+(-1,-2)$. However, vertex $\bc+(-1,-1)$ is adjacent to both of those codewords and hence, $|I(\bc+(-1,-1))|\geq3$, a contradiction. Therefore,  $\{\bc+(2,-1),\bc+(1,-2)\}\subseteq C$. 
In this case, we have $s(\bc)\leq 1+\frac{6}{2}+\frac{2}{3}=4\frac{2}{3}$. Here we have the inequality since it is possible that one of the two corner neighbours is actually covered by four codewords (in that case $s(\bc)=4\frac{7}{12}$). Furthermore, since Rule $1$ does not affect the codeword $\bc$, we have $s'(\bc)\leq 4\frac{2}{3}$.

Observe that both codewords $\bc+(2,-1)$ and $\bc+(1,-2)$ cannot be covered by only two codewords since they are separated by code $C$.  Moreover, similar considerations can also be applied to codewords $\bc+(-2,1)$ and $\bc+(-1,2)$. However, as they are in a symmetric position compared to $\bc+(2,-1)$ and $\bc+(1,-2)$, we do not mention them in the following arguments.  Let us denote     $\bc'=\bc+(2,-1)$ and assume, without loss of generality, that $I(\bc')\not\subseteq I(\bc+(1,-2))$. Thus, at least one of vertices $\bc+(3,0)$, $\bc+(3,-1)$ and $\bc+(3,-2)$ is a codeword. Let us now divide the proof into three cases based on which one of these three vertices is a codeword.

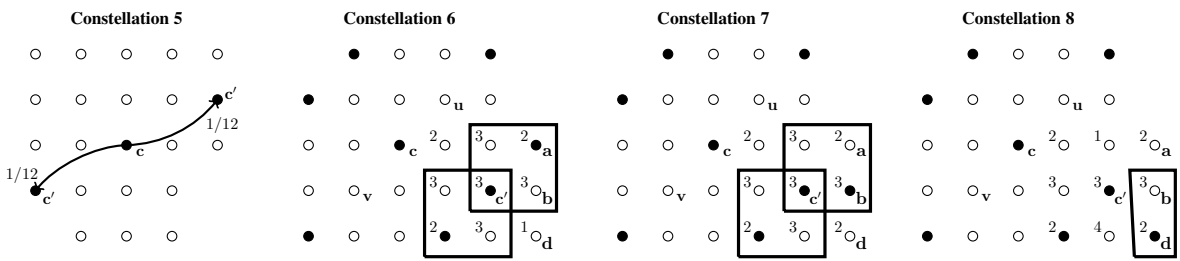
\begin{figure}[!h]
\vspace*{3mm}
    \centering
    \begin{tikzpicture}[scale=0.6]
                \begin{scope}[xshift=-9cm]

        \node[scale=0.6] at (1,4.8) {\textbf{Constellation 5}};

        \draw[] (0,0) circle(3pt);
        \draw[] (1,0) circle(3pt);
        \draw[] (2,0) circle(3pt);

        \draw[] (0,1) circle(3pt);
        \draw[] (1,1) circle(3pt); 
        \draw[] (2,1) circle(3pt);

        \draw[] (-1,2) circle(3pt);
        \draw[] (0,2) circle(3pt);
        \draw[fill=black] (1,2) circle(3pt); \node[scale=0.6] at (1.3,1.85) {\small $\bc$};
        \draw[] (2,2) circle(3pt);
        \draw[] (3,2) circle(3pt);

        \draw[] (-1,3) circle(3pt);
        \draw[] (0,3) circle(3pt);
        \draw[] (1,3) circle(3pt);
        \draw[] (2,3) circle(3pt);
        \draw[fill=black] (3,3) circle(3pt); \node[scale=0.6] at (3.3,3.15) {\small $\bc'$};
        \draw[fill=black] (-1,1) circle(3pt); \node[scale=0.6] at (-0.7,0.85) {\small $\bc'$};

        \draw[] (-1,4) circle(3pt);
        \draw[] (0,4) circle(3pt);
        \draw[] (1,4) circle(3pt);
        \draw[] (2,4) circle(3pt);
        \draw[] (3,4) circle(3pt);

        \node[scale=0.6] at (-1.3,1.35) {\small $1/12$};
        \node[scale=0.6] at (3.1,2.45) {\small $1/12$};

\path[draw,thick]
     (1,2) edge [->,bend right=25, looseness=0.9] node {} (3,3)
     (1,2) edge [->,bend right=25, looseness=0.8] node {} (-1,1)

    ;


        \node[scale=0.6] at (7,4.8) {\textbf{Constellation 6}};

        \draw[fill=black] (5,0) circle(3pt);
        \draw[] (6,0) circle(3pt);
        \draw[] (7,0) circle(3pt);
        \draw[fill=black] (8,0) circle(3pt);

        \draw[] (5,1) circle(3pt);
        \draw[] (6,1) circle(3pt);
        \node[scale=0.6] at (6.3,0.85) {\small $\bv$};
        \draw[] (7,1) circle(3pt);
        \draw[] (8,1) circle(3pt);
        \draw[fill=black] (9,1) circle(3pt);
        \node[scale=0.6] at (9.25,0.85) {\small $\bc'$};
        \draw[] (9,0) circle(3pt);

        \draw[] (5,2) circle(3pt);
        \draw[] (6,2) circle(3pt);
        \draw[fill=black] (7,2) circle(3pt); 		\node[scale=0.6] at (7.3,1.85) {\small $\bc$};
        \draw[] (8,2) circle(3pt);
        \draw[] (9,2) circle(3pt);

        \draw[fill=black] (5,3) circle(3pt);
        \draw[] (6,3) circle(3pt);
        \draw[] (7,3) circle(3pt);
        \draw[] (8,3) circle(3pt);
        \node[scale=0.6] at (8.3,2.85) {\small $\bu$};
        \draw[] (9,3) circle(3pt);

        \draw[fill=black] (6,4) circle(3pt);
        \draw[] (7,4) circle(3pt);
        \draw[] (8,4) circle(3pt);
        \draw[fill=black] (9,4) circle(3pt);

       \draw[fill=black] (10,2) circle(3pt);
        \draw[] (10,1) circle(3pt);
        \draw[] (10,0) circle(3pt);
        \node[scale=0.6] at (10.25,1.85)  {$\ba$};
        \node[scale=0.6] at (10.25,0.85)  {$\bb$};
        \node[scale=0.6] at (10.25,-0.15)  {$\bd$};

        \draw[very thick] (8.55,0.55) -- (10.45,0.55) -- (10.45,2.45) -- (8.55,2.45) -- (8.55,0.55);
        \node[scale=0.6] at (8.75,1.2) {\small $3$};
        \node[scale=0.6] at (9.75,2.2) {\small $2$};
        \node[scale=0.6] at (8.75,2.2) {\small $3$};
        \node[scale=0.6] at (9.75,1.2) {\small $3$};

        \begin{scope}[xshift=-1cm,yshift=-1cm]
        \draw[very thick] (8.55,0.55) -- (10.45,0.55) -- (10.45,2.45) -- (8.55,2.45) -- (8.55,0.55);
        \node[scale=0.6] at (8.75,1.2) {\small $2$};
        \node[scale=0.6] at (8.75,2.2) {\small $3$};
        \node[scale=0.6] at (8.75,3.2) {\small $2$};
        \node[scale=0.6] at (9.75,1.2) {\small $3$};
        \node[scale=0.6] at (10.75,1.2) {\small $1$};
        \end{scope}
        \end{scope}
        \begin{scope}[xshift=-2.1cm,yshift=0cm]
        \node[scale=0.6] at (7,4.8) {\textbf{Constellation 7}};
        \draw[fill=black] (5,0) circle(3pt);
        \draw[] (6,0) circle(3pt);
        \draw[] (7,0) circle(3pt);
        \draw[fill=black] (8,0) circle(3pt);

        \draw[] (5,1) circle(3pt);
        \draw[] (6,1) circle(3pt);
        \node[scale=0.6] at (6.3,0.85) {\small $\bv$};
        \draw[] (7,1) circle(3pt);
        \draw[] (8,1) circle(3pt);
        \draw[fill=black] (9,1) circle(3pt);
        \node[scale=0.6] at (9.25,0.85) {\small $\bc'$};
        \draw[] (9,0) circle(3pt);

        \draw[] (5,2) circle(3pt);
        \draw[] (6,2) circle(3pt);
        \draw[fill=black] (7,2) circle(3pt); 		\node[scale=0.6] at (7.3,1.85) {\small $\bc$};
        \draw[] (8,2) circle(3pt);
        \draw[] (9,2) circle(3pt);

        \draw[fill=black] (5,3) circle(3pt);
        \draw[] (6,3) circle(3pt);
        \draw[] (7,3) circle(3pt);
        \draw[] (8,3) circle(3pt);
        \node[scale=0.6] at (8.3,2.85) {\small $\bu$};
        \draw[] (9,3) circle(3pt);

        \draw[fill=black] (6,4) circle(3pt);
        \draw[] (7,4) circle(3pt);
        \draw[] (8,4) circle(3pt);
        \draw[fill=black] (9,4) circle(3pt);

       \draw[] (10,2) circle(3pt);
        \draw[fill=black] (10,1) circle(3pt);
        \draw[] (10,0) circle(3pt);
        \node[scale=0.6] at (10.25,1.85)  {$\ba$};
        \node[scale=0.6] at (10.25,0.85)  {$\bb$};
        \node[scale=0.6] at (10.25,-0.15)  {$\bd$};

                \draw[very thick] (8.55,0.55) -- (10.45,0.55) -- (10.45,2.45) -- (8.55,2.45) -- (8.55,0.55);
        \node[scale=0.6] at (8.75,1.2) {\small $3$};
        \node[scale=0.6] at (9.75,2.2) {\small $2$};
        \node[scale=0.6] at (8.75,2.2) {\small $3$};
        \node[scale=0.6] at (9.75,1.2) {\small $3$};

        \begin{scope}[xshift=-1cm,yshift=-1cm]
        \draw[very thick] (8.55,0.55) -- (10.45,0.55) -- (10.45,2.45) -- (8.55,2.45) -- (8.55,0.55);
        \node[scale=0.6] at (8.75,1.2) {\small $2$};
        \node[scale=0.6] at (8.75,2.2) {\small $3$};
        \node[scale=0.6] at (8.75,3.2) {\small $2$};
        \node[scale=0.6] at (9.75,1.2) {\small $3$};
        \node[scale=0.6] at (10.75,1.2) {\small $2$};
        \end{scope}
     \end{scope}
        \begin{scope}[xshift=4.6cm,yshift=0cm]
        \node[scale=0.6] at (7,4.8) {\textbf{Constellation 8}};

        \draw[fill=black] (5,0) circle(3pt);
        \draw[] (6,0) circle(3pt);
        \draw[] (7,0) circle(3pt);
        \draw[fill=black] (8,0) circle(3pt);

        \draw[] (5,1) circle(3pt);
        \draw[] (6,1) circle(3pt);
        \node[scale=0.6] at (6.3,0.85) {\small $\bv$};
        \draw[] (7,1) circle(3pt);
        \draw[] (8,1) circle(3pt);
        \draw[fill=black] (9,1) circle(3pt);
        \node[scale=0.6] at (9.25,0.85) {\small $\bc'$};
        \draw[] (9,0) circle(3pt);

        \draw[] (5,2) circle(3pt);
        \draw[] (6,2) circle(3pt);
        \draw[fill=black] (7,2) circle(3pt); 		\node[scale=0.6] at (7.3,1.85) {\small $\bc$};
        \draw[] (8,2) circle(3pt);
        \draw[] (9,2) circle(3pt);

        \draw[fill=black] (5,3) circle(3pt);
        \draw[] (6,3) circle(3pt);
        \draw[] (7,3) circle(3pt);
        \draw[] (8,3) circle(3pt);
        \node[scale=0.6] at (8.3,2.85) {\small $\bu$};
        \draw[] (9,3) circle(3pt);

        \draw[fill=black] (6,4) circle(3pt);
        \draw[] (7,4) circle(3pt);
        \draw[] (8,4) circle(3pt);
        \draw[fill=black] (9,4) circle(3pt);

       \draw[] (10,2) circle(3pt);
        \draw[] (10,1) circle(3pt);
        \draw[fill=black] (10,0) circle(3pt);
        \node[scale=0.6] at (10.25,1.85)  {$\ba$};
        \node[scale=0.6] at (10.25,0.85)  {$\bb$};
        \node[scale=0.6] at (10.25,-0.15)  {$\bd$};
                \node[scale=0.6] at (8.75,1.2) {\small $3$};
        \node[scale=0.6] at (9.75,2.2) {\small $2$};
        \node[scale=0.6] at (8.75,2.2) {\small $1$};
        \node[scale=0.6] at (9.75,1.2) {\small $3$};

        \begin{scope}[xshift=-1cm,yshift=-1cm]
        \draw[very thick] (10.55,0.55) -- (11.45,0.55) -- (11.45,2.45) -- (10.45,2.45) -- (10.55,0.55);
        \node[scale=0.6] at (8.75,1.2) {\small $2$};
        \node[scale=0.6] at (8.75,2.2) {\small $3$};
        \node[scale=0.6] at (8.75,3.2) {\small $2$};
        \node[scale=0.6] at (9.75,1.2) {\small $4$};
        \node[scale=0.6] at (10.75,1.2) {\small $2$};
        \end{scope}

     \end{scope}
    \end{tikzpicture}
    \caption{Constellations in the king grid. The edges have been omitted for simplicity. Numbers adjacent to vertices denote the minimum number of codewords that might cover them. Numbers within boxes can sometimes be switched with each other within the same box.}
    \label{constellations2}\vspace*{-1mm}
\end{figure}

\textbf{Case 1. $\ba=\bc+(3,0)\in C$}: This case is illustrated in Figure \ref{constellations2} Constellation 6. Now, at least three of the four vertices of $\bc'+\{(0,0),$ $(1,0),$ $(1,1),$ $(0,1)\}$ and of $\bc'+\{(0,0),(-1,0),(-1,-1),(0,-1)\}$ are covered by at least three codewords while the fourth is covered by at least two codewords. Furthermore, $\bc'+(-1,1)$ is covered by exactly two codewords. Thus, $$s(\bc')\leq 1+\frac{3}{2}+\frac{5}{3}=4\frac{1}{6}.$$ 

\textbf{Case 2. $\bb=\bc+(3,-1)\in C$}: This case is illustrated in Figure \ref{constellations2} Constellation 7. As in the previous case,  at least three of the four vertices of $\bc'+\{(0,0),(1,0),(1,1),(0,1)\}$ and of $\bc'+\{(0,0),(-1,0),(-1,-1),(0,-1)\}$ are covered by at least three codewords while the fourth is covered by at least two codewords. Furthermore, $\bc'+(-1,1)$ and $\bc'+(1,-1)$ are covered by at least two codewords. Thus, $$s(\bc')\leq \frac{4}{2}+\frac{5}{3}=3\frac{2}{3}.$$ 

\textbf{Case 3. $\bbd=\bc+(3,-2)\in C$}: This case is illustrated in Figure \ref{constellations2} Constellation 8. Now, $\bc'+(-1,0)$ is covered by three codewords, $\bc'+(-1,1)$ is covered by two codewords, $\bc'+(-1,-1)$ by at least two codewords, both $\bc'$ and $\bc'+(0,-1)$ are covered by at least three codewords and one of them by at least four codewords. Furthermore, both $\bbd$ and $\bb$ are covered by at least two codewords and at least one by at least three codewords. Finally, at least one of $\ba$ and $\bc'+(0,1)$ is covered by at least two codewords. Thus, $$s(\bc')\leq 1+\frac{4}{2}+\frac{3}{3}+\frac{1}{4}=4\frac{1}{4}.$$

Hence, in all three cases $s(\bc')\leq 4\frac{1}{4}=4\frac{3}{12}$. Moreover, if Rule $1$ shifts share to $\bc'$, then $s'(\bc')\leq 4$ as we have seen in Lemma \ref{LemmaKingRule1}. Hence,  $s'(\bc')\leq 4\frac{1}{4}$.

\medskip
Furthermore, there are at most three codewords at Euclidean distance $\sqrt{5}$ from $\bc'$ which are covered only by themselves. Indeed, we have $\bc$ and other possibilities are at points $\bc+(4,0), \bc+(4,-2)$ and $\bc+(3,-3)$. However, at most three of these vertices can be in $C$ and be covered only by themselves, simultaneously. Hence, $$s''(\bc')\leq 4\frac{3}{12}+\frac{3}{12}=\frac{9}{2}.$$

Therefore $s''(\bc)\leq \frac{9}{2}$ for each $\bc\in C$ since $4\frac{2}{3}-\frac{2}{12}=\frac{9}{2}$ and the claim follows with Lemma~\ref{Lemma shareshift}.
\end{proof}

\section{Conclusions}\label{sec:conclusions}

We introduced two new classes of covering codes for every positive integer $r$ -- the local $r$-identifying and local $r$-locating-dominating codes -- and studied them in binary hypercubes and infinite grids for $r=1$.
We studied the sizes of optimal local identifying codes in binary hypercubes and gave a general lower bound and a general upper bound that are asymptotically close.
Also, for some small binary hypercubes precise values for the optimal codes were found.
We studied the densities of optimal local identifying and local locating-dominating codes in (infinite) square, hexagonal, triangular and king grids.
In all except one of the cases we obtained optimal constructions.

For future research, we suggest studying the introduced new codes in binary hypercubes and in infinite grids for $r > 1$ and in different graphs.
This is especially interesting since the relationship between local 1-identification and local $r$-identification differs from the relationship between 1-identification and $r$-identification.
The same applies to location-domination.
Also, one could try improving the bounds of this paper in the cases where the size or the density of an optimal code was not settled.

\end{document}